\documentclass{lmcs}
\pdfoutput=1

\usepackage{lastpage}
\lmcsdoi{15}{3}{32}
\lmcsheading{}{\pageref{LastPage}}{}{}%
{Aug.~30,~2018}{Sep.~20,~2019}{}


\usepackage{amsmath}
\usepackage{amsfonts}
\usepackage{wasysym}
\usepackage{booktabs}
\usepackage{color}
\usepackage{changebar}
\usepackage{hyperref}
\usepackage{graphicx,accents}
\usepackage{tikz}
\usetikzlibrary{automata,positioning}

\allowdisplaybreaks%

\newcommand{\AP}{\mathrm{AP}}

\newcommand{\Ag}{\mathrm{Ag}}
\newcommand{\Ac}{\mathrm{Ac}}
\newcommand{\Dir}{\mathrm{Dir}}
\newcommand{\States}{\mathrm{St}}
\newcommand{\valf}{\lambda}
\newcommand{\evalf}{\valf}
\newcommand{\transf}{\delta}
\newcommand{\set}[1]{\{#1\}}

\newcommand{\direction}{d}

\newcommand{\cgspruns}{\mathit{runs}}

\newcommand{\computation}{\kappa}
\newcommand{\fincomputation}{\kappa}
\newcommand{\computations}{\mathit{comps}^{\omega}}
\newcommand{\fincomputations}{\mathit{comps}}

\newcommand{\VarSet}{\mathrm{Var}}

\newcommand{\paths}{\cgspruns}
\newcommand{\finruns}{\paths}

\newcommand{\asgFun}{\chi}
\newcommand{\AsgSet}{\mathrm{Asg}}

\newcommand{\labFun}{\lambda}
\newcommand{\AgSet}{\Ag}
\newcommand{\APSet}{\AP}
\newcommand{\EExs}[1]{\langle\langle#1\rangle\rangle}
\newcommand{\AAll}[1]{[[ #1 ]]}

\def\TEMPORAL#1{\mbox{\small$\mathbf{#1}$}}
\def\ltlnext{\TEMPORAL{X}}
\def\sometime{\TEMPORAL{F}}
\def\always{\TEMPORAL{G}}
\def\until{\,\TEMPORAL{U}\,}

\newcommand{\X}{\ltlnext}
\newcommand{\U}{\until}

\newcommand{\SetN}{\mathbb{N}}
\newcommand{\strFun}{f}
\newcommand{\StrSet}{\mathrm{Str}}
\newcommand{\xtransarrow}[1]{\mathrel{\raisebox{-0.2em}{$\xrightarrow{#1}$}}}

\newcommand{\transarrow}[1]{\overset{#1}{\longrightarrow}}
\newcommand{\Bisim}{R}
\newcommand{\Bisimilar}{\sim}
\newcommand{\bisim}{\Bisimilar}
\newcommand{\pairwisebisim}{\mathrel{\dot{\bisim}}}
\newcommand{\statewisebisim}{\mathrel{\dot{\bisim}}}
\newcommand{\run}{\rho}

\newcommand{\runs}{\mathit{runs}^\omega}
\newcommand{\traces}{\mathit{traces}^\omega}
\newcommand{\fintraces}{\mathit{traces}}
\newcommand{\trace}{\tau}
\newcommand{\fintrace}{\trace}
\newcommand{\finrun}{\pi}

\newcommand{\actions}{\Ac}
\newcommand{\action}{a}

\newcommand{\states}{\States}
\newcommand{\state}{s}
\newcommand{\midd}{\mathrel:}
\newcommand{\strategy}{f}

\newcommand{\stratprof}{\strategy}
\newcommand{\goalset}{\Gamma}

\newcommand{\agents}{\mathrm{Ag}}

\newlength{\wordlength}
\newcommand{\wordbox}[3][c]{\settowidth{\wordlength}{#3}\text{\makebox[\wordlength][#1]{#2}}}
\newcommand{\mathwordbox}[3][c]{\settowidth{\wordlength}{$#3$}\text{\makebox[\wordlength][#1]{$#2$}}}

\newcommand{\dcheck}[1]{\check{\check{#1}}}

\newcommand{\bphtext}[1]{\textcolor{black}{#1}}
\newcommand{\bphnotenew}[1]{{\color{cyan!70!black}{{\bf BPH}: #1}}}
\renewcommand{\bphnotenew}[1]{}
\newcommand{\bphnote}[1]{\bphnotenew{#1}}

\newenvironment{bph}{\color{cyan!70!black}}{\color{black}}

\newcommand{\mwbactions}[3]{\mathwordbox[l]{#2}{x},\mathwordbox[l]{#3}{x},\mathwordbox[l]{#1}{x}}

\newcommand{\mwbactionstwo}[3]{\mathwordbox[l]{#1}{x},\mathwordbox[l]{#2}{x}}

\newcommand{\aonetwooo}{a}
\newcommand{\bonetwooo}{b}
\newcommand{\aonethree}{a'}
\newcommand{\bonethree}{b'}

\newcommand{\atwo}{a}
\newcommand{\btwo}{b}

\newcommand{\athree}{a}
\newcommand{\bthree}{b}

\renewcommand{\phi}{\varphi}

\begin{document}

\title[Nash Equilibrium and Bisimulation Invariance]{Nash Equilibrium and Bisimulation Invariance}

\author[J.~Gutierrez]{Julian Gutierrez\rsuper{a}}

\author[P.~Harrenstein]{Paul Harrenstein\rsuper{a}}
\address{\lsuper{a}Department of Computer Science, University of Oxford}
\email{\{julian.gutierrez,paul.harrenstein,mjw\}@cs.ox.ac.uk}

\author[G.~Perelli]{\texorpdfstring{\\}{}Giuseppe Perelli\rsuper{b}}
\address{\lsuper{b}Department of Computer Science, University of G\"{o}teborg}
\email{perelli@chalmers.se}

\author[M.~Wooldridge]{Michael Wooldridge\rsuper{a}}

\begin{abstract}
  Game theory provides a well-established framework for the analysis
  of concurrent and multi-agent systems. The basic idea is that
  concurrent processes (agents) can be understood as corresponding to
  players in a game; plays represent the possible computation runs of
  the system; and strategies define the behaviour of
  agents. Typically, strategies are modelled as functions from
  sequences of system states to player actions. Analysing a system in
  such a setting involves computing the set of (Nash) equilibria in
  the concurrent game. However, we show that, with respect to the
  above model of strategies (arguably, the ``standard'' model in the
  computer science literature), {\em bisimilarity does not preserve the existence of
    Nash equilibria}. Thus, two concurrent games which are
  behaviourally equivalent from a semantic perspective, and which from
  a logical perspective satisfy the same temporal logic formulae,
  may  nevertheless have fundamentally different properties (solutions)
  from a game theoretic perspective. Our aim in this paper is to
  explore the issues raised by this discovery. After illustrating the
  issue by way of a motivating example, we present three models of
  strategies with respect to which the existence of Nash equilibria is
  preserved under bisimilarity. We use some of these models of
  strategies to provide new semantic foundations for logics for
  strategic reasoning, and investigate restricted scenarios where
  bisimilarity can be shown to preserve the existence of Nash
  equilibria with respect to the conventional model of strategies in
  the computer science literature.
\end{abstract}

\maketitle

\section{Introduction}\label{secn:intro}
The concept of \emph{bisimilarity} plays a central role in both the
theory of concurrency~\cite{milner:89a,HM85} and
logic~\cite{vanbenthem:76a,HM85}.  In the context of concurrency, bisimilar
systems are regarded as \emph{behaviourally equivalent}---appearing to
have the same behaviour when interacting with an arbitrary
environment. From a logical/verification perspective, bisimilar
systems are known to satisfy the \emph{same temporal logic properties}
with respect to languages such as LTL, CTL, or the
$\mu$-calculus~\cite{Pnueli77,ClarkeE81,Kozen83}.
These features, in turn, make it possible to verify
temporal logic properties of concurrent systems using
bisimulation-based approaches~\cite{SR12}. For example, temporal logic model
checking techniques~\cite{CGP02} may be optimised by applying them to the smallest
bisimulation equivalent model of the system being analysed; or,
indeed, to any other model within the system's bisimulation
equivalence class. This is possible because the properties that one is
interested in checking are {\em bisimulation invariant}.

Model checking~\cite{CGP02} is not the only verification technique
that can benefit from bisimulation invariance: consider abstraction
and refinement techniques~\cite{CGL94,CC02} (where a set of states is
either collapsed or broken down in order to build a somewhat simpler
set of states); coinduction methods~\cite{S09} (which can be used to
check the correctness of an implementation with respect to a given
specification); or reduced BDD representations of a system~\cite{B92}
(where isomorphic, and therefore bisimilar, subgraphs are merged,
thereby eliminating part of the initial state space of the
system). Bisimulation invariance is therefore a powerful and fundamental concept
in the formal analysis and verification of concurrent and multi-agent
systems, which plays an important role in many verification tools.

Game theory provides another important framework for the analysis and
verification of concurrent and multi-agent systems. Within this
framework, a concurrent/multi-agent system is viewed as a game, where
processes/agents correspond to players, system executions (that is,
computation runs) to plays, and individual process behaviours are
modelled as player strategies, which are used to resolve the possible
nondeterministic choices available to each player.
A widely-used model for strategies in concurrent games is to view a
strategy for a process/agent/player~$i$ as a function~$f_i$ which maps finite
histories $s_0,s_1,\ldots ,s_k$ of system states to actions
$f_i(s_0,s_1,\ldots,s_k)$ available to~$i$ at
state~$s_k$. (In what follows, we use the terms process, agent, and
player interchangeably.) We refer to this as the ``conventional''
model of strategies, as it is the best-known and most widely-used
model in logic, AI, and computer science (and indeed in extensive form
games~\cite{OR94}).  For instance, specification languages such as
Alternating-time Temporal Logic (ATL~\cite{AHK02}), and formal models
such as concurrent game structures~\cite{AHK02} use this model of
strategies. If we model a concurrent/multi-agent system as a game in
this way, then the analysis and verification of the system reduces to
computing the set of (Nash) equilibria in the
associated 
game; in some cases, the analysis reduces to the computation of a
winning strategy in the game, that is, a strategy that ensures that
the players who follow such a plan will achieve their goal no matter
how the other players in the system play, {\em i.e.}, against any
other possible counter-strategy.

Now, because bisimilar 
systems are regarded as 
behaviourally equivalent, and bisimilar systems satisfy the same set
of temporal logic properties, it is natural to ask whether the Nash
equilibria of bisimilar structures can be identified in a similar way; that is, we ask the
following question:
\begin{center}
\emph{Is Nash equilibrium invariant under bisimilarity?}
\end{center}
We show that, for the ``conventional'' model of strategies, the answer
to this question is, in general, no. More specifically, the answer
critically depends on precisely how players' strategies are
modelled. With the conventional model of strategies, we find the answer is
positive only for some two-player games, but negative in general
for games with more than two players. This means, for instance, that,
in the general case, bisimulation-based techniques cannot be used when
one is also reasoning about the Nash equilibria of concurrent systems
that are formally modelled as multi-player (concurrent) games.

For instance, given a concurrent and reactive system, represented as a
collection of individual system components, say $P_1,\ldots,P_n$, one
may want to know if a given temporal logic property, say $\phi$, is
satisfied by these system components whenever they choose to use
strategies that form an equilibrium, that is, we want to know whether
for some/every computation run~$\rho\in \mathit{NE}(P_1,\ldots,P_n)$
we have $\rho\models\phi$, where $\mathit{NE}(P_1,\ldots,P_n)$ denotes
the set of all computation runs that may be generated as a result of
$P_1,\ldots,P_n$ selecting strategies that form a Nash equilibrium.
Because we are interested in concurrent systems, and bisimilarity is
one of the most important behavioural equivalences in
concurrency~\cite{Milner80,HM85,NicolaV95,GlabbeekW96}, it is
highly desirable  that properties which hold in equilibrium are
sustained across all systems that are bisimilar to $P_1,\ldots,P_n$, meaning
that for every (temporal logic) property $\phi$ and every process
$P'_i$, if $P'_i$ is bisimilar to $P_i\in\{P_1,\ldots,P_n\}$, then
$\phi$ is satisfied in equilibrium by $P_1,\ldots P_i \ldots P_n$ if
and only if is also satisfied in equilibrium by
$P_1,\ldots P'_i \ldots,P_n$, the system in which $P_i$ is replaced by
$P'_i$, that is, across all bisimilar systems
to~$P_1,\ldots,P_n$. This property, called {\em invariance under
  bisimilarity}, has been widely used for decades for the semantic
analysis ({\em e.g.}, for modular and compositional reasoning) and
formal verification ({\em e.g.}, for temporal logic model checking) of
concurrent systems. Unfortunately, as shown here, and already
discussed in~\cite{GHW15-concur}, the satisfaction of temporal logic
properties in equilibrium is not invariant under bisimilarity, thus
posing a verification challenge for the modular and compositional
reasoning of concurrent systems, since individual system components in
a concurrent system cannot be replaced by (behaviourally equivalent)
bisimilar ones, while preserving the temporal logic properties that
the overall system satisfies in equilibrium.  This is also
a problem from a synthesis point of view.  Indeed, a strategy for a
system component~$P_i$ may not be a valid strategy for a bisimilar
system component~$P'_i$.  As a consequence, the problem of building
strategies for individual processes in the concurrent
system~$P_1,\ldots P_i \ldots P_n$ may not, in general, be the same as
building strategies for a bisimilar
system~$P_1,\ldots P'_i \ldots P_n$, again, dashing any hope of
modular reasoning on concurrent systems.

Motivated by these observations---which bring together in a striking
way a fundamental concept in game theory and a fundamental concept in
logic/concurrency---the purpose of the present paper is to investigate
these issues in detail. We first present a motivating example, to
illustrate the basic point that using the conventional model of
strategies, bisimulation need not preserve Nash equilibria. We then
present three alternative models of strategies in which Nash
equilibria and their existence are preserved under bisimilarity.  We
also study the above question for different classes of systems, for
instance deterministic and nondeterministic ones, and explore
applications to logic. Specifically, we investigate the implications
of replacing the conventional model of strategies with some of the
models we propose in this paper in logics for strategic
reasoning~\cite{MMPV14,CHP10}, in particular, the semantic
implications with respect to Strategy Logic (SL~\cite{MMPV14}).
We also show that, within the conventional model of strategies, Nash
equilibrium is preserved by bisimilarity in certain two-player games
as well as in the class of concurrent game structures that are induced
by iterated Boolean games~\cite{GHW15}, a framework that can be used
to reason about the strategic behaviour of AI, autonomous, and
multi-agent systems~\cite{WGHMPT16}.
Our main invariance results are summarised in Table~\ref{tab:summary}.

\subsection{A Motivating Example}\label{section:motivating_example}
So far we have mentioned some cases where one needs or desires a
property to be invariant under bisimilarity. However, one may still
wonder why it is so important that the particular property of having a Nash equilibrium
is preserved under bisimilarity. One reason has its roots in
automated formal verification. To illustrate this, imagine that the
system of Figure~\ref{fig:cgswithne} is given as input to a
verification tool. It is likely that such a tool will try to perform
as many optimisations as possible to the system before any analysis is
performed. Perhaps the simplest of such optimisations---as is being done
by virtually every model checking tool---is to reduce the input system
by merging \emph{isomorphic} subtrees.
{This is done, \emph{e.g.}, when generating the ROBDD
  representation of a system.}  If such an optimisation is made,
the tool will construct the (bisimilar) system in
Figure~\ref{fig:cgswithoutne}.
(Observe that the subgraphs rooted at~$s_1$ and~$s'_1$ are isomorphic.)
However, with respect to the existence
of Nash equilibria, such a transformation is unsound in the
general case.

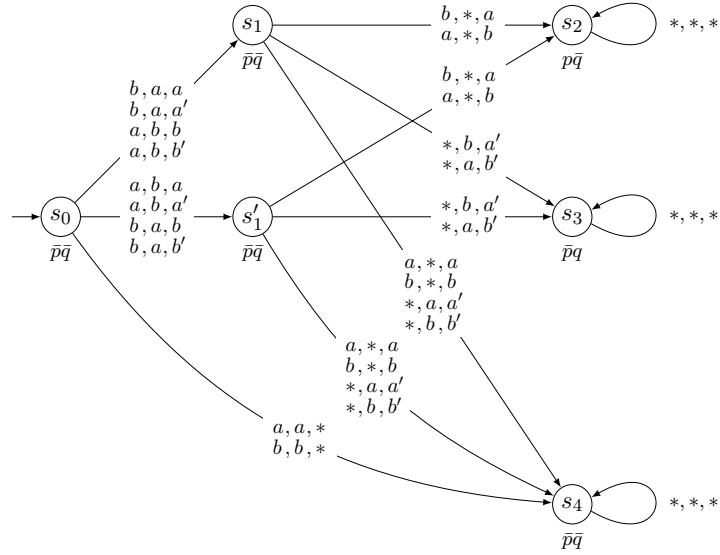
\begin{figure}
\centering
 	\scalebox{.85}{
 	  \begin{tikzpicture}[scale=1]

 	  \tikzstyle{every ellipse node}=[draw,inner xsep=3.5em,inner ysep=1.2em,fill=black!15!white,draw=black!15!white]
 	  \tikzstyle{every circle node}=[fill=white,draw,minimum size=1.6em,inner sep=0pt]



 	  \draw(0,0) node(0){}
 	  			++ 	( 0:.9)	  	node[label=-90:{\footnotesize$\bar p\bar q$},circle](v0){$s_0$}
 				++	(  0:3)   	node[label=-90:{\footnotesize$\bar p\bar q$},circle](v1){$s'_1$}
					+( 90:3)  	node[label=-90:{\footnotesize$\bar p\bar q$},circle](w1){$s_1$}
 					++( 0:5)  	node[label=-90:{\footnotesize$\bar p     q$},circle](v3){$s_3$}
 					+( 90:3)  	node[label=-90:{\footnotesize$     p\bar q$},circle](v2){$s_2$}
					+(-90:4.5)	node[label=-90:{\footnotesize$\bar p\bar q$},circle](v4){$s_4$}
 				;

 	\draw[-latex] (v0) --node[inner xsep=0pt, inner ysep=0pt, pos=.5,fill=white](){
						\scalebox{.85}{$\begin{array}{l}
							\mwbactions{\aonetwooo}{\atwo}{\bthree}\\[-.2em]
							\mwbactions{\aonethree}{\atwo}{\bthree}\\[-.2em]
							\mwbactions{\bonetwooo}{\btwo}{\athree}\\[-.2em]
							\mwbactions{\bonethree}{\btwo}{\athree}
						\end{array}$}} (v1);

	\draw[-latex] (v0) --node[inner xsep=0pt, inner ysep=0pt, pos=.5,fill=white](){
						\scalebox{.85}{$\begin{array}{l}
							\mwbactions{\aonetwooo}{\btwo}{\athree}\\[-.2em]
							\mwbactions{\aonethree}{\btwo}{\athree}\\[-.2em]
							\mwbactions{\bonetwooo}{\atwo}{\bthree}\\[-.2em]
							\mwbactions{\bonethree}{\atwo}{\bthree}
						\end{array}$}} (w1);

 	\draw[-latex] (v0) to[bend right=25]node[inner xsep=0pt, inner ysep=0pt, pos=.545,fill=white](){
						\scalebox{.85}{$\begin{array}{l}
							\mwbactions{\ast}{\atwo}{\athree}\\[-.2em]
							\mwbactions{\ast}{\btwo}{\bthree}
						\end{array}$}} (v4);
 	\draw[-latex] (v1) --node[inner xsep=0pt, inner ysep=0pt, pos=.7,fill=white](){\scalebox{.85}{$\begin{array}{l}\mwbactions{\aonetwooo}{\btwo}{\ast}\\[-.2em]\mwbactions{\bonetwooo}{\atwo}{\ast}\end{array}$}} (v2);
 	\draw[-latex] (v1) --node[inner xsep=0pt, inner ysep=0pt, pos=.7,fill=white](){\scalebox{.85}{$\begin{array}{l}\mwbactions{\aonethree}{\ast}{\bthree}\\[-.2em]\mwbactions{\bonethree}{\ast}{\athree}\end{array}$}} (v3);
 	\draw[-latex] (w1) --node[inner xsep=0pt, inner ysep=0pt, pos=.7,fill=white](){\scalebox{.85}{$\begin{array}{l}\mwbactions{\aonetwooo}{\btwo}{\ast}\\[-.2em]\mwbactions{\bonetwooo}{\atwo}{\ast}\end{array}$}} (v2);
 	\draw[-latex] (w1) --node[inner xsep=0pt, inner ysep=0pt, pos=.7,fill=white](){\scalebox{.85}{$\begin{array}{l}\mwbactions{\aonethree}{\ast}{\bthree}\\[-.2em]\mwbactions{\bonethree}{\ast}{\athree}\end{array}$}} (v3);
	\draw[-latex] (v1) to[bend right=17.5]node[inner xsep=0pt, inner ysep=0pt, pos=.45,fill=white,fill opacity=1](){\scalebox{.85}{$\begin{array}{l}\mwbactions{\aonetwooo}{\atwo}{\ast}\\[-.2em]\mwbactions{\bonetwooo}{\btwo}{\ast}\\[-.2em]\mwbactions{\aonethree}{\ast}{\athree}\\[-.2em]\mwbactions{\bonethree}{\ast}{\bthree}\end{array}$}} (v4);
	\draw[-latex] (w1) --node[inner xsep=0pt, inner ysep=0pt, pos=.565,fill=white,fill opacity=1](){\scalebox{.85}{$\begin{array}{l}\mwbactions{a}{\atwo}{\ast}\\[-.2em]\mwbactions{\bonetwooo}{\btwo}{\ast}\\[-.2em]\mwbactions{\aonethree}{\ast}{\athree}\\[-.2em]\mwbactions{\bonethree}{\ast}{\bthree}\end{array}$}} (v4);
 	\draw[-latex] (0) -- (v0);

 	\draw[-latex] (v2.70-90) .. controls +(60-90:4em) and +(120-90:4em) .. node[pos=.5,fill=white,right,xshift=.5ex](){\footnotesize$\ast,\ast,\ast$} (v2.110-90);

 	\draw[-latex] (v3.70-90) .. controls +(60-90:4em) and +(120-90:4em) .. node[pos=.5,fill=white,right,xshift=.5ex](){\footnotesize$\ast,\ast,\ast$} (v3.110-90);

 	\draw[-latex] (v4.70-90) .. controls +(60-90:4em) and +(120-90:4em) .. node[pos=.5,fill=white,right,xshift=.5ex](){\footnotesize$\ast,\ast,\ast$} (v4.110-90);

 	\end{tikzpicture}
 	}
 	\caption{The game~$G_0$ on concurrent game structure~$M_0$ with a Nash equilibrium.}%
 	\label{fig:cgswithne}
\end{figure}

For instance, suppose that the system in
Figure~\ref{fig:cgswithne} represents a 3-player game, where each
transition is labelled by the choices~$x,y,z$ made by player~$1$,~$2$,
and~$3$, respectively, and asterisk~$\ast$ being a wildcard for any
action for the player in the respective position. Thus, whereas
players~$1$ and~$2$ can choose to play either~$a$ or~$b$ at each
state, player~$3$ can choose between~$a$, $b$, $a'$, or~$b'$. The
states are labelled by valuations~$xy$ over~$\set{p,q}$,
where~$\bar x$ indicates that~$x$ is set to false. Assume that
player~$1$ would like~$p$ to be true sometime, that player~$2$
would like~$q$ to be true sometime, and that player~$3$ desires
to prevent both player~$1$ and player~$2$ from achieving their goals.
Accordingly, their preferences/goals can, respectively, be formally
represented by the LTL formulae
 \begin{align*}
 	\gamma_1	&	=	\sometime p,	&
 	\gamma_2	&	=	\sometime q,&
	\text{and}	&&
 	\gamma_3	&	=	\always \neg(p\vee q),
 \end{align*}
where, informally, $\sometime \phi$ means ``eventually $\phi$ holds'' and $\always \phi$ means ``always $\phi$ holds''.
Moreover,  given these players' goals and the conventional model of
strategies, we will see later in Section~\ref{subsection:run_based_no_preservation} that the system in
Figure~\ref{fig:cgswithne} has a Nash equilibrium, whereas  no Nash equilibria
exists in the
(bisimilar) concurrent system presented in Figure~\ref{fig:cgswithoutne}.

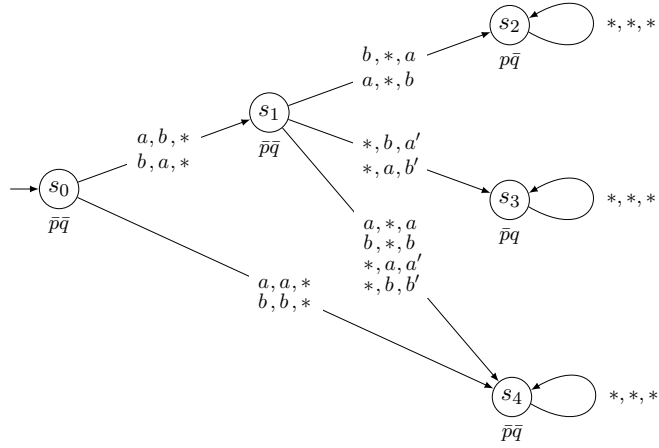
\begin{figure}
\centering
 	\scalebox{.85}{
 	  \begin{tikzpicture}[scale=1]

 	  \tikzstyle{every ellipse node}=[draw,inner xsep=3.5em,inner ysep=1.2em,fill=black!15!white,draw=black!15!white]
 	  \tikzstyle{every circle node}=[fill=white,draw,minimum size=1.6em,inner sep=0pt]


 	  \draw(0,0) node(0){}
 	  			++ 	( 0:.9)	  node[label=-90:{\footnotesize$\bar p\bar q$},circle](v0){$s_0$}
 				++	(  20:3.5)  node[label=-90:{\footnotesize$\bar p\bar q$},circle](v1){$s_1$}
 					+( 20:4)  node[label=-90:{\footnotesize$     p\bar q$},circle](v2){$s_2$}
 					+(-20:4)  node[label=-90:{\footnotesize$\bar p     q$},circle](v3){$s_3$}
 				;
 	 \draw(0,.9) ++(-27.5:9)	node[label=-90:{\footnotesize$\bar p\bar q$},circle](v4){$s_4$};

 	\draw[-latex] (v0) --node[inner xsep=0pt, inner ysep=0pt,pos=.5,fill=white](){\scalebox{.85}{$\begin{array}{l}\mwbactions{\ast}{\atwo}{\bthree}\\\mwbactions{\ast}{\btwo}{\athree}\end{array}$}} (v1);

 	\draw[-latex] (v1) --node[inner xsep=0pt, inner ysep=0pt,pos=.5,fill=white](){\scalebox{.85}{$\begin{array}{l}\mwbactions{\aonetwooo}{\btwo}{\ast}\\\mwbactions{\bonetwooo}{\atwo}{\ast}\end{array}$}} (v2);

 	\draw[-latex] (v1) --node[inner xsep=0pt, inner ysep=0pt,pos=.5,fill=white](){\scalebox{.85}{$\begin{array}{l}\mwbactions{\aonethree}{\ast}{\bthree}\\\mwbactions{\bonethree}{\ast}{\athree}\end{array}$}} (v3);

 	\draw[-latex] (v0) to[bend right=0]node[inner xsep=0pt, inner ysep=0pt,pos=.5,fill=white](){\scalebox{.85}{$\begin{array}{l}
							\mwbactions{\ast}{\atwo}{\athree}\\[-.2em]
							\mwbactions{\ast}{\btwo}{\bthree}
						\end{array}$}} (v4);

	\draw[-latex] (v1) to[bend right=0]node[inner xsep=0pt, inner ysep=0pt,pos=.5,fill=white](){\scalebox{1}{\scalebox{.85}{$\begin{array}{l}\mwbactions{\aonetwooo}{\atwo}{\ast}\\[-.2em]\mwbactions{\bonetwooo}{\btwo}{\ast}\\[-.2em]\mwbactions{\aonethree}{\ast}{\athree}\\[-.2em]\mwbactions{\bonethree}{\ast}{\bthree}\end{array}$}}} (v4);

 	\draw[-latex] (0) -- (v0);

 	\draw[-latex] (v2.70-90) .. controls +(60-90:4em) and +(120-90:4em) .. node[pos=.5,fill=white,right,xshift=.5ex](){\footnotesize$\ast,\ast,\ast$} (v2.110-90);

 	\draw[-latex] (v3.70-90) .. controls +(60-90:4em) and +(120-90:4em) .. node[pos=.5,fill=white,right,xshift=.5ex](){\footnotesize$\ast,\ast,\ast$} (v3.110-90);

 	\draw[-latex] (v4.70-90) .. controls +(60-90:4em) and +(120-90:4em) .. node[pos=.5,fill=white,right,xshift=.5ex](){\footnotesize$\ast,\ast,\ast$} (v4.110-90);

 	\end{tikzpicture}
 	}
 	\caption{The game~$G_1$ on concurrent game structure~$M_1$ without a Nash equilibrium.}%
 	\label{fig:cgswithoutne}
\end{figure}

This example illustrates a major issue when analysing (the existence of) Nash equilibria in the most widely used models of strategies and multi-player games in the computer science literature, namely, that even the simplest and most innocuous optimisations commonly used in automated verification are not necessarily sound with respect to game-theoretic analyses.

Because the problem is so fundamental, one may wonder whether bisimilarity is not the right behavioural equivalence for multi-player games, or whether Nash equilibrium is not the right solution concept for game-theoretic analyses of concurrent and multi-agent systems modelled as multi-player games. We will discuss these questions in more detail in Section~\ref{secn:conc}, as we do not have a definite answer, but for now we would like to make a couple of observations. On the one hand, that our results also hold both for ``alternating'' (bisimilarity) relations, as defined in~\cite{AHKV98}, which are intended to capture strategic behaviour in multi-player games, as well as for trace equivalence, as defined in CSP~\cite{BrookesHR84}, an equivalence much weaker than bisimilarity. On the other hand, that our negative results also hold for solution concepts stronger than Nash equilibrium, {\em e.g.}, for strong and subgame-perfect Nash equilibria, suggesting that the problem is not a particular defect of Nash equilibrium. Indeed, we think that the issue underlying the mismatch between bisimilarity and Nash equilibrium lies elsewhere. We will propose a very general solution to this problem, that is, a way to reconcile bisimilarity and Nash equilibrium, based on a new definition of strategy in a multi-player game. To do this, some concepts and definitions have to be introduced first.

\section{Preliminaries}
We begin by introducing the main notational conventions, models, and
technical concepts used in this 
paper.

\subsection*{Sets}
Given any set~$S=\set{s,q,r,\ldots}$, we use
$S^*$, $S^\omega$, and $S^+$ for, respectively, the sets of finite, infinite, and non-empty finite sequences of elements in~$S$.
If 
$w_1\in S^*$
and $w_2$ is any other (finite or
infinite) sequence, we write $w_1w_2$ for their concatenation. The
empty sequence is denoted by~$\epsilon$.

\subsection*{Concurrent Game Structures}
We use the model of concurrent game structures, which are well-established in the logic and computer science literatures~(see, for instance,~\cite{AHK02}).
A \emph{concurrent game structure~(CGS)} is a tuple $M=(\Ag,\AP,\Ac,\States,s^0_M,\valf,\transf)$, where
$\Ag=\set{1,\dots,n}$ is a set of \emph{players} or agents, $\AP$ a set of \emph{propositional variables},~$\Ac$ is a set of \emph{actions}, $\States$ is a set of \emph{states} containing a unique \emph{initial state}~$s^0_M$.
With each player~$i\in\Ag$ and each state~$s\in\states$, we associate a non-empty set~$\Ac_i(s)$ of \emph{feasible} actions that, intuitively,~$i$ can perform when in state~$s$.
By a \emph{direction} or \emph{decision} we understand a profile of actions~$\direction=(\action_1,\dots,\action_n)$ in $\actions\times\cdots\times\actions$ and we let~$\Dir$ denote the set of directions. A direction~$\direction=(\action_1,\dots,\action_n)$ is \emph{legal at state~$s$} if $a_i\in\Ac_i(s)$ for every player~$i$.
Unless stated otherwise, by ``direction'' we will henceforth generally mean ``legal direction''.
Furthermore, $\valf\colon\states\to 2^{\AP}$ is a \emph{labelling function}, associating with every state~$s$ a \emph{valuation}~$v\in 2^\AP$.
Finally,~$\delta$ is a \emph{deterministic} \emph{transition function}, which
associates with each state~$s$ and every legal direction~$\direction=(a_1,\dots,a_n)$ at~$s$ a state $\delta(s,a_1,\dots,a_n)$.
As such~$\delta$
characterises the behaviour of the system when $\direction=(\action_1,\dots,\action_n)$ is performed at state~$\state$.

\subsubsection*{Computations, Runs, and Traces}

The possible behaviours exhibited by a CGS can be described at at
least three different levels of abstraction. In what follows, we
distinguish between \emph{computations}, \emph{runs}, and
\emph{traces}.  Computations carry the most information, while traces
carry the least, in the sense that every computation induces a unique
run and every run induces a unique trace, but not necessarily the
other way round. \emph{The distinctions we make between computations,
  runs, and traces may appear to be insignificant, but are in fact central in
  our analysis of bisimilarity and Nash equilibrium.}

A state~$\state'$ is \emph{accessible} from another state~$\state$
whenever there is some~$\direction=(\action_1,\dots,\action_n)$ such
that~$\direction$ is legal at~$\state$ and
$\delta(s,\action_1,\dots,\action_n)=\state'$.  {For easy readability
  we then also write $s\transarrow{\direction} s'$.}  An
\emph{(infinite) computation} is then an infinite sequence of
directions $\computation=\direction_0,\direction_1,\direction_2,\dots$
such that there are states $s_0,s_1,\dots$ with~$s_0=s^0_M$ and
$
s_0\transarrow{\direction_0}s_1\transarrow{\direction_1}s_2\transarrow{\direction_2}\cdots\text.
$ Observe that, having assumed the transition function~$\transf$ to be
complete and deterministic, in every concurrent game model the
states~$s_0,s_1,\dots$ in the above definition always exist and are
unique.  A \emph{finite computation} is any finite prefix of a
computation~$\fincomputation$. We also allow a finite computation to
be the empty sequence~$\epsilon$ of directions.
The sets of infinite and finite computations are denoted by $\computations_M$ and $\fincomputations_M$, respectively.
We also use $\transf^{*}(s, \direction_{0}, \direction_{1}, \ldots \direction_{k})$ to denote the unique state that is reached from the state $s$ after applying the computation $\direction_{0}, \direction_{1}, \ldots \direction_{k}$.

An \emph{(infinite) run} is an infinite sequence $\run=\state_0,\state_1,\state_2\dots$ of states of sequentially accessible states, with $s_0=s^0_M$.
We say that run~$s_0,\dots,s_k$ is \emph{induced} by computation~$d_0,\dots,d_{k-1}$ if $	s_0\transarrow{\direction_0}s_1\transarrow{\direction_1}s_2\transarrow{\direction_2}\cdots$ and $s_0=s^0_M$.
 Thus, every computation induces a unique run and every run is induced by at least one computation.
By a \emph{finite run} or \emph{(finite) history} we mean a finite prefix of a run.
The sets of infinite and finite runs are denoted by $\runs_M$ and $\finruns_M$, respectively.

An \emph{(infinite) trace} is a sequence $\trace=v_0,v_1,v_2,\dots$ of
valuations such that there is a run~$\run=s_0,s_1,s_2,\dots$
in~$\runs_M$ such that $v_k=\lambda(s_k)$ for every $k\ge 0$, that is,
$\trace=\valf(s_0),\valf(s_1),\valf(s_2),\dots$.  In that case we
say that trace~$\trace$ is \emph{induced} by run~$\run$, and if~$\run$
is induced by computation~$\computation$, also that~$\trace$ is
induced by~$\computation$.
By a \emph{finite trace} we
mean a finite prefix of a trace.  We denote the sets of finite and
infinite traces of a concurrent game structure~$M$ by $\fintraces_M$
and $\traces_M$, respectively.
%

We use~$\run_M(\computation)$ to denote the run induced by a computation~$\computation$ in CGS~$M$, and write~$\finrun_M(\computation)$ if~$\computation$ is finite on the understanding that $\finrun_M(\epsilon)=s_M^0$.
Also, if $\run=s_0,s_1,s_2,\dots$ is a run, by $\trace_M(\run)$ we denote the trace $\valf(s_0),\valf(s_1),\valf(s_2),\dots$, and similarly for finite runs~$\finrun\in\finruns_M$. Finally, $\trace_M(\run_M(\computation))$ is abbreviated as~$\trace_M(\computation)$. When no confusion is likely, we omit the subscript~$M$ and the qualification `finite'.

\subsection*{Bisimilarity}

One of the most important behavioural/observational equivalences in concurrency is bisimilarity, which is usually defined over Kripke structures or labelled transition systems (see, {\em e.g.},~\cite{milner:89a,HM85}). However, the equivalence can be uniformly defined for general concurrent game structures, where decisions/directions play the role of, for instance, actions in transition systems.
Formally, let $M=(\AP,\Ag,\Ac,\States,s^0_{M},\valf,\transf)$ and $M'=(\AP,\Ag,\Ac,\States',s^0_{M'},\valf',\transf')$ be two concurrent game structures.
A \emph{bisimulation}, denoted by $\Bisimilar$,
between states $s^*\in\States$ and $t^*\in\States'$ is a non-empty binary
relation $\Bisim \subseteq \States \times \States'$, such that $s^* \mathrel{\Bisim} t^*$ and for all $s,s'\in\States$, $t,t'\in\States'$, and $d\in\Dir$: 
\begin{itemize}
	\item $s\mathrel{\Bisim} t$ implies $\valf(s)=\valf'(t)$,
	\item $s\mathrel{\Bisim} t$ and $s\transarrow{\direction}s'$ implies $t\transarrow{\direction}t''$ for some $t''\in\States'$ with $s'\mathrel{\Bisim}t''$,
	\item $s\mathrel{\Bisim} t$ and $t\transarrow{\direction}t'$ implies $s\transarrow{\direction}s''$ for some $s''\in\States$ with $s''\mathrel{\Bisim}t'$.
\end{itemize}
Then, if there is a bisimulation between two states~$s^*$ and~$t^*$, we say that they are \emph{bisimilar} and write $s^*\Bisimilar t^*$ in such a case. We also say that concurrent game structures~$M$ and~$M'$ are \emph{bisimilar} (in symbols $M\Bisimilar M'$) if $s^0_{M} \Bisimilar s^0_{M'}$.
Since the transition functions of concurrent game structures, as defined, are deterministic, we have
the following simple but useful facts.
We say that runs~$\run=s_0,s_1,\dots$ and $\run'=s'_0,s_1',\dots$  are \emph{statewise bisimilar} (in symbols $\run\pairwisebisim\run'$) if $s_k\sim s'_k$ for every $k\ge 0$.
%
Both bisimilarity and statewise bisimilarity are equivalence relations, which is a standard result in the literature (see, for instance,~\cite{demri:2016a,baier_katoen:2008a,milner:89a}).

We find, moreover, that the sets of (finite) computations as well as the sets of (finite) traces of two bisimilar concurrent game structures are \emph{identical}. In order to see this, the following simple auxiliary result is useful.\label{page:bisim_identical}
\begin{lem}%
	 \label{lemma:deterministic_bisimulation}
	Let $M\sim M'$, $s,s'\in\States$ and $t,t'\in\States'$, and $\direction$ a direction. Then,
	$s\sim t$, $s\transarrow{\direction} s'$, and $t\transarrow{\direction} t'$ together imply $s'\sim t'$.
\end{lem}

\begin{proof}
	Assume $s\sim t$, $s\transarrow{\direction} s'$, and $t\transarrow{\direction} t'$. As $M\sim M'$, there is a $t''\in\States$ such that $t\transarrow{\direction} t''$ and $s'\sim t''$. Since the transition function is deterministic, moreover, it follows that $t''=t'$. Hence, $s'\sim t'$, as desired.
\end{proof}

Using this observation we also have the following result.
\begin{lem}%
	\label{lemma:trace_identity}\label{lemma:computation_identity}
	Let~$M$ and~$M'$ be bisimilar concurrent game structures. Then,
	\begin{enumerate}
		\renewcommand{\itemsep}{.5ex}
	\item\label{item:app:computation_identity}\label{item:computation_identity} {$\computations_M=\computations_{M'}$} and $\fincomputations_M=\fincomputations_{M'}$,
	\item\label{item:app:trace_identity}\label{item:trace_identity} $\mathwordbox{\traces_M=\traces_{M'}}{{\computations_M=\computations_{M'}}
	}$ and $\mathwordbox{\fintraces_M=\fintraces_{M'}}{\computations_M=\computations_{M'}}$.
	\end{enumerate}
\end{lem}

\begin{proof}
	For part~\ref{item:app:computation_identity}, let~$\computation=d_0,d_1,\dots,d_2,\dots$ be a computation in~$\computations_M$ and $s_0,s_1,s_2,\dots$ the states in~$\states_M$
	such that $s_0\xtransarrow{d_0}s_1\xtransarrow{d_1}s_2\xtransarrow{d_2}\cdots$. We show, by induction on~$k$, that there are states~$t_0,t_1,t_2,\dots$ in~$\states_{M'}$ such that
	$t_0\xtransarrow{d_0}t_1\xtransarrow{d_1}t_2\xtransarrow{d_2}\cdots$, where $t_0=s^0_{M'}$.
	It suffices to prove by induction on~$k$ that for every $k\ge 0$ there is a $t_{k+1}$ such that $t_0\xtransarrow{d_0}\dots\xtransarrow{d_k}t_{k+1}$ and $s_{k+1}\bisim t_{k+1}$.
	For $k=0$, we have $t_0=s_{M'}^0$. Then, observe that, by definition,~$s_0=s_M^0$ and, as $M\bisim M'$, it immediately follows that $s_0\bisim t_0$
	and that there is a~$t_{1}$ such that $t_0\xtransarrow{d_0}t_1$ and $s_1\bisim t_1$.
	For the induction step, we may assume that  there are $t_0,\dots,t_k$ with $s_{M'}^0=t_0\xtransarrow{d_0}\dots\xtransarrow{d_{k-1}}t_{k}$ and $s_k\bisim t_k$. By bisimilarity of~$M$ and~$M'$ we then immediately obtain that there is a $t_{k+1}$ such that $t_0\xtransarrow{d_0}\dots\xtransarrow{d_{k-1}}t_k\xtransarrow{d_k}t_{k+1}$. By Lemma~\ref{lemma:deterministic_bisimulation} it then follows that $s_{k+1}\bisim t_{k+1}$.
	Hence, $\computations_M\subseteq\computations_{M'}$.
	As the inclusion in the opposite direction is proven by an analogous argument, we may conclude that $\computations_M=\computations_{M'}$. It also follows that $\fincomputations_M=\fincomputations_{M'}$, the latter being defined as the finite prefixes of~$\computations_M$ and~$\computations_{M'}$, respectively.

Observe that from the argument in part~\ref{item:app:computation_identity} it also follows that $\run_M(\computation)\statewisebisim\run_{M'}(\computation)$ for every $\computation\in\computations_{M}=\computations_{M'}$.
	For part~\ref{item:app:trace_identity}, consider an arbitrary trace~$\trace\in\traces_M$. Then, there is a computation~$\computation\in\computations_M$ such that $\trace_M(\computation)=\trace$. By part~\ref{item:app:computation_identity}, also $\computation\in\computations_{M'}$. Moreover, $\run_M(\computation)\statewisebisim\run_{M'}(\computation)$. By the definition of (statewise) bisimilarity it then follows that $\trace=\trace_M(\computation)=\trace_{M'}(\computation)$. Accordingly, $\traces_M\subseteq\traces_{M'}$ and the inclusion in the opposite direction ensues by an analogous argument. We then conclude that $\fintraces_M=\fintraces_{M'}$, the latter being defined as the finite prefixes of~$\traces_M$ and~$\traces_{M'}$, respectively.
\end{proof}
 Moreover, every (finite) computation~$\computation$ gives rise to statewise bisimilar (finite) runs and identical (finite) traces in bisimilar concurrent game structures.
\begin{lem}%
\label{lemma:computations_induce_bisimilar_runs}
		Let~$M$ and~$M'$ be bisimilar concurrent game structures and $\computation\in\computations_M$ and $\computation'\in\fincomputations_M$. Then,
	\begin{enumerate}
		\renewcommand{\itemsep}{.5ex}
		\item\label{item:app:computations_induce_i}\label{item:computations_induce_i} $\run_M(\computation)\statewisebisim\run_{M'}(\computation)$ and $\finrun_{M}(\computation')\statewisebisim\finrun_{M'}(\computation')$,
		\item\label{item:app:computations_induce_ii}\label{item:computations_induce_ii} $\mathwordbox{\trace_M(\computation)}{\run_M(\computation)}=\mathwordbox[l]{\trace_{M'}(\computation)}{\run_{M'}(\computation)}$
			and $\mathwordbox{\fintrace_{M}(\computation')}{\finrun_{M}(\computation')}=\mathwordbox{\fintrace_{M'}(\computation')}{\finrun_{M'}(\computation')}$.
	\end{enumerate}
\end{lem}

\begin{proof}
	For part~\ref{item:computations_induce_i}, first observe that by virtue of Lemma~\ref{lemma:computation_identity}, we also have that $\computation\in\computations_{M'}$ and $\fincomputation\in\fincomputations_{M'}$.
	Let $\computation=d_0,d_1,d_2,\dots$,
	$\run_M(\computation)=s_0,s_1,s_2,\dots$, and
	$\run_{M'}(\computation)=t_0,t_1,t_2,\dots$.
	We prove by induction on~$k$ that $s_k\bisim t_k$ for every~$k\ge 0$.
	If $k=0$, then $s_0=s^0_M\bisim s^0_{M'}=t_0$. For the induction step,
	we may assume that $s_k\bisim t_k$. Then, $s_k\xtransarrow{d_k}s_{k+1}$ and
	$t_k\xtransarrow{d_k}t_{k+1}$. Lemma~\ref{lemma:deterministic_bisimulation} now yields $s_{k+1}\bisim t_{k+1}$, as desired.
	The argument for the second part of~\ref{item:computations_induce_i} proceeds by an analogous argument.

	Part~\ref{item:app:computations_induce_ii} then follows almost immediately from part~\ref{item:app:computations_induce_i}.
	Let $\run_M(\computation)=s_0,s_1,s_2,\dots$ and $\run_{M'}(\computation)=t_0,t_1,t_2,\dots$. Now observe that for every $k\ge0$ we have that $\evalf_M(s_k)=\evalf_{M'}(t_k)$.
	Accordingly, $\trace_M(\computation)=\trace_{M'}(\computation)$. For~$\computation'$ a similar argument yields the result.
\end{proof}
However, as runs are sequences of states and the states of different concurrent game structures~$M$ and~$M'$ may be distinct, even if they are bisimilar, no identification of their sets~$\runs_M$ and~$\runs_{M'}$  of runs can generally be made.

\section{Games on Concurrent Game Structures}%
\label{secn:bis-in-cgs}

Concurrent game structures specify the actions the players can take at
each state and which states are reached if they all concurrently
decide on an action. In game theoretic terms, these structures loosely
correspond to what are called \emph{game forms}.
A full understanding of the game-theoretic aspects of the system and the
strategic behaviour of its constituent players---and therefore which computations/runs/traces will be generated in equilibrium---also essentially depends on what goals the players desire to
achieve and on what strategies they may adopt in pursuit of these
goals. We therefore augment concurrent game structures with
preferences and strategies for the players. In this way CGSs define fully fledged strategic games and as such they are amenable to game theoretic analysis by standard solution concepts, among which Nash equilibrium is arguably the most prominent.
%

\subsection*{Strategies and Strategy Profiles}
Based on the distinction between computations, runs, and traces, we can also distinguish three types of strategy: computation-based, run-based, and trace-based strategies.
The importance of these distinctions is additionally corroborated by Bouyer et al.~\cite{BBMU:2011a,BBMU:2015a}, who show how the specific model of strategies adopted affects the computational complexity of some standard decision problems related to multi-agent systems.

A \emph{computation-based strategy} for a player~$i$ in a concurrent game structure~$M$ is a function
\[
	f^{\mathit{comp}}_i\colon\fincomputations_M\to\actions\text,
\]
such that $f_i^{\mathit{comp}}(\fincomputation)\in\Ac_i(s_{k})$ for every finite~$\fincomputation\in\fincomputations_M$ with $\finrun_M(\fincomputation)=s_0,\dots,s_k$.
Thus, in particular, $f_i^{\mathit{comp}}(\epsilon)\in\Ac_i(s^0_M)$, where~$\epsilon$ is the empty sequence of directions.

Similarly, a \emph{run-based strategy} for player~$i$ is a
function
\[
	f^{\mathit{run}}_i\colon\finruns_M\to\actions\text{,}
\]
where $f_i^{\mathit{run}}(s_0,\dots,s_k)\in\Ac_i(s_k)$ for every finite run~$(s_0,\dots,s_k)\in\finruns_M$.
Finally, a \emph{trace-based strategy} for~$i$ is a function
\[
	f^{\mathit{trace}}_i\colon\fintraces_M\to\actions\text,
\]
such that $f_i^{\mathit{trace}}(\fintrace)\in\Ac_i(s_k)$ for every trace~$\fintrace\in\fintraces_M$ and every run~$\finrun=s_0,\dots,s_k$ such that $\trace=\valf(s_0),\dots,\valf(s_k)$.

A \emph{computation-based strategy profile} is then a tuple
$\stratprof=(\strategy_1,\dots,\strategy_n)$ that associates with each
player~$i$ a computation-based strategy~$f_i$. Run-based and
trace-based strategy profiles are defined analogously.

Every computation-based strategy profile $f=(f_1,\dots,f_n)$ induces a unique computation
\[
	\computation_M(f)=d_0,d_1,d_2,\dots
\]
in~$M$ that is defined inductively as follows:
\begin{align*}
	d_0		&	=	(f_1(\epsilon),\dots,f_n(\epsilon))\\
	d_{k+1}	&	=	(f_1(d_0,\dots,d_{k}), \ldots,f_n(d_0,\dots,d_{k})).
\end{align*}
A run-based strategy profile $f=(f_1,\dots,f_n)$ defines a unique computation~$\computation_M(f)=d_0,d_1,d_2,\dots$ in a similar manner:
\begin{align*}
	d_0		&	=	(f_1(s^0_M),\dots,f_n(s^0_M))\text{, and}\\
	d_{k+1}	&	=	(f_1(\finrun(d_0,\dots,d_k)),\dots,f_n(\finrun(d_0,\dots,d_k)))\text.
\end{align*}
Finally, the computation~$\computation_M(f)$ defined by a trace-based strategy profile~$f$ is given by
\begin{align*}
	d_0		&	=	(f_1(\valf(s^0_M)),\dots,f_n(\valf(s^0_M)))\\
	d_{k+1}	&	=	(f_1(\trace(d_0,\dots,d_k)), \dots,f_n(\trace(d_0,\dots,d_k)))\text.
\end{align*}
If~$M$ is clear from the context, we usually omit the subscript in~$\computation_M(f)$.
For $f=(f_1,\dots,f_n)$ a profile of computation-based, run-based, or trace-based strategies,
we write with a slight abuse of notation  $\run(f_1,\dots,f_n)$ for $\run(\computation(f_1,\dots,f_n))$ and $\trace(f_1,\dots,f_n)$ for $\trace(\run(f_1,\dots,f_n))$.

As the computations of bisimilar concurrent games structures coincide (Lemma~\ref{lemma:computation_identity}), we can now establish that a
player's computation-based strategies coincide in bisimilar
concurrent game structures. Moreover, the computations induced by them will be identical.
Also, from the coincidence of traces between bisimilar concurrent game structures (Lemma~\ref{lemma:trace_identity}), we can establish also trace-based strategies coincide in bisimilar concurrent game structures.
%
%
%

\begin{lem}%
	\label{lemma:strategy_identity}\label{lemma:computation_tilde0}\label{lemma:trace_tilde0}
	Let~$M$ and $M'$ be bisimilar concurrent game structures and~$i$ a player.
	Then, every computation-based  strategy for~$i$ in~$M$ is also a computation-based strategy for~$i$ in~$M'$, and every trace-based  strategy for~$i$ in~$M$ is also a trace-based strategy for~$i$ in~$M'$. 	Moreover, for every computation-based profile~$f$ for~$M$ we have that $\computation_M(f)=\computation_{M'}(f)$, and for every trace-based profile~$g$ that $\computation_M(g)=\computation_{M'}(g)$.
\end{lem}

\begin{proof}
	First, let $f_i$ be a computation-based strategy for~$i$ in~$M$. We show that~$f_i$ is also a computation-based strategy for~$i$ in~$M'$. To this end, consider an arbitrary $\fincomputation\in\fincomputations_{M'}$. Let $\finrun_{M'}(\fincomputation)=t_0,\dots,t_k$.
It suffices to prove that $f_i(\fincomputation)\in\actions_i(t_k)$.
To see this, first observe that by Lemma~\ref{lemma:computation_identity} also $\fincomputation\in\fincomputations_M$ and let $\finrun_M(\fincomputation)=s_0,\dots,s_k$.
In virtue of Lemma~\ref{lemma:deterministic_bisimulation}, then $s_k\bisim t_k$.
Moreover, because~$f_i$ is a computation-based strategy for~$i$ in~$M$, we have $f_i(\fincomputation)\in\actions_i(s_k)$. Now consider any legal direction~$d_{k}=(a_1,\dots,a_n)$ at~$s_k$ with $a_i=f_i(\fincomputation)$.
Then, there is some state~$s_{k+1}\in\states_M$ such that $s_k\xtransarrow{d_k}s_{k+1}$. As $s_k\bisim t_k$, moreover, there is also a state~$t_{k+1}\in\states_{M'}$ such that $t_k\xtransarrow{d_k}t_{k+1}$. Accordingly,~$d_{k}$ is legal at~$t_k$ in~$M'$ and in particular $a_i=f_i(\fincomputation)\in\actions_i(t_k)$ as desired.

The case if~$g_i$ is a trace-based strategy for~$i$ in~$M$ is similar. We then have to prove that~$g_i$ is also a trace-based strategy for~$i$ in~$M'$ as well. To this end, consider an arbitrary finite trace~$\fintrace\in\fintraces_{M'}$ and run~$\finrun=t_0,\dots,t_k$ such that $\fintrace=\valf_{M'}(t_0),\dots,\valf_{M'}(t_k)$. It then suffices to prove that $g_i(\fintrace)\in\actions_i(t_k)$. We may assume that~$\finrun$ is induced by a computation~$\fincomputation\in\fincomputations_{M'}$, that is, $\finrun=\finrun_{M'}(\fincomputation)$. By Lemma~\ref{lemma:computation_identity} we have $\fincomputation\in\fincomputations_M$ and by Lemma~\ref{lemma:computations_induce_bisimilar_runs} both $\finrun_M(\fincomputation)\statewisebisim\finrun_{M'}(\fincomputation)$ and
$\trace_M(\fincomputation)=\trace_{M'}(\fincomputation)$.
Let $\finrun_M(\fincomputation)=s_0,\dots,s_k$. Hence, $s_k\bisim t_k$. As~$g_i$ is a run-based strategy for~$i$ in~$M$ we have $g_i(\fintrace)\in\actions_i(s_k)$.
Let, furthermore, $d=(a_1,\dots,a_k)$ be a legal direction at~$s_k$ with $a_i=g_i(\fintrace)$.
Then, there is some state~$s_{k+1}\in\states_M$ such that $s_k\xtransarrow{d_k}s_{k+1}$. As $s_k\bisim t_k$, there is also a state~$t_{k+1}\in\states_{M'}$ such that $t_k\xtransarrow{d_k}t_{k+1}$. Accordingly,~$d_{k}$ is legal at~$t_k$ in~$M'$ and in particular $a_i=g_i(\fintrace)\in\actions_i(t_k)$.

For the second part of the lemma, let $f=(f_1,\dots,f_n)$ be a computation-based strategy profile in~$M$. Then,~$f$ is a computation-based strategy profile in~$M'$ as well. Let $\computation_M(f)=d_0,d_1,d_2,\dots$
and $\computation_{M'}(f)=d'_0,d'_1,d'_2,\dots$. We show by induction on~$k$ that for every $k\ge0$ we have~$d_k=d'_k$. For $k=0$, immediately,
$
	d_0=(f_1(\epsilon),\dots,f_n(\epsilon))=d_0'\text.
$
For the induction step we may assume that $d_0,\dots,d_k=d_0',\dots,d'_k$. Hence,
\begin{align*}
	d_{k+1}		=	(f_1(d_0,\dots,d_k),\dots,f_n(d_0,\dots,d_k))
				=	(f_1(d'_0,\dots,d'_k),\dots,f_n(d'_0,\dots,d'_k))
				=	d'_{k+1}\text.
\end{align*}

Finally, let $g=(g_1,\dots,g_n)$ be a trace-based strategy profile. Again we let $\computation_M(f)=d_0,d_1,d_2,\dots$
and $\computation_{M'}(f)=d'_0,d'_1,d'_2,\dots$ and show by induction on~$k$ that for every $k\ge0$ we have that $d_k=d'_k$.
If $k=0$, observe that having assumed $M\bisim M'$ also $s^0_M=s^0_{M'}$.
Accordingly, $\evalf_M(s_M^0)=\evalf_{M'}(s^0_{M'})$ and, hence,
\begin{align*}
	d_{0}		=	(g_1(\evalf_M(s_M^0)),\dots,g_n(\evalf_M(s_M^0)))
				=	(g_1(\evalf_{M'}(s^0_{M'})),\dots,g_n(\evalf_{M'}(s^0_{M'})))
				=	d'_{k+1}\text.
\end{align*}
For the induction step, we may assume that $d_0,\dots,d_k=d'_0,\dots,d'_k$ and by Lemma~\ref{lemma:computations_induce_bisimilar_runs}, moreover, $\fintrace_{M}(d_0,\dots,d_k)=\fintrace_{M'}(d'_0,\dots,d'_k)$.
Now the following equations hold:
\begin{align*}
	d_{k+1}	&	=	(g_1(\fintrace_M(d_0,\dots,d_k)),\dots,g_n(\fintrace_{M}(d_0,\dots,d_k)))\\
			&	=	(g_1(\fintrace_{M'}(d'_0,\dots,d'_k)),\dots,g_n(\fintrace_{M'}(d'_0,\dots,d'_k)))\\
			&	=	d'_{k+1}\text,
\end{align*}
which concludes the proof.
\end{proof}

With the states of bisimilar structures possibly being
distinct, however, a statement  analogous to Lemma~\ref{lemma:strategy_identity} cannot
be shown to hold for run-based strategies.

\subsection*{Preferences and Goals}\label{section:preferences}

We assume the agents of a concurrent game structure to have preferences on basis of which they choose their strategies.
Formally, we specify the preferences of a player~$i$ of a CGS~$M$
as a subset~$\goalset_i$ of \emph{computations},
that is, $\goalset_i \subseteq \computations_M$ and refer
to~$\goalset_i$ as~$i$'s \emph{goal set}. Player~$i$ is then
understood to {(strictly) prefer} computations in~$\goalset_i$ to
those not in~$\goalset_i$ and to be {indifferent}
otherwise. Accordingly, each player's preferences are dichotomous,
only distinguishing between the preferred computations in~$\goalset_i$
and the not preferred ones not in~$\goalset_i$.  Formally, player~$i$
is said to \emph{weakly prefer} computation~$\computation$ to
computation~$\computation'$ if $\computation\in\goalset_i$
whenever~$\computation'\in\goalset_i$, and to \emph{strictly prefer}
$\computation$ to~$\computation'$ if~$i$ weakly prefers~$\computation$
to~$\computation'$ but not the other way round.  If~$i$ both weakly
prefers~$\computation$ to~$\computation'$ and weakly
prefers~$\computation'$ to~$\computation$, player~$i$ is said to be
\emph{indifferent} between~$\computation$ and~$\computation'$.

%
Our choice to assume the players' preferences to be \emph{computation-based} preferences---that is, to model their goals as sets of
\emph{computations} rather than, say, sets of runs or sets of
traces---is for technical convenience and flexibility. Recall that
every set of runs induces a set of computations, namely the set of
computations that give rise to the same runs, and similarly for every
set of traces. Thus, we say that a goal
set~$\goalset_i\subseteq\computations_M$ is \emph{run-based} if for
any two computations~$\computation$ and~$\computation'$ with
$\run(\computation)=\run(\computation')$ we have that
$\computation\in\goalset_i$ if and only if
$\computation'\in\goalset$. Similarly, $\goalset_i$ is said to be
\emph{trace-based} whenever
$\trace(\computation)=\trace(\computation')$ implies
$\computation\in\goalset_i$ if and only if
$\computation'\in\goalset_i$. In other words, in our setting, formally,
\emph{run-based} goals are \emph{computation-based} goals closed under induced runs, and
\emph{trace-based} goals are \emph{computation-based} goals closed
under induced traces.%
\footnote{We do not directly consider sets of runs or sets of traces as
possible models of players' preferences in this paper---formally, they are induced sets of computations. Accordingly, when talking about preferences,
we need not make the distinction between `run-based' (`trace-based') and
`run-invariant' (`trace-invariant') as we do for strategies. Our run-based
preferences and trace-based preferences can with as much justification be
referred to as run-invariant preferences and trace-invariant preferences,
respectively.}

Sometimes---as we did in the example in the introduction---players' goals are specified by \emph{temporal logic formulae}~\cite{demri:2016a}.
As the satisfaction of goals only depends on traces, they will
directly correspond to trace-based goals, given our formalisation of
goals and preferences.

%
\subsection*{Games and Nash Equilibrium}%
\label{section:trace_invariant}
With the above definitions in place, we are now in a position to
define a \emph{game on a concurrent game structure~$M$} (also called a \emph{CGS-game}) with
  $\agents=\set{1,\dots,n}$ as a
tuple
\[
	G=(M,\goalset_1,\dots,\goalset_n
	)\text,
\]
where, for each player~$i$ in~$M$, the
set~$\goalset_i\subseteq\computations_M$ is a goal set
specifying~$i$'s dichotomous preferences over the computations in~$M$.

In a CGS-game the players can all play either computation-based strategies,  run-based strategies, or trace-based strategies.
For each such choice of type of strategies, with the set of players and their preferences specified,
every CGS-game defines a strategic game in the standard game-theoretic sense.
Observe that the set of strategies is infinite in general.
Thus the game-theoretic solution concept of Nash equilibrium becomes available for the analysis of games on concurrent game structures.
If $f=(f_1,\dots,f_n)$ is a strategy profile and~$g_i$ a strategy for player~$i$, we write $(f_{-i},g_i)$ for the strategy profile
$(f_1,\dots,g_i,\dots,f_n)$,
which is identical to~$f$ except that~$i$'s strategy is replaced by~$g_i$.
Formally, given a CGS-game, we  say that a profile~$\strategy=(\strategy_1,\dots,\strategy_n)$ of computation-based strategies is a \emph{Nash equilibrium in computation-based strategies} (or \emph{computation-based equilibrium}) if, for every player~$i$ and every computation-based  strategy~$g_i$ available to~$i$,
\[
	\text{
	$\computation_M(f_{-i},g_i)\in\goalset_i$
	implies
	$\computation_M(f)\in\goalset_i$.
	}
\]
The concepts of \emph{Nash equilibrium in run-based strategies} and \emph{Nash equilibrium in trace-based strategies} are defined analogously, where, importantly, the strategies in~$f_{-i}$ and~$g_i$ are required to be of the same type, that is, either they are all run-based or they are all trace-based.
If $\computation(f)\notin\goalset_i$ whereas~$\computation(f_{-i},g_i)\in\goalset_i$, we also say that player~$i$ would like to \emph{deviate from~$f$} (and play~$g_i$ instead). Thus, a run-based profile~$f$ is a \emph{run-based equilibrium} whenever no player would like to deviate from it and play some \emph{run-based} strategy different from~$f_i$. Similarly, a trace-based profile~$f$ is a trace-based equilibrium if no player likes to deviate and play another \emph{trace-based} strategy.

We say that a computation~$\computation$, run~$\run$, or a trace~$\trace$ is \emph{sustained by a Nash equilibrium}~$\strategy=(\strategy_1,\dots,\strategy_n)$ (of any type) whenever $\computation=\computation(\strategy)$, $\run=\run(\strategy)$, and $\trace=\trace(\strategy)$, respectively. We also refer to a computation, run, or trace that is sustained by a Nash equilibrium as an \emph{equilibrium computation}, \emph{equilibrium run}, and \emph{equilibrium trace}, respectively.

Computation-based equilibrium is a weaker notion than run-based
equilibrium, in the sense that if~$f$ is a run-based equilibrium there
is also a corresponding computation-based equilibrium, but not
necessarily the other way round. Run-based equilibrium, in turn, is in
a similar way a weaker concept than trace-based equilibrium.
As computation-based, run-based, and trace-based strategies are
set-theoretically of different types, a comparison cannot be made
directly. To make the comparison precise, we therefore identify
two subclasses of computation-based strategies, \emph{run-invariant strategies} and \emph{trace-invariant strategies}, that characterise the behaviour of, respectively, run-based and trace-based strategies.

We say that a computation-based strategy~$f_i\colon\fincomputations_M\to\actions_i$ is \emph{run-invariant} in CGS~$M$ whenever  $\finrun_M(\fincomputation)=\finrun_M(\fincomputation')$ implies $f_i(\fincomputation)=f_i(\fincomputation')$, for all computations $\fincomputation,\fincomputation'\in\fincomputations_M$.
Similarly,~$f_i$ is \emph{trace-invariant} in~$M$ whenever  $\fintrace_M(\fincomputation)=\fintrace_M(\fincomputation')$ implies $f_i(\fincomputation)=f_i(\fincomputation')$, for all
$\fincomputation,\fincomputation'\in\fincomputations_M$.
Observe that thus a strategy~$f_i$ being trace-invariant implies~$f_i$ being run-invariant, but not necessarily the other way around.

We observe that there are one-to-one correspondences between run-based strategies on the one hand and run-invariant computation-based strategies on the other, and similarly between trace-based strategies and trace-invariant computation-based strategies.
Let $f_i\colon\finruns_M\to\actions$ be a run-based strategy. Then define $\check f_i\colon\fincomputations_M\to\actions$ as the computation-based strategy such that for every finite computation~$\fincomputation\in\fincomputations_M$ we have
$
	\check f_i(\fincomputation)=f_i(\finrun_M(\fincomputation)).
$
A similar statment holds if $g_i\colon\fintraces_M\to\actions$ is a trace-based strategy. Then, define $\dcheck g_i\colon\fincomputations\to\actions_i$ as the computation-based strategy such that for every finite computation~$\fincomputation\in\fincomputations_M$ we have
$
	\dcheck g_i(\fincomputation)=g_i(\fintrace_M(\fincomputation)).
$

\begin{lem}%
	\label{lemma:invariant_one_one}
	For run-based strategies~$f_i$ and trace-based strategies~$g_i$, the mapping that transforms~$f_i$ into $\check f_i$ and the mapping that transforms $g_i$ into~$\dcheck g_i$ are both one-to-one.
\end{lem}

\begin{proof}
	Let~$f_i$ a run-based strategy. We first show that~$\check f_i$ is run-invariant.
	To this end, let $\fincomputation,\fincomputation'\in\fincomputations_M$ be computations such that $\finrun_M(\fincomputation)=\finrun_M(\fincomputation')$. Then,
	\[
		\check f_i(\fincomputation)
		=	f_i(\finrun_M(\fincomputation))
		=	f_i(\finrun_M(\fincomputation'))
		= \check f_i(\fincomputation')\text.
	\]

	To show that the mapping is onto, let~$g_i$ be an arbitrary run-invariant strategy.
	Now define run-based strategy~$\hat g_i$ such that, for every run $\finrun\in\finruns_M$
	and $\fincomputation\in\fincomputations_M$ with $\finrun=\finrun_M(\fincomputation)$ we have
	$\hat g_i(\finrun)=g_i(\fincomputation)$.
Observe that~$\hat g_i$ is well-defined since, by run-invariance of~$g_i$, for all $\fincomputation,\fincomputation'\in\fincomputations_M$ with $\finrun_M(\fincomputation)=\finrun_M(\fincomputation')=\finrun$ we have that $g_i(\fincomputation)=g_i(\fincomputation')$.

	Finally, to see that the mapping is injective, let~$f_i$ and~$f'_i$ be two distinct run-based strategies. Then, there is a run~$\finrun\in\finruns_M$ such that $f_i(\finrun)\neq f'_i(\finrun)$. We may assume the existence of a computation~$\fincomputation\in\fincomputations_M$ such that $\finrun_M(\fincomputation)=\finrun$. Then,
	\[
		\check f_i(\fincomputation)
		=	f_i(\finrun_M(\fincomputation))
		=	f_i(\finrun)
		\neq f'_i(\finrun)
		=	f'_i(\finrun_M(\fincomputation))
		= 	\check f_i(\fincomputation)\text,
	\]
as desired. This concludes the proof.
\end{proof}
Furthermore, each profile of run-invariant strategies induces the same computation in a concurrent game structure as its run-based counterpart. A similar remark applies to trace-invariant and trace-based profiles.
\begin{lem}\label{lemma:invariant_computations}
	Let $f=(f_1,\dots,f_n)$ be a run-based profile and $g=(g_1,\dots,g_n)$ a trace-based profile. Then,
\begin{align*}
	\computation_M(f_1,\dots,f_n)&=\computation_M(\check f_1,\dots,\check f_n)
	&\text{and}&&
	\computation_M(f_1,\dots,f_n)&=\computation_M(\dcheck f_1,\dots,\dcheck f_n)\text.
\end{align*}
\end{lem}

\begin{proof}
	Let $\computation_M(f_1,\dots,f_n)=d_0,d_1,d_2,\dots$
	and $\computation_M(\check f_1,\dots,\check f_n)=d'_0,d'_1,d'_2,\dots$.
	We prove by induction on~$k$ that $d_0,\dots,d_k=d'_0,\dots,d'_k$ for every $k\ge 0$.

	If $k=0$, recall that $\finrun_M(\epsilon)=s_M^0$. Hence,
	\[
		d_0
		=
		(f_1(s_M^0),\dots,f_n(s_M^0))
		=
		(f_1(\finrun_M(\epsilon)),\dots,f_n(\finrun_M(\epsilon)))
		=
		(\check f_1(\epsilon),\dots,\check f_n(\epsilon))
		=
		d'_0\text.
	\]

    \noindent
	For the induction step, we may assume that $d_0,\dots,d_k=d'_0,\dots,d'_k$.
	Now the following equalities hold.
	\begin{align*}
		d_{k+1}
		&	=_{\phantom{i.h.}}
		(f_1(\finrun_M(d_0,\dots,d_k)),\dots,f_n(\finrun_M(d_0,\dots,d_k)))
		\\
		&=_{\phantom{i.h.}}
		(\check f_1(d_0,\dots,d_k),\dots,\check f_n(d_0,\dots,d_k))
		\\
		&=_{i.h.}
		(\check f_1(d'_0,\dots,d'_k),\dots,\check f_n(d'_0,\dots,d'_k))
		\\
		&=_{\phantom{i.h.}}
		d'_{k+1}\text.
	\end{align*}
	We may conclude that $d_0,\dots,d_{k+1}=d'_0,\dots,d'_{k+1}$.
		The argument for trace-based and trace-invariant strategies runs along analogous lines, \emph{mutatis mutandis}.
\end{proof}
We say that a computation-based profile $f=(f_1,\dots,f_n)$ is a \emph{run-invariant equilibrium} in a CGS-game if~$f$ is run-invariant and no player~$i$ wishes to deviate from~$f$ and play another \emph{run-invariant} strategy~$f'_i$.
Similarly, a computation-based profile $f=(f_1,\dots,f_n)$ is a \emph{trace-invariant equilibrium} in a CGS-game if~$f$ is trace-invariant and no player~$i$ wishes to deviate from~$f$ and play another \emph{trace-invariant} strategy~$f'_i$.
As an immediate consequence of Lemmas~\ref{lemma:invariant_one_one} and~\ref{lemma:invariant_computations} we have the following corollary.

\begin{lem}\label{lemma:invariant_eq}
	Let $f=(f_1,\dots,f_n)$ and $g=(g_1,\dots,g_n)$ be a run-based profile, respectively, a trace-based profile in a CGS-game~$G=(M,\goalset_1,\dots,\goalset_n)$ based on~$M$. Then,
	\begin{enumerate}
		\item $f$ is a \wordbox{run-based equilibrium}{trace-based equilibrium} if and only if~~$\check f$ is a \wordbox{run-invariant equilibrium}{trace-invariant equilibrium},
		\item $g$ is a trace-based equilibrium if and only if~~$\dcheck g$ is a trace-invariant equilibrium.
	\end{enumerate}
\end{lem}

\begin{proof}
	For the run-based case, the following equivalences hold by virtue of Lemmas~\ref{lemma:invariant_one_one} and~\ref{lemma:invariant_computations}:
	\begin{align*}
	&	\text{$f$ is not a run-based equilibrium in $G$}\\
	&	\text{iff}\quad \text{$\computation_M(f)\notin\goalset_i$ and $\computation_{M}(f_{-i},f'_i)\in\goalset_i$ for some run-based strategy~$f'_i$ for some player~$i$}\\
	& 	\text{iff}\quad \text{$\computation_M(\check f)\notin\goalset_i$ and $\computation_{M}(\check f_{-i},\check f'_i)\in\goalset_i$ for some run-based strategy~$f'_i$ for some player~$i$}\\
	& 	\text{iff}\quad \text{$\computation_M(\check f)\notin\goalset_i$ and $\computation_{M}(\check f_{-i},f''_i)\in\goalset_i$ for some run-invariant strategy~$f''_i$ for some player~$i$}\\
	&	\text{iff}\quad\text{$\check f$ is not a run-invariant equilibrium in $G$}
	\end{align*}
	The proof of the second part is by an analogous argument, \emph{mutatis mutandis}.
\end{proof}

Computation-based strategies grant a player more strategic flexibility than do run-invariant strategies. A similar remark applies to run-invariant strategies and trace-invariant strategies. Still, we find that, if a player~$i$ wishes to deviate from a computation-based profile~$f$ and play another computation-based strategy,~$i$ would also like to deviate by playing a run-invariant or even a trace-invariant strategy. This insight underlies the following result.\footnote{The situation can be compared to the relation between equilibria in \emph{pure} and \emph{mixed} (or \emph{randomised}) strategies in game theory. There every equilibrium in pure strategies is also an equilibrium in mixed strategies, because, if a player wishes to deviate from a mixed profile, she wishes to deviate by playing a pure, that is, not randomised, strategy.}

\begin{thm}\label{theorem:comp_run_trace_invariant_eq}
	Let $f=(f_1,\dots,f_n)$ be a run-invariant profile and $g=(g_1,\dots,g_n)$ a trace-invariant profile in CGS-game~$G=(M,\goalset_1,\dots,\goalset_n)$ based on~$M$. Then,
	\begin{enumerate}
	\item\label{item:comp_run_trace_invariant_eq_i}  \wordbox{$f$ is a run-invariant equilibrium}{$g$ is a trace-invariant equilibrium} if and only if \wordbox{$f$ is a computation-based equilibrium}{$g$ is a computation-based equilibrium},
	\item\label{item:comp_run_trace_invariant_eq_ii} $g$ is a trace-invariant equilibrium if and only if~$g$ is a computation-based equilibrium.
	\end{enumerate}
\end{thm}

\begin{proof}
	For part~\ref{item:comp_run_trace_invariant_eq_i}, first assume that~$f$ is a run-invariant equilibrium in $G$.
	For a contradiction assume moreover that~$f$ is not a computation-based equilibrium.
	Then, there is a player~$i$ and a computation-based strategy~$f'_i$ such that $\computation_M(f)\notin\goalset_i$ whereas
	$\computation_M(f_{-i},f'_i)\in\goalset_i$.
	Let $\computation_M(f_{-i},f'_i)=d'_0,d'_1,d'_2,\dots$.
	 Observe that~$f'_i$ need not be run-invariant.
	We therefore define strategy~$f''_i$ for player~$i$ such that $f''_i(\epsilon)=f'_i(\epsilon)$ and,
	for all finite computations $d_0,\dots,d_k$,
	\[
		f''_i(d_0,\dots,d_k)
		=
		\begin{cases}
			f'_i(d'_0,\dots,d'_k)	&	\text{if $\finrun_M(d_0,\dots,d_k)=\finrun_M(d'_0,\dots,d'_k)$,}\\
			f_i(d_0,\dots,d_k)	&	\text{otherwise.}
		\end{cases}
	\]
	As~$f_i$ is run-invariant, this definition guarantees that~$f''_i$ is run-invariant as well.
	Let $\computation_M(f_{-i},f''_i)=d_0'',d_2'',d_2'',\dots$. We prove by induction
	on~$k$ that $d'_0,\dots,d'_k=d''_0,\dots,d''_k$, for every $k\ge 0$, and hence that $\computation_M(f_{-i},f'_i)=\computation_M(f_{-i},f''_i)$.
	If $k=0$, we immediately obtain that
	\begin{align*}
	d'_0
	&	=	(f_1(\epsilon),\dots,f'_i(\epsilon),\dots,f_n(\epsilon))
		=	(f_1(\epsilon),\dots,f''_i(\epsilon),\dots,f_n(\epsilon))
	=	\direction''_0\text.
	\end{align*}
	For the induction step, we may assume that $d'_0,\dots,d'_k=d''_0,\dots,d''_k$.
and the following equalities hold:
	\begin{align*}
		d'_{k+1}
		&		=_{\phantom{i.h.}}	(f_1(d_1',\dots,d_k'),\dots,f'_i(d_1',\dots,d_k'),\dots,f_n(d_1',\dots,d_k'))
		\\&		=_{\phantom{i.h.}}	(f_1(d_1',\dots,d_k'),\dots,f''_i(d_1'',\dots,d_k''),\dots,f_n(d_1',\dots,d_k'))
		\\&		=_{{i.h.}}			(f_1(d_1'',\dots,d_k''),\dots,f''_i(d_1'',\dots,d_k''),\dots,f_n(d_1'',\dots,d_k''))
		\\&		=_{\phantom{i.h.}} 	d''_{k+1}\text.
	\end{align*}
	Observe that the second equality holds by virtue of the definition of~$f''_i$ and $\finrun_M(d'_0,\dots,d'_k) = \finrun_M(d''_0,\dots,d''_k)$.
	It would follow that $\computation_M(f_{-i},f''_i)\in\goalset_i$ as well, and, as $f''_i$ is run-invariant, moreover that~$f$ is not a run-invariant equilibrium, a contradiction.

	For the opposite direction, assume for contraposition that~$f$ is not a run-invariant equilibrium. Then, there is some player~$i$
	who would like to deviate  from~$f$ and play some run-invariant strategy~$f'_i$. As run-invariant strategies are strategy-based by definition,
	it follows that~$f$ is not a computation-based equilibrium either.

	Part~\ref{item:comp_run_trace_invariant_eq_ii} follows by an analogous argument, \emph{mutatis mutandis}.
\end{proof}
Theorem~\ref{theorem:comp_run_trace_invariant_eq} does not preclude the existence of computation-based equilibria that fail to be run-invariant or trace-invariant, that is, the three equilibrium concepts---computation-based, run-invariant, and trace-invariant equilibrium---are not equivalent. However, they can be ordered with respect to how restrictive they are, that is, with respect to the sets of profiles they exclude as solutions.

\begin{cor}
	Let $f=(f_1,\dots,f_n)$ be a computation-based profile in some CGS-game~$G=(M,\goalset_1,\dots,\goalset_n)$ based on~$M$. Then,
	\begin{enumerate}
	\item \wordbox[l]{$f$ is a run-invariant equilibrium}{$f$ is a trace-invariant equilibrium} implies~$f$ is a computation-based equilibrium,
	\item $f$ is a trace-invariant equilibrium implies~$f$ is a run-invariant equilibrium.
	\end{enumerate}
\end{cor}

\begin{proof}
	Merely observe that if~$f$ is a run-invariant equilibrium, it is also a run-invariant profile. If~$f$ is moreover trace-invariant it is also run-invariant. The result then immediately follows from Theorem~\ref{theorem:comp_run_trace_invariant_eq}.
\end{proof}
On basis of the findings in this section, we may with justification claim that every trace-based equilibrium corresponds to a run-based equilibrium, and that every run-based equilibrium corresponds with some computation-based equilibrium, even if the converses of these statements do not generally hold.

\section{Invariance of Nash Equilibria under Bisimilarity}%
\label{secn:fail-bis}

From a computational point of view, one may expect games based on bisimilar concurrent game structures and with identical players' preferences to exhibit similar properties, in particular with respect to their Nash equilibria.
We find that that this is indeed the case for games with
computation-based strategies as well as for games with trace-based
strategies.
Recall that (finite) computations and (finite) traces
are unaffected by (state-splitting and state-merging) operations on CGSs that preserve bisimilarity (Lemma~\ref{lemma:trace_identity}).
As a consequence the sets of computation-based strategies and trace-based strategies available to an again are the same in bisimilar CGSs (Lemma~\ref{lemma:strategy_identity}), providing the intuitive basis for these observations.

For games with run-based strategies the situation is
considerably more complicated. Here, a key observation is that, by contrast
to computation-based and trace-based strategies, there need not be a
natural \emph{one-to-one} mapping between the sets of run-based
strategies in bisimilar concurrent game models. By restricting attention to
so-called bisimulation-invariant run-based strategies, however, we
find that order can be restored.
%
\subsection*{Invariance under Bisimilarity and Preference Types}%
\label{section:congruence}
We are primarily interested in the
Nash equilibria of games that are the same up to bisimilarity of the
underlying concurrent game structures. The Nash equilibria of a game, however, essentially depend on the players' preferences. Accordingly, the Nash equilibria of two bisimilar CGS-games can only be meaningfully compared if we also
we assume that the players' preferences in these two games are identical. We
formalised players' preferences as sets of computations, and, due to
\bphtext{Lemma~\ref{lemma:computation_identity}}, this enables a
straightforward comparison of players' goal sets across bisimilar
concurrent game structures.

In Section~\ref{section:preferences}, we also distinguished run-based
and trace-based preferences, that is, goal sets closed under
computations that induce the same runs and traces, respectively. We
are also interested in the invariance of the existence of Nash equilibria
 in games on bisimilar concurrent game structures
where the players' preferences games are what we will call \emph{congruent}, that is, both \emph{the same} and \emph{of the same type in both games}.

For computation-based and
trace-based preferences the issue of congruence is moot.
\bphtext{Observe that for
  bisimilar concurrent game structures~$M$ and~$M'$, if a goal
  set~$\goalset_i$ is computation-based in~$M$, then it is also
  computation-based in~$M'$. Due to
  Lemma~\ref{lemma:computations_induce_bisimilar_runs}, the same holds
  for trace-based preferences.}

\begin{figure}[t]
\centering
 	\scalebox{.8}{
 	  \begin{tikzpicture}[scale=1]

 	  \tikzstyle{every ellipse node}=[draw,inner xsep=3.5em,inner ysep=1.2em,fill=black!15!white,draw=black!15!white]
 	  \tikzstyle{every circle node}=[fill=white,draw,minimum size=1.6em,inner sep=0pt]

	  \draw(0,-3)	node(){};
	  \draw(7.8,3.25)	node(){};

 	  \draw(0,0) node(0){}
 	  			++ 	( 0:.9)	  node[label=-90:{\footnotesize$p$},circle](v0){$s_0$}
					+ (.8*6,.8*3)  node[label=-90:{\footnotesize$p$},circle](v1){$s_1$}
 					+ (.8*6,.8*0) node[label=-90:{\footnotesize$p$},circle](v3){$s_2$}
 					+ (.8*6,.8*-3) node[label=-90:{\footnotesize$\bar p$},circle](v4){$s_3$}
 				;


 	\draw[-latex] (v0) --node[pos=.55,fill=white](){\footnotesize$\begin{array}{l}a,a\end{array}$} (v1);
 	\draw[-latex] (v0) --node[pos=.55,fill=white](){\footnotesize$\begin{array}{l}a,b\\b,a\end{array}$} (v3);
 	\draw[-latex] (v0) --node[pos=.55,fill=white](){\footnotesize$\begin{array}{l}b,b\end{array}$}  (v4);

 	\draw[-latex] (0) -- (v0);
 	%
 	\draw[-latex] (v1.70-90) .. controls +(60-90:4em) and +(120-90:4em) .. node[pos=.5,fill=white,right,xshift=.5ex](){\footnotesize$\ast,\ast$} (v1.110-90);


 	\draw[-latex] (v3.70-90) .. controls +(60-90:4em) and +(120-90:4em) .. node[pos=.5,fill=white,right,xshift=.5ex](){\footnotesize$\ast,\ast$} (v3.110-90);

 	\draw[-latex] (v4.70-90) .. controls +(60-90:4em) and +(120-90:4em) .. node[pos=.5,fill=white,right,xshift=.5ex](){\footnotesize$\ast,\ast$} (v4.110-90);

 	\end{tikzpicture}
 	}
\qquad\qquad
\scalebox{.8}{
 	  \begin{tikzpicture}[scale=1]

 	  \tikzstyle{every ellipse node}=[draw,inner xsep=3.5em,inner ysep=1.2em,fill=black!15!white,draw=black!15!white]
 	  \tikzstyle{every circle node}=[fill=white,draw,minimum size=1.6em,inner sep=0pt]

	  \draw(0,-3)	node(){};
	  \draw(7.8,3.25)	node(){};

 	  \draw(0,0) node(0){}
 	  			++ 	( 0:.9)	  node[label=-90:{\footnotesize$p$},circle](v0){$t_0$}
					+ (.8*6,.8*3)  node[label=-90:{\footnotesize$p$},circle](v1){$t_1$}
 					+ (.8*6,.8*0) node[label=-90:{\footnotesize$p$},circle](v3){$t_2$}
 					+ (.8*6,.8*-3) node[label=-90:{\footnotesize$\bar p$},circle](v4){$t_3$}
 				;


 	\draw[-latex] (v0) --node[pos=.55,fill=white](){\footnotesize$\begin{array}{l}a,a\\a,b\end{array}$} (v1);
 	\draw[-latex] (v0) --node[pos=.55,fill=white](){\footnotesize$\begin{array}{l}b,a\end{array}$} (v3);
 	\draw[-latex] (v0) --node[pos=.55,fill=white](){\footnotesize$\begin{array}{l}b,b\end{array}$}  (v4);

 	\draw[-latex] (0) -- (v0);
 	%
 	\draw[-latex] (v1.70-90) .. controls +(60-90:4em) and +(120-90:4em) .. node[pos=.5,fill=white,right,xshift=.5ex](){\footnotesize$\ast,\ast$} (v1.110-90);


 	\draw[-latex] (v3.70-90) .. controls +(60-90:4em) and +(120-90:4em) .. node[pos=.5,fill=white,right,xshift=.5ex](){\footnotesize$\ast,\ast$} (v3.110-90);

 	\draw[-latex] (v4.70-90) .. controls +(60-90:4em) and +(120-90:4em) .. node[pos=.5,fill=white,right,xshift=.5ex](){\footnotesize$\ast,\ast$} (v4.110-90);

 	\end{tikzpicture}
 	}
 	\caption{\bphtext{Two games $G_2$ (left) and $G_3$ (right) based on~$M_2$ and~$M_3$, respectively, showing that run-based preferences may not be preserved across bisimilar systems.}}%
 	\label{fig:run_based_preferences_closed}%
	\label{fig:another_counterexample}
\end{figure}
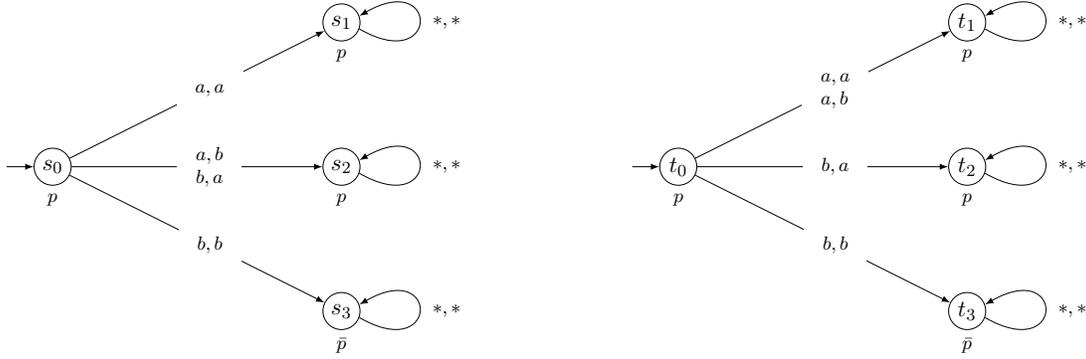

This preservation of preference type under bisimilarity, however, does
not extend to run-based preferences. To see this, consider
Figure~\ref{fig:another_counterexample} and let the goal
set~$\goalset_i$ of some player~$i$ be given by all
computations~$\computation=d_0,d_1,d_2,\dots$ with $d_0=(a,a)$.
Then,
obviously, $\goalset_i$ is run-based in the game~$G_2$ based on~$M_2$ on the left, but not in the game~$G_3$ based on~$M_3$ to the right.
%
To see the latter, consider any computation $\computation'=d'_0,d_1',d_2',\dots$ with $d'_0=(a,b)$.
Then,~$\computation'\notin\goalset_i$, but, nevertheless, in~$G_3$ we have $\run_{M_3}(\computation)=\run_{M_3}(\computation')$.
By contrast, the goal set given by all
computations~$\computation=d_0,d_1,d_2,\dots$ such that $d_0\neq(b,b)$ is
run-based in both games.

\subsection*{Computation-based Strategies}

If strategies are computation-based, players can have their actions
depend on virtually all information that is available in the
system. In an important sense full transparency prevails and different
actions can be chosen on bisimilar states provided that the
computations that led to them are different. Moreover, the strategies
available to players in bisimilar concurrent game structures are
identical. Thus we obtain our first main result.

\begin{thm}%
	\label{theorem:NE_computations}\label{app:theorem:NE_computations}
	Let $G=(M,\goalset_1,\dots,\goalset_n)$ and $G'=(M',\goalset_1,\dots,\goalset_n)$ be games on bisimilar concurrent game structures~$M$ and~$M'$, respectively, and let  $f=(f_1,\dots,f_n)$ be a computation-based profile.
	Then,~$f$ is a Nash equilibrium in computation-based strategies in~$G$ if and only if~$f$ is a Nash equilibrium in computation-based strategies in~$G'$.
\end{thm}
\begin{proof}
	First assume
	for contraposition that $f=(f_1,\dots,f_n)$ is \emph{not} a Nash equilibrium in computation-based strategies in~$M'$.
	Then, there is a player~$i$ and a strategy~$g_i$ for~$i$ in~$M'$ such that $\computation_{M'}(f)\notin\goalset_i$ and $\computation_{M'}( f_{-i},g_i)\in\goalset_i$.
	Observe that, as the computation-based strategies of players in bisimilar structures coincide (cf., Lemma~\ref{lemma:strategy_identity}),~$f$ is also a strategy profile in~$M$ and~$g_i$ a strategy for~$i$ in~$M$.
	Moreover, recall that the computations induced by the same strategy profile in different but bisimilar structures are identical (cf., second part of Lemma~\ref{lemma:strategy_identity}). This yields $\computation_{M}(f)=\computation_{M'}(f)$
	and $\computation_{M}(f_{-i},g_i)=\computation_{M'}(f_{-i},g_i)$.
	Accordingly, $\computation_{M}(f)\notin\goalset_i$ whereas
	$\computation_{M}(f_{-i},g_i)\in\goalset_i$.
	 We may conclude that~$f$ is not a computation-based equilibrium in~$M$ either. The opposite direction follows by an analogous argument.
\end{proof}


Theorem~\ref{theorem:NE_computations} holds for computation-based preferences. As run-based preferences and trace-based preferences  are computation-based preferences of a special kind, the result immediately extends to games in which the players' preferences are run-based in both games or trace-based preferences in both games.
As a consequence of Theorem~\ref{theorem:NE_computations}, moreover, we find that sustenance of runs and traces by computation-based equilibrium is also preserved under bisimilarity.

\begin{cor}%
 \label{corollary:sustained_invariance_under_bisimulation}
	Let $G=(M,\goalset_1,\dots,\goalset_n)$ and $G'=(M',\goalset_1,\dots,\goalset_n)$ be games on bisimilar concurrent game structures~$M$ and~$M'$, respectively,~$\computation\in\computations_{M}$, and~$\trace\in\traces_M$.
	Then,
	\begin{enumerate}
	\item\label{item:app:sustained_invariance_under_bisimulation_i}\label{item:sustained_invariance_under_bisimulation_i} $\computation$  is sustained by computation-based equilibrium in~$G$ if and only if $\computation$  is sustained by a computation-based equilibrium in~$G'$.
	\item\label{item:app:sustained_invariance_under_bisimulation_ii}\label{item:sustained_invariance_under_bisimulation_ii} $\trace$ is sustained by a computation-based equilibrium in~$G$ if and only if~$\trace$ is sustained by a computation-based equilibrium in~$G'$.
	\end{enumerate}
\end{cor}
\begin{proof}
Recall that by Lemma~\ref{lemma:trace_identity}, we have that $\computation\in\computations_{M'}$ and $\trace\in\traces_{M'}$.
For part~\ref{item:sustained_invariance_under_bisimulation_i}, let~$\computation_M(f)=\computation$, where~$f$ is a computation-based equilibrium in~$G'$. Then, by Theorem~\ref{theorem:NE_computations}, profile~$f$ is a computation-based equilibrium in~$G'$ as well. By virtue of Lemma~\ref{lemma:strategy_identity}, moreover,
 $\computation_M(f)=\computation_{M'}(f)$, which gives the result. The implication in the other direction follows by the same argument \emph{mutatis mutandis}.

The argument for part~\ref{item:sustained_invariance_under_bisimulation_ii} runs along analogous lines. First assume that~$\trace$ is sustained by computation-based equilibrium~$f$, that is, $\trace=\trace_M(f)$.
By Theorem~\ref{theorem:NE_computations}, we have that~$f$ is a computation-based equilibrium in~$G'$ as well.
 Now consider~$\computation_M(f)$. By Lemma~\ref{lemma:strategy_identity} then  $\computation_{M'}(f)=\computation_M(f)$. Lemma~\ref{lemma:computations_induce_bisimilar_runs} then yields $\trace_M(\computation_M(f))=\trace_{M'}(\computation_{M'}(f))$. It thus follows that~$\trace$ is sustained by~$f$, a computation based Nash equilibrium, in~$G'$.
The argument in the opposite direction is analogous, giving the result.
\end{proof}

\subsection*{Trace-based Strategies}


As we saw in Lemma~\ref{lemma:trace_identity}, the sets of (finite) traces of two bisimilar concurrent game structures coincide. Lemma~\ref{lemma:trace_tilde0} shows that the same holds for the trace-based strategies that are available to the players.
As a consequence,
 we can directly compare their trace-based Nash-equilibria. We find that, like computation-based equilibria, trace-based Nash equilibria are preserved in  CGS-games based on bisimilar concurrent game structures.
%


\begin{thm}%
	\label{app:theorem:NE_computations_traces}%
	\label{theorem:NE_computations_traces}
	Let $G=(M,\goalset_1,\dots,\goalset_n)$ and
	$G'=(M',\goalset_1,\dots,\goalset_n)$ be games on bisimilar
	concurrent game structures~$M$ and~$M'$, respectively, and
	$f=(f_1,\dots,f_n)$ be a trace-based strategy profile.
	Then,~$f$ is a Nash equilibrium in trace-based strategies in~$G$ if
	and only if~$f$ is a Nash equilibrium in trace-based strategies
	in~$G'$.
\end{thm}

\begin{proof}
	The proof is analogous to the one for Theorem~\ref{theorem:NE_computations} for the computation-based case.
	First assume
	for contraposition that $f=(f_1,\dots,f_n)$ is \emph{not} a Nash equilibrium in trace-based strategies in~$M'$.
	Then, there is a player~$i$ and a trace-based strategy~$g_i$ for~$i$ in~$M'$ such that $\computation(f)\notin\goalset_i$ and $\computation( f_{-i},g_i)\in\goalset_i$.
	Observe that, as the trace-based strategies of players in bisimilar structures coincide (cf., Lemma~\ref{lemma:strategy_identity}),
we have that~$f$ is also a trace-based strategy profile in~$M$ and~$g_i$ a trace-based strategy for~$i$ in~$M$.
	By the second part of Lemma~\ref{lemma:strategy_identity}, moreover, $\computation_{M}(f)=\computation_{M'}(f)$
	and $\computation_{M}(f_{-i},g_i)=\computation_{M'}(f_{-i},g_i)$.
	Accordingly, $\computation_{M}(f)\notin\goalset_i$ whereas
	$\computation_{M}(f_{-i},g_i)\in\goalset_i$.
	 We may conclude that~$f$ is not a trace-based equilibrium in~$M$ either. The opposite direction follows by an analogous argument.
%
\end{proof}


Like Theorem~\ref{theorem:NE_computations}, this result is for computation-based preferences in general, and as such immediately extends
to the case in which the players' preferences are stipulated to be run-based in both games or trace-based in both games. Theorem~\ref{theorem:NE_computations_traces} has further the following result as an immediate consequence, which is analogous to Corollary~\ref{corollary:sustained_invariance_under_bisimulation}.

\begin{cor}%
 \label{corollary:sustained_invariance_under_bisimulation_trace-based}
	Let $G=(M,\goalset_1,\dots,\goalset_n)$ and $G'=(M',\goalset_1,\dots,\goalset_n)$ be games on bisimilar concurrent game structures~$M$ and~$M'$, respectively,~$\computation\in\computations_{M}$, and~$\trace\in\traces_M$.
	Then,
	\begin{enumerate}
	\item\label{item:app:sustained_invariance_under_bisimulation_i_trace-based}\label{item:sustained_invariance_under_bisimulation_i_trace-based} $\computation$  is sustained by trace-based equilibrium in~$G$ if and only if $\computation$  is sustained by a trace-based equilibrium in~$G'$.
	\item\label{item:app:sustained_invariance_under_bisimulation_ii_trace-based}\label{item:sustained_invariance_under_bisimulation_ii_trace-based} $\trace$ is sustained by a trace-based equilibrium in~$G$ if and only if~$\trace$ is sustained by a trace-based equilibrium in~$G'$.
	\end{enumerate}
\end{cor}

\begin{proof}
	The proof is analogous to the one for Corollary~\ref{corollary:sustained_invariance_under_bisimulation}.
\end{proof}

\subsection*{Run-based Strategies}%
\label{subsection:run_based_no_preservation}

The positive results obtained using computation-based and trace-based
strategies are now followed by a negative result, already
mentioned in the introduction of the paper, which establishes that
Nash equilibria in run-based strategies---the most widely-used
strategy model in logic, computer science, and AI---are not preserved by
bisimilarity. Previously we observed that the players'
run-based strategies cannot straightforwardly be identified across two
different concurrent game structures, even if they are
bisimilar. We would therefore have to establish a correspondence
between the run-based strategies in the one game and the run-based
strategies in the other in an arguably ad hoc way. To cut this Gordian
knot, we therefore show in this section the stronger result that the
very \emph{existence} of run-based equilibria is not preserved under
bisimilarity.  That is, we can have two bisimilar concurrent game
structures, say~$M$ and~$M'$, on which we base two games $G$
and $G'$ with congruent preferences, such that~$G$ has a Nash equilibrium and~$G'$
does not.

\begin{thm}%
\label{thm:variance}
	The existence of run-based Nash equilibria is not preserved under bisimilarity. That is, there are games $(M,\goalset_1,\dots,\goalset_n)$
	and $(M',\goalset_1,\dots,\goalset_n)$ on bisimilar concurrent game structures~$M$ and~$M'$ such that a  Nash equilibrium in run-based strategies exists in~$G$ but not in~$G'$.
\end{thm}

To see that the above statement holds,
consider again the three-player game~$G_0$ on the concurrent game structure~$M_0$ in Figure~\ref{fig:cgswithne}. Assume, as before, that player~$1$'s goal set~$\goalset_1$ is given by those computations~$\computation$ such that $\trace_{M_0}(\computation)=v_0,v_1,v_2,\dots$, implies $p\in v_k$ for some $k\ge 0$. Similarly,~$\goalset_2$ consists of all computations~$\computation$ with $\trace_{M_0}(\computation)=v_0,v_1,v_2,\dots$ and $q\in v_k$ for some $k\ge 0$ and~$\goalset_3$ by those computations~$\computation$ with $\trace_{M_0}(\computation)=v_0,v_1,v_2,\dots$ and $v_k=\emptyset$ for all $k\ge0$. Recall that, consequently, the preferences of players~$1$, $2$, and~$3$ are trace-based and can be represented by the LTL formulas  $	\gamma_1=	\sometime p$,
$\gamma_2		=	\sometime q$, and
$ 	\gamma_3		=	\always \neg(p\vee q)$, respectively.

Let~$f^*_1$ and~$f^*_2$ be any run-based strategies for players~$1$ and~$2$ such that $f^*_1(s_0)=f^*_2(s_0)=a$. Let, furthermore, player~$3$'s run-based strategy~$f^*_3$ be such that
\[
\text{
	$f^*_3(s_0)		=	a$,\quad
	$f^*_3(s_0,s_1)		= 	a'$,\quad and \quad
	$f^*_3(s_0,s'_1)		=	b$.}
\]
Let $f^*=(f^*_1,f^*_2,f^*_3)$ and observe that $\run_{M_0}(f^*)=s_0,s_4,s_4,s_4,\dots$. Accordingly, player~$3$ has her goal achieved and does not want to deviate from~$f^*$. Players~$1$ and~$2$ do not have their goals achieved, but do not want to deviate from~$f^*$ either. To see this, let~$g_1$ be any run-based strategy for~$1$ such that $g_1(s_0)=b$; observe that this is required for any meaningful deviation from~$f^*$ by~$1$. Then $\run_{M_0}(g_1,f^*_2,f^*_3)=s_0,s_1,s_3,s_3,s_3,\dots$ or
$\run_{M_0}(g_1,f^*_2,f^*_3)=s_0,s_1,s_4,s_4,s_4,\dots$, depending on whether $f^*_2(s_0,s_1)=b$ or $f^*_2(s_0,s_1)=a$, respectively. In either case, player~$1$ does not get his goal achieved and it follows that he does not want to deviate from~$f^*$. An analogous argument---notice that the roles of player~$1$ and~$2$ are symmetric---shows that player~$2$ does not want to deviate from~$f^*$ either. We may thus conclude that~$f^*$ is a run-based equilibrium in~$G_0$.

Now, consider the game~$G_1$ on concurrent game structure~$M_1$ in Figure~\ref{fig:cgswithoutne} with the players' preferences as in~$M_0$. It is easy to check that~$M_0$ and~$M_1$ are bisimilar. Still, there is no run-based equilibrium in $G_1$. To see this, consider an arbitrary run-based strategy profile~$f=(f_1,f_2,f_3)$.
%
%
First, assume that $\run_{M_1}(f)=s_0,s_1,s_2,s_2,s_2,\dots$. Then, player~$1$ gets his goal achieved and  players~$2$ and~$3$ do not. If~$f_1(s_0,s_1)=a$ then~$f_3(s_0,s_1)=b$ and player~$3$ would like to deviate and play a strategy~$g_3$ with~$g_3(s_0,s_1)=a$. On the other hand, if $f_1(s_0,s_1)=b$, player~$3$ would like to deviate and play a strategy~$g_3$ with~$g_3(s_0,s_1)=b$.
%
%
Player~$3$ would similarly like to deviate from~$f$ if we assume that $\run_{M_1}(f)=s_0,s_1,s_3,s_3,s_3\dots$, in whose case it is player~$2$ who gets his goal achieved.
%
%
Now, assume that $\run_{M_1}(f)=s_0,s_1,s_4,s_4,s_4\dots$. In this case player~$3$ does get her goal achieved, but players~$1$ and~$2$ do not. However, player~$1$ would like to deviate and play a strategy~$g_1$ with $g_1(s_0,s_1)=b$ or $g_1(s_0,s_1)=a$, depending on whether $f_3(s_0,s_1)=a$ or~$f_3(s_0,s_1)=b$;
in a similar fashion, player~$2$ would like to deviate and play a strategy~$g_2$ with~$g_2(s_0,s_1)=b$ if $f_1(s_0,s_1)=a'$, and to one with~$g_2(s_0,s_1)=a$ if $f_1(s_0,s_1)=b'$.
%
%
Finally, assume that $\run_{M_1}(f)=s_0,s_4,s_4,s_4,\dots$. Then, neither player~$1$ nor player~$2$ gets his goal achieved. Now either $f_3(s_0,s_1)\in\set{a,b}$ or $f_3(s_0,s_1)\in\set{a',b'}$. If the former, player~$1$ would like to deviate and play a strategy~$g_1$ with $g_1(s_0)\neq f_1(s_0)$ and $g_1(s_0,s_1)\neq f_3(s_0,s_1)$. If the latter, player~$2$ would like to deviate and play a strategy~$g_2$ with $g_2(s_0)\neq f_2(s_0)$ and either $g_2(s_0,s_1)=b$ if $f_3(s_0,s_1)=a'$ or $g_2(s_0,s_1)=a$ if $f_3(s_0,s_1)=b'$.
We can then conclude that the CGS-game~$G_1$ does not have any run-based Nash equilibria.

The main idea behind this counter-example is that in~$G_0$ player~$3$
could distinguish which player deviates from~$f^*$ if the state
reached after the first round is not~$s_4$: if that state is~$s_1$, it
was player~$1$ who deviated, otherwise player~$2$. By choosing
either~$a'$ or~$b'$ at~$s_1$, and either~$a$ or~$b$ at~$s'_1$,
player~$3$ can guarantee that neither player~$1$ nor player~$2$ gets
his goal achieved (``punish'' them) and thus deter them from deviating
from~$f^*$. This possibility to punish deviations from~$f^*$ by players~$1$
and~$2$ in a single strategy is not available in the game on~$M_1$:
choosing from~$a$ and~$b$ will induce a deviation by player~$1$,
choosing from~$a'$ and~$b'$ one by player~$2$.

Observe that the games~$G_0$ and~$G_1$ do \emph{not} constitute a
counter-example against either the preservation under bisimilarity of
computation-based equilibria or the preservation of trace-based
equilibria. The reasons why such games fail to do so, however, are different. For the setting of computation-based strategies, player~$3$
can still distinguish and ``punish'' the deviating player in~$G_1$
as~$(a,b,a)$ and~$(b,a,a)$ are different directions and player~$3$ can
still have his action at~$s_1$ depend on which of these is played
at~$s_0$. By contrast, if we assume trace-based strategies, player~$3$
has to choose the same action at both~$s_1$ and~$s'_1$ in~$G_0$. As a
consequence, and contrarily to computation-based equilibria,
trace-based equilibria exist in neither~$G_0$ nor~$G_1$.
\bphnote{Please check this last remark.}
Also note that the goal sets~$\goalset_1$,~$\goalset_2$, and~$\goalset_3$ are run-based as well as computation-based both in~$G_1$ and~$G_2$. Accordingly, the counter-example also applies to settings wherein the players' preferences are assumed to be finer-grained in these two ways.

Observe at this point that~$s_1$ and~$s'_1$ are bisimilar states. Yet,
players are allowed to have run-based strategies
(which depend on state histories only) that choose \emph{different} actions
at bisimilar states. The above counter-example shows how this relative
richness of strategies 
makes a crucial difference. This raises the question as to whether we
actually want players to adopt run-based strategies in which they choose
different actions at bisimilar states. This observation leads us to the next
section.

\subsection*{Bisimulation-invariant Run-based Strategies}

Bisimilarity formally captures an informal concept of observational
indistinguishability on the part of an external observer of the
system. Now, the players in a concurrent game structure are in
essentially the same situation as an external observer if they are
assumed to be only able to observe the behaviour of the other players,
without knowing their internal structure or their interaction.

Drawing on this idea of indistinguishability, it is natural that
players cannot distinguish statewise bisimilar runs and, as a
consequence, can only adopt strategies that choose the same action at
runs that are statewise bisimilar. The situation is comparable to the
one in extensive games of imperfect information, in which players are
required to choose the same action in histories that are in the same
information set, that is, histories that cannot be distinguished (cf., {\em e.g.},~\cite{OR94,maschler_etal:2013a}).

To make this idea formally precise, we say that a run-based strategy
$f_i$ is \emph{bisimulation-invariant} if $f_i(\finrun)=f_i(\finrun')$
for all histories~$\finrun$ and~$\finrun'$ that are statewise
bisimilar.
The concept of Nash equilibrium is then similarly restricted to
bisimulation-invariant strategies. A profile
$\stratprof=(\strategy_1,\dots,\strategy_n)$ of
\emph{bisimulation-invariant} strategies is a \emph{Nash equilibrium
  in bisimulation-invariant strategies} (or a
\emph{bisimulation-invariant equilibrium}) in a game
$(M,\goalset_1,\dots,\goalset_n)$ if for all players~$i$ and every
\emph{bisimulation-invariant} strategy~$g_i$ for~$i$,
\[
	\trace(f_{-i},g_i)\in\goalset_i \ \text{implies} \ \trace(f)\in\goalset_i
\]
That is,~$f$ is a bisimulation-invariant equilibrium if no player wishes to deviate from~$f$ by playing a different bisimulation-invariant strategy.
In contrast to the situation for general run-based strategies, we
find that computations and traces that are sustained by a
bisimulation-invariant Nash equilibrium are preserved by
bisimulation. We show this by establishing a one-to-one correspondence
between the bisimulation-invariant strategies available to the players
in two bisimilar structures.

Based on the concept of state-wise bisimilarity, we associate with
every bisimulation-invariant strategy~$f_i$ for player~$i$ in
concurrent game structure~$M$, another bisimulation-invariant
strategy~$\tilde f_i$ for player~$i$ in any bisimilar concurrent game
structure~$M'$ such that for all $\finrun\in\finruns_{M'}$ and
$a\in\actions$,
\[
	\tilde f_i(\finrun) =a \quad \text{if}\ f_i(\finrun')=a \ \text{for some} \ \finrun'\in\finruns_M \ \text{with} \ \finrun\pairwisebisim\finrun'.
\]
Transitivity of~$\pairwisebisim$ guarantees that~$\tilde f_i$ is well
defined. To see this, observe that for all $\pi',\pi''\in\finruns_M$
with $\pi'\pairwisebisim \pi$ and $\pi''\pairwisebisim\pi$,
we have $\pi'\pairwisebisim\pi''$.
Having assumed that~$f_i$ is bisimulation-invariant, then
$f_i(\pi')=f_i(\pi'')$. By very much the same argument,~$\tilde f_i$
is 
bisimulation-invariant, if~$f_i$ is.

\begin{lem}%
	\label{lemma:tilde_properties_runs}
	Let~$M$ and~$M'$ be bisimilar concurrent game structures and let~$f_i$ be a bisimulation-invariant strategy for player~$i$ in~$M$. Then, $\tilde f_i$ is a bisimulation-invariant strategy in~$M'$.
\end{lem}

\begin{proof}
	Consider two statewise bisimilar runs~$\finrun,\finrun'\in\finruns_{M'}$, that is, $\finrun\pairwisebisim\finrun'$. Then, there are computations~$\fincomputation,\fincomputation'\in\fincomputations_{M'}$ such that $\finrun_{M'}(\fincomputation)=\finrun$ and $\finrun_{M'}(\fincomputation')=\finrun'$.
	By Lemma~\ref{lemma:computations_induce_bisimilar_runs}, we have  $\finrun_{M}(\fincomputation)\pairwisebisim\finrun_{M'}(\fincomputation)$ and $\finrun_{M}(\fincomputation')\pairwisebisim\finrun_{M'}(\fincomputation')$.
	Now, transitivity of~$\statewisebisim$ yields $\finrun_{M}(\fincomputation)\statewisebisim\finrun_{M}(\fincomputation')$.
 Having assumed that~$f_i$ is bisimulation-invariant, we obtain that
	$f_i(\finrun_{M}(\fincomputation))=f_i(\finrun_M(\fincomputation'))$. Finally, we may conclude that $\tilde f_i(\finrun)=\tilde f_i(\finrun')$, as desired.
\end{proof}



Moreover, it is easily appreciated that the mapping that transforms a strategy~$f_i$ into
 strategy $\tilde f_i$ is one-to-one. We will find that this is an essential property for bisimulation-invariant equilibria to be preserved under bisimilarity.

For a profile of bisimulation-invariant strategies $f=(f_1,\dots,f_n)$
in~$M$ we denote $\tilde f=(\tilde f_1,\dots,\tilde f_n)$.  We then
find that profiles~$f$ and $\tilde f$ of bisimulation-invariant
strategies induce identical computations.

\begin{lem}\label{lemma:tilde_runs} Let $M$ and~$M'$ be bisimilar concurrent game structures, $f=(f_1,\dots,f_n)$ a bisimulation-invariant strategy profile. Then,
$\computation_M(f)=\computation_{M'}(\tilde f)$.
\end{lem}
\begin{proof}
	Let $\computation_M(f)=d_0,d_1,d_2,\dots$ and $\computation_{M'}(\tilde f)=d'_0,d'_1,d'_2,\dots$. We prove by induction on~$k$ that $d_k=d'_k$ for every~$k\ge 0$. If $k=0$, we have $d_0=(f_1(s^0_M),\dots,f_n(s^0_M))$. Observe that, as $M\bisim M'$ also $s^0_M\bisim s^0_{M'}$ and, hence, $f_i(s^0_M)=\tilde f_i(s^0_{M'})$. Therefore,
	\[
		d'_0
		=
		(\tilde f_1(s^0_{M'}),\dots,\tilde f_n(s^0_{M'}))
		=
		(f_1(s^0_M),\dots,f_n(s^0_M))
		=
		d_0\text.
	\]
	For the induction step, we may assume that $d_0,\dots,d_k=d'_0,\dots,d_k'$.
	By Lemma~\ref{lemma:computations_induce_bisimilar_runs}, then $\finrun_M(d_0,\dots,d_k)\statewisebisim\finrun_{M'}(d_0',\dots,d_k')$. Accordingly, for every player~$i$ we have that $f_i(\finrun_M(d_0,\dots,d_k))=\tilde f_i(\finrun_{M'}(d'_0,\dots,d'_k))$.
It thus follows that
\begin{align*}
		d'_{k+1}
		&	=
		(\tilde f_1(\finrun_{M'}(d'_0,\dots,d_k')),\dots,\tilde f_n(\finrun_{M'}(d'_0,\dots,d_k')))\\
		&	=
		(f_1(\finrun_{M}(d_0,\dots,d_k)),\dots,f_n(\finrun_{M}(d_0,\dots,d_k)))\\
		&	=
		d_{k+1}\text.
\end{align*}
This concludes the proof.
\end{proof}

We are now in a position to state an equilibrium preservation theorem
for bisimulation-invariant strategies in a similar way as we were able
to obtain Theorem~\ref{theorem:NE_computations}, the analogous result
for computation-based and trace-based strategies.

\begin{thm}\label{theorem:NE_runs_pwbisim}
	Let $G=(M,\goalset_1,\dots,\goalset_n)$ and
	$G'=(M',\goalset_1,\dots,\goalset_n)$ be games on bisimilar
	concurrent game structures~$M$ and~$M'$, respectively.
	Then,~$f$ is a bisimulation-invariant equilibrium in~$G$ if and only if~$\tilde f$ is a bisimulation-invariant equilibrium in~$G'$.
\end{thm}
\begin{proof}
	First assume
	for contraposition  that $\tilde f=(\tilde f_1,\dots,\tilde f_n)$ is \emph{not} a Nash equilibrium in bisimulation-invariant strategies in~$G'$.
	Then, there is a player~$i$ and a bisimulation-invariant strategy~$g_i$ for~$i$ in~$M'$ such that $\computation_{M'}(\tilde f)\notin\goalset_i$ and $\computation_{M'}(\tilde f_{-i},g_i)\in\goalset_i$.
As the mapping that transforms a strategy~$f_i$ into
 strategy $\tilde f_i$ is one-to-one, there is a bisimulation-invariant strategy~$f'_i$ for~$i$ in~$M$ with $\tilde f'_i= g_i$.
	Accordingly, $\computation_{M'}(\tilde f_{-i},\tilde f'_i)\in\goalset_i$. Lemma~\ref{lemma:tilde_runs} then yields that $\computation_M(f_{-i},f'_i)\in\goalset_i$ and $\computation_M(f)\notin\goalset_i$. As~$f'_i$ is bisimulation-invariant, it follows that~$f$ is not an equilibrium in bisimulation-invariant strategies in~$G$.

The proof in the opposite direction runs along analogous lines.
\end{proof}

As an immediate corollary of Theorem~\ref{theorem:NE_runs_pwbisim}, it follows that the property of a computation or trace to be sustained by a bisimulation-invariant equilibria is also preserved  under bisimilarity.

\begin{cor}%
\label{app:corollary:sustained_invariance_under_bisimulation_run}%
\label{corollary:sustained_invariance_under_bisimulation_run}
	Let $G=(M,\goalset_1,\dots,\goalset_n)$ and $G'=(M',\goalset_1,\dots,\goalset_n)$ be  games on bisimilar concurrent game structures~$M$ and~$M'$, respectively.
	Then, for every computation $\computation\in\computations_M=\computations_{M'}$ and every trace $\trace\in\traces_{M}=\traces_{M'}$,
	\begin{enumerate}
	\item\label{item:app:sustained_invariance_under_bisimulation_run_i}\label{item:sustained_invariance_under_bisimulation_run_i} $\computation$  is sustained by a bisimulation invariant equilibrium in~$G$ if and only if $\computation$ is sustained by a  bisimulation-invariant equilibrium in~$G'$.
	\item\label{item:app:sustained_invariance_under_bisimulation_run_ii}%
	\label{item:sustained_invariance_under_bisimulation_run_ii} $\trace$ is sustained by a bisimulation-invariant equilibrium in~$G$ if and only if~$\trace$ is sustained by a bisimulation-invariant equilibrium in~$G'$.
	\end{enumerate}
\end{cor}
\begin{proof}
For part~\ref{item:sustained_invariance_under_bisimulation_run_i}, assume that~$\computation_M(f)=\computation$. Then in virtue of Lemma~\ref{lemma:tilde_runs} also $\computation_{M'}(\tilde f)=\computation$.
Moreover,
by Theorem~\ref{theorem:NE_runs_pwbisim} we have that profile~$f$ is a Nash equilibrium in bisimulation-invariant strategies in~$G$ if and only if~$\tilde f$ is a Nash equilibrium in bisimulation-invariant strategies in~$G'$, which gives the result.

The argument for part~\ref{item:sustained_invariance_under_bisimulation_run_ii} runs along analogous lines. First assume that~$\trace$ is sustained by bisimulation-invariant Nash equilibrium~$f$. Let~$\computation=\computation_{M}(f)$. Then, $\trace=\trace_{M}(\computation)$. By Theorem~\ref{theorem:NE_runs_pwbisim}, moreover,~$\tilde f$ is a bisimulation-invariant Nash equilibrium in~$G'$. An application of Lemma~\ref{lemma:tilde_runs} yields $\computation_M(\tilde f)=\computation_M(f)$. By Lemma~\ref{lemma:computations_induce_bisimilar_runs} then $\trace_M(\computation)=\trace_{M'}(\computation)$. It follows that $\trace$ is sustained by~$\tilde f$, a bisimulation-invariant Nash equilibrium, in~$G'$.
The argument in the opposite direction is analogous, giving the result.
\end{proof}

\section{Special Cases}
In the previous section we provided results about the preservation of
a given Nash equilibrium under bisimilarity, specifically, as long as we
do not consider run-based strategies or goals. In this
section we study two important special scenarios where
this is not the case.

Firstly, consider the scenario where we have two-player games with
run-based strategies and trace-based goals. This is an important
special case since run-based strategies, as we emphasised in the
introduction, are the ``conventional'' model of strategies used in
logics such as ATL$^*$ or SL, as well as in systems represented as concurrent
game structures. In particular, we show that with respect to
two-player games with run-based strategies and trace-based goals
(which include temporal logic goals), the
setting coincides with the one with bisimulation-invariant strategies
and trace-based goals, for which the preservation of Nash equilibria
under bisimilarity holds. A key observation in this case is that in
two-player games the existence of Nash equilibria can be characterised
in terms of the existence of certain winning strategies, which are
preserved across bisimilar systems.

Secondly, we also study the scenario where concurrent game structures are restricted to those that are induced by iterated Boolean games~\cite{GHW15} and Reactive Modules games~\cite{WGHMPT16}, two frameworks for the strategic analysis of AI and multi-agent systems, in particular, using model checking techniques.%
\footnote{For instance, Reactive Modules games provide a game semantics to formal specification languages such as Reactive Modules~\cite{AH99b}, which is widely used in model checking tools,
such as MOCHA~\cite{AHMQRT98} and PRISM~\cite{KNP11}.}
In this case, we show that bisimulation-invariant strategies also
coincide with run-based strategies, and therefore, that the positive
results for bisimulation-invariant strategies presented in the
previous section also transfer to this special case.

\subsection*{Two-Player Games}
This section concerns the preservation under bisimilarity of Nash equilibria under bisimulation in two-player games. We deal with the cases in which the players' preferences are computation-based, trace-based, and run-based separately.

\subsubsection*{Computation-based Preferences}
The counter-example against the preservation of the existence of Nash equilibria in Section~\ref{subsection:run_based_no_preservation} involved three players. We find that, if preferences are computation-based, the example can be adapted so as to involve only two players, which gives rise to the following result.

\begin{figure}[t]
\centering
 	\scalebox{.85}{
 	  \begin{tikzpicture}[scale=1]

 	  \tikzstyle{every ellipse node}=[draw,inner xsep=3.5em,inner ysep=1.2em,fill=black!15!white,draw=black!15!white]
 	  \tikzstyle{every circle node}=[fill=white,draw,minimum size=1.6em,inner sep=0pt]


 	  \draw(0,0) node(0){}
 	  			++ 	( 0:.9)	  	node[label=-90:{\footnotesize$\bar p\bar q$},circle](v0){$s_0$}
 				++	(  0:3)   	node[label=-90:{\footnotesize$\bar p\bar q$},circle](v1){$s_1'$}
					+( 90:3)  	node[label=-90:{\footnotesize$\bar p\bar q$},circle](w1){$s_1$}
 					++( 0:5)  	node[label=-90:{\footnotesize$\bar p     q$},circle](v3){$s_3$}
 					+( 90:3)  	node[label=-90:{\footnotesize$     p\bar q$},circle](v2){$s_2$}
					+(-90:4.5)	node[label=-90:{\footnotesize$\bar p\bar q$},circle](v4){$s_4$}
 				;

 	\draw[-latex] (v0) --node[inner xsep=0pt, inner ysep=0pt, pos=.5,fill=white](){
						\scalebox{.85}{$\begin{array}{l}
							\mwbactionstwo{c}{a}{\bthree}
						\end{array}$}} (v1);

	\draw[-latex] (v0) --node[inner xsep=0pt, inner ysep=0pt, pos=.5,fill=white](){
						\scalebox{.85}{$\begin{array}{l}
							\mwbactionstwo{b}{a}{\athree}
						\end{array}$}} (w1);

 	\draw[-latex] (v0) to[bend right=25]node[inner xsep=0pt, inner ysep=0pt, pos=.545,fill=white](){
						\scalebox{.85}{$\begin{array}{l}
							\mwbactionstwo{a}{a}{\athree}
						\end{array}$}} (v4);
 	\draw[-latex] (v1) --node[inner xsep=0pt, inner ysep=0pt, pos=.7,fill=white](){\scalebox{.85}{$\begin{array}{l}\mwbactionstwo{a}{b}{\ast}\\[-.2em]\mwbactionstwo{\bonetwooo}{\atwo}{\ast}\end{array}$}} (v2);
 	\draw[-latex] (v1) --node[inner xsep=0pt, inner ysep=0pt, pos=.7,fill=white](){\scalebox{.85}{$\begin{array}{l}\mwbactionstwo{a}{b'}{\bthree}\\[-.2em]\mwbactionstwo{b}{a'}{\athree}\end{array}$}} (v3);
 	\draw[-latex] (w1) --node[inner xsep=0pt, inner ysep=0pt, pos=.7,fill=white](){\scalebox{.85}{$\begin{array}{l}\mwbactionstwo{a}{b}{\ast}\\[-.2em]\mwbactionstwo{\bonetwooo}{\atwo}{\ast}\end{array}$}} (v2);
 	\draw[-latex] (w1) --node[inner xsep=0pt, inner ysep=0pt, pos=.7,fill=white](){\scalebox{.85}{$\begin{array}{l}\mwbactionstwo{a}{b'}{\bthree}\\[-.2em]\mwbactionstwo{b}{a'}{\athree}\end{array}$}} (v3);
	\draw[-latex] (v1) to[bend right=17.5]node[inner xsep=0pt, inner ysep=0pt, pos=.45,fill=white,fill opacity=1](){\scalebox{.85}{$\begin{array}{l}\mwbactionstwo{a}{a}{\ast}\\[-.2em]\mwbactionstwo{a}{a'}{\ast}\\[-.2em]\mwbactionstwo{b}{b}{\athree}\\[-.2em]\mwbactionstwo{b}{b'}{\bthree}\end{array}$}} (v4);
	\draw[-latex] (w1) --node[inner xsep=0pt, inner ysep=0pt, pos=.565,fill=white,fill opacity=1](){\scalebox{.85}{$\begin{array}{l}\mwbactionstwo{a}{a}{\ast}\\[-.2em]\mwbactionstwo{a}{a'}{\ast}\\[-.2em]\mwbactionstwo{b}{b}{\athree}\\[-.2em]\mwbactionstwo{b}{b'}{\bthree}\end{array}$}} (v4);
 	\draw[-latex] (0) -- (v0);

 	\draw[-latex] (v2.70-90) .. controls +(60-90:4em) and +(120-90:4em) .. node[pos=.5,fill=white,right,xshift=.5ex](){\footnotesize$\ast,\ast,\ast$} (v2.110-90);

 	\draw[-latex] (v3.70-90) .. controls +(60-90:4em) and +(120-90:4em) .. node[pos=.5,fill=white,right,xshift=.5ex](){\footnotesize$\ast,\ast,\ast$} (v3.110-90);

 	\draw[-latex] (v4.70-90) .. controls +(60-90:4em) and +(120-90:4em) .. node[pos=.5,fill=white,right,xshift=.5ex](){\footnotesize$\ast,\ast,\ast$} (v4.110-90);

 	\end{tikzpicture}
 	}
 	\caption{The concurrent game structure~$M_4$ underlying the game~$G_4$.}%
 	\label{fig:cgswithne_twoplayer}
\end{figure}
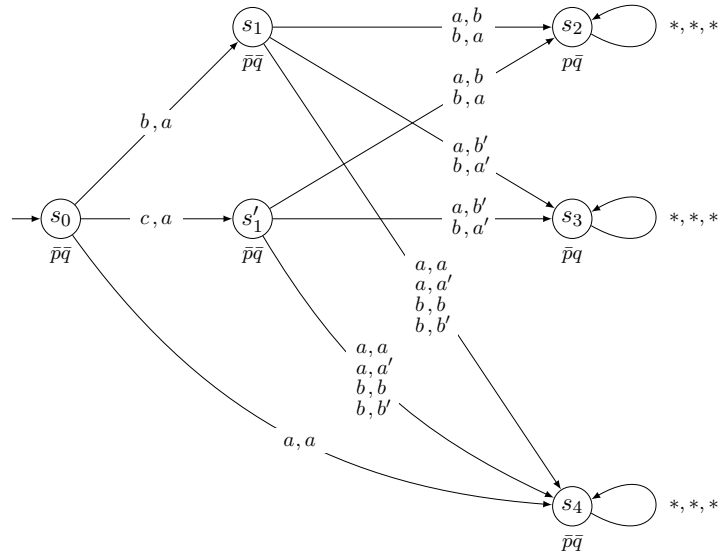

\begin{figure}[tb]
\centering
 	\scalebox{.85}{
 	  \begin{tikzpicture}[scale=1]

 	  \tikzstyle{every ellipse node}=[draw,inner xsep=3.5em,inner ysep=1.2em,fill=black!15!white,draw=black!15!white]
 	  \tikzstyle{every circle node}=[fill=white,draw,minimum size=1.6em,inner sep=0pt]


 	  \draw(0,0) node(0){}
 	  			++ 	( 0:.9)	  node[label=-90:{\footnotesize$\bar p\bar q$},circle](v0){$s_0$}
 				++	(  20:3.5)  node[label=-90:{\footnotesize$\bar p\bar q$},circle](v1){$s_1$}
 					+( 20:4)  node[label=-90:{\footnotesize$     p\bar q$},circle](v2){$s_2$}
 					+(-20:4)  node[label=-90:{\footnotesize$\bar p     q$},circle](v3){$s_3$}
 				;
 	 \draw(0,.9) ++(-27.5:9)	node[label=-90:{\footnotesize$\bar p\bar q$},circle](v4){$s_4$};

 	\draw[-latex] (v0) --node[inner xsep=0pt, inner ysep=0pt,pos=.5,fill=white](){\scalebox{.85}{$\begin{array}{l}\mwbactionstwo{b}{a}{\bthree}\\\mwbactionstwo{c}{a}{\athree}\end{array}$}} (v1);

 	\draw[-latex] (v1) --node[inner xsep=0pt, inner ysep=0pt,pos=.5,fill=white](){\scalebox{.85}{$\begin{array}{l}\mwbactionstwo{a}{b}{\ast}\\\mwbactionstwo{b}{a}{\ast}\end{array}$}} (v2);

 	\draw[-latex] (v1) --node[inner xsep=0pt, inner ysep=0pt,pos=.5,fill=white](){\scalebox{.85}{$\begin{array}{l}\mwbactionstwo{a}{b'}{\bthree}\\\mwbactionstwo{b}{a'}{\athree}\end{array}$}} (v3);

 	\draw[-latex] (v0) to[bend right=0]node[inner xsep=0pt, inner ysep=0pt,pos=.5,fill=white](){\scalebox{.85}{$\begin{array}{l}
							\mwbactionstwo{a}{a}{\athree}
						\end{array}$}} (v4);

	\draw[-latex] (v1) to[bend right=0]node[inner xsep=0pt, inner ysep=0pt,pos=.5,fill=white](){\scalebox{1}{\scalebox{.85}{$\begin{array}{l}\mwbactionstwo{a}{a}{\ast}\\[-.2em]\mwbactionstwo{a}{a'}{\ast}\\[-.2em]\mwbactionstwo{b}{b}{\athree}\\[-.2em]\mwbactionstwo{b}{b'}{\bthree}\end{array}$}}} (v4);

 	\draw[-latex] (0) -- (v0);

 	\draw[-latex] (v2.70-90) .. controls +(60-90:4em) and +(120-90:4em) .. node[pos=.5,fill=white,right,xshift=.5ex](){\footnotesize$\ast,\ast,\ast$} (v2.110-90);

 	\draw[-latex] (v3.70-90) .. controls +(60-90:4em) and +(120-90:4em) .. node[pos=.5,fill=white,right,xshift=.5ex](){\footnotesize$\ast,\ast,\ast$} (v3.110-90);

 	\draw[-latex] (v4.70-90) .. controls +(60-90:4em) and +(120-90:4em) .. node[pos=.5,fill=white,right,xshift=.5ex](){\footnotesize$\ast,\ast,\ast$} (v4.110-90);

 	\end{tikzpicture}
		}
 	\caption{The concurrent game structure~$M_5$ underlying the game~$G_5$.}%
 	\label{fig:cgswithoutne_2player}
\end{figure}
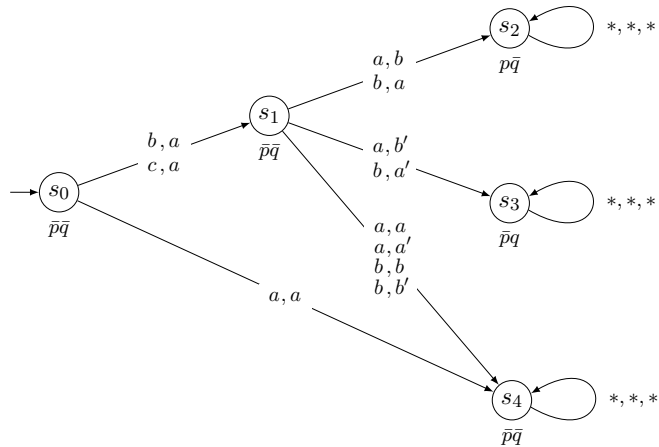

\begin{thm}%
\label{thm:variance_two-player}
	There are two-player games $(M,\goalset_1,\goalset_2)$
	and $(M',\goalset_1,\goalset_2)$ on bisimilar concurrent game structures~$M$ and~$M'$ with~$\goalset_1$ and~$\goalset_2$ computation-based such that a run-based Nash equilibrium exists in~$G$ but not in~$G'$.
\end{thm}
\begin{proof}
	Consider the concurrent game structures~$M_4$ and~$M_5$ depicted in Figures~\ref{fig:cgswithne_twoplayer} and~\ref{fig:cgswithoutne_2player}, respectively.
	Observe that
	$\actions_1(s_0)=\set{a,b,c}$ and $\actions_1(s)=\set{a,b}$ at all states~$s$ distinct from~$s_0$, and that
	$\actions_2(s_0)=\set{a}$ and $\actions_2(s)=\set{a,a',b,b'}$ at all states~$s$ distinct from~$s_0$.
We define the games $G_4=(M_4,\goalset_1,\goalset_2)$ and $G_5=(M_5,\goalset_1,\goalset_2)$ by letting
 $\goalset_1$ contain exactly those computations $\computation=d_0,d_1,d_2,\dots$
such that either $d_0=(b,a)$ and $d_1\in\set{(a,b),(b,a)}$ or $d_0=(c,a)$ and $d_1\in\set{(a,b'),(b,a')}$, and
letting~$\goalset_2$ consist precisely of those computations $\computation=d_0,d_1,d_2,\dots$ such that
$d_0=(a,a)$ or $d_1\in\set{(a,a),(a,a'),(b,b),(b,b')}$. By an argument analogous to that presented in Section~\ref{subsection:run_based_no_preservation}, it can then be appreciated that~$G_4$ has a run-based equilibrium, whereas~$G_5$ has not.

To see the former, observe that any run-based profile~$f^*=(f^*_1,f^*_2)$ will be an equilibrium if
$f^*_1(s_0)=f^*_2(s_0)=a$, $f^*_2(s_1)\in\set{a',b'}$, and $f^*_2(s_1')\in\set{a,b}$. For every strategy~$g_1$ for player~$1$, we then have $\computation_{M_4}(g_1,f^*_2)\notin\goalset_1$, whereas $\computation_{M_4}(f^*_1,f^*_2)\in\goalset_2$.

To see that~$G_5$ has no run-based equilibrium, first let $f=(f_1,f_2)$ be a run-based profile such that $\computation_{M_5}(f_1,f_2)\notin\goalset_1$.
Then, if $f_2(s_0,s_1)=a$, player~$1$ would like to deviate and play a strategy~$g_1$ with $g_1(s_0)=g_1(s_0,s_1)=b$; if $f_2(s_0,s_1)=a'$, to deviate and play a strategy~$g'_1$ with
$g'_1(s_0)=c$ and $g'_1(s_0,s_1)=b$; if~$f_2(s_0,s_1)=b$ to deviate and play a strategy~$g''_1$ with $g''_1(s_0)=b$ and $g''_1(s_0,s_1)=a$; and, finally, if $f_2(s_0,s_1)=b'$ to deviate and play a strategy~$g'''_1$ with $g'''_1(s_0)=c$ and $g'''_1(s_0,s_1)=a$. On the other hand, if~$\computation_{M_3}(f_1,f_2)\notin\goalset_2$, it must be the case that $f_1(s_0)\in\set{b,c}$. Observe, however, that player~$2$ would then like to deviate and play any strategy~$g_2$ with $g_2(s_0,s_1)=a$ if $f_1(s_0,s_1)=a$ and to a strategy~$g'_2$ with $g'_2(s_0,s_1)=b$ if $f_1(s_0,s_1)=b$. As, furthermore,~$\goalset_1$ and~$\goalset_2$ are disjoint, that is, the goals players~$1$ and~$2$ cannot simultaneously be satisfied, it follows that~$G_5$ has no run-based equilibria.
\end{proof}

\subsubsection*{Run-based Preferences}%
\label{subsection:runbased_strats_runbased_prefs}

We now address the preservation (of the existence) of Nash equilibria in two-player CGS-games where both preferences and strategies are run-based.
In contrast to our findings in the previous section, we find that,
under a natural closure restriction on the players' preferences, we are able to obtain a positive result.
Our proof relies on the equivalence of run-based profiles and run-invariant profiles as expounded in Section~\ref{section:trace_invariant}.

As already noted above, run-based strategies cannot generally be identified directly across  bisimilar CGS-games. The reason for this complication is that runs are sequences of states, and the sets of states of the two CGS-games need not coincide. In Section~\ref{section:trace_invariant}, we saw, however, how run-based strategies correspond to run-invariant strategies, which are computation-based by definition.  Lemma~\ref{lemma:computation_identity}, moreover, allows us to compare computation-based strategies, even if they may be run-invariant in the one model but not in the other.
%
%


Let~$f=(f_1,f_2)$ be a given run-invariant equilibrium  in a CGS-game $G=(M,\goalset_1,\goalset_2)$ based on~$M$ and let $G'=(M',\goalset_1,\goalset_2)$ be a CGS-game based on a concurrent game structure~$M'$ bisimilar to~$M$.
We  define
a (computation-based) profile~$f^K=(f^K_1,f^K_2)$ that is a run-invariant equilibrium in both~$G$ and~$G'$.
To prove that $f^K=(f^K_1,f^K_2)$ is a run-invariant equilibrium if $f=(f_1,f_2)$ is,
we exploit a characterisation of Nash equilibria in terms of winning strategies.%
\footnote{Winning strategies have also been used to characterise the existence of Nash equilibria in other two-player games with binary outcomes---see, {\em e.g.},~\cite{GHW15-concur,GW14}.}
We say that a run-invariant strategy~$f_i$ for player~$i$ is \emph{winning against player~$j$} whenever $\computation_M(f_i,f_j)\notin\goalset_j$ for all run-invariant strategies~$f_j$ of player~$j$. We then have the following lemma, which is independent of the type of preferences that players have.

	\begin{lem}\label{lemma:NE_winningstrats}
	Let $G=(M,\goalset_1,\goalset_2)$ be a game. Then, a profile
	$f=(f_1,f_2)$ is a run-invariant equilibrium if and only if
	both
	\begin{enumerate}
		\item\label{item:NE_winningstrats_i} $\computation_M(f_1,f_2)\notin\goalset_1$ implies that~$f_2$ is a winning strategy against player~$1$, and
		\item\label{item:NE_winningstrats_ii} $\computation_M(f_1,f_2)\notin\goalset_2$ implies that~$f_1$ is a winning strategy against player~$2$.
	\end{enumerate}
\end{lem}

\begin{proof}
	For the ``if'' direction assume for contraposition that~$f=(f_1,f_2)$ is not a run-invariant equilibrium. Then, either $\computation_M(f_1,f_2)\notin\goalset_1$ and $\computation_M(g_1,f_2)\in\goalset_1$ for some run-invariant strategy~$g_1$ for player~$1$, or $\computation_M(f_1,f_2)\notin\goalset_2$ and $\computation_M(f_1,g_2)\in\goalset_2$ for some run-invariant strategy~$g_2$ for player~$2$. If the former,~$f_2$ is not winning against player~$1$, refuting~\ref{item:NE_winningstrats_i}. If the latter,~$f_1$ is not winning  against player~$2$, which refutes~\ref{item:NE_winningstrats_ii}.

	The opposite direction is also by contraposition. Assume that either~\ref{item:NE_winningstrats_i} or~\ref{item:NE_winningstrats_ii} is not satisfied. Without loss of generality we may assume the former.
	Then, $\computation_M(f_1,f_2)\notin\goalset_1$ and~$f_2$ is not a winning strategy against player~$1$. Accordingly, there is some run-invariant strategy~$g_1$ for player~$1$ such that $\computation_{M}(g_1,f_2)\in\goalset_1$ and it follows that~$f=(f_1,f_2)$ is not a run-invariant equilibrium.
\end{proof}

In order to have a formally convenient characterisation of the goal sets $\goalset_1$ and~$\goalset_2$ to be run-based in two bisimilar CGS-games and to define the profile $f^K=(f^K_1,f^K_2)$, we furthermore introduce the following notations and auxiliary concepts.
For a concurrent game structure~$M$ and finite computations $\fincomputation,\fincomputation'\in\fincomputations_M$, we write $\computation\equiv_M\computation'$ if $\finrun_M(\computation)=\finrun_M(\computation)$.
Furthermore, we  say that finite computations~$\fincomputation$ and $\fincomputation'$
are \emph{finitely congruent in $M$ and $M'$}, in symbols $\fincomputation\equiv_{K_{M,M'}} \fincomputation'$, whenever there are (not necessarily distinct) \emph{intermediate computations} $\fincomputation_0,\dots,\fincomputation_m$ such that
\begin{enumerate}
\item $\mathwordbox[l]{\fincomputation  }{\fincomputation_j}=\fincomputation_0$,
\item $\mathwordbox[l]{\fincomputation' }{\fincomputation_j}=\fincomputation_m$, and
\item $\mathwordbox[l]{\fincomputation_j}{\fincomputation_j}\equiv_M \fincomputation_{j+1}$ or
$\fincomputation_j\equiv_{M'} \fincomputation_{j+1}$, for every $0\le j<m$.
\end{enumerate}
As $\equiv_M$ and $\equiv_{M'}$ are equivalence relations, we may assume that here
$\equiv_M$ and $\equiv_{M'}$ alternate and  $\computation_j\equiv_M\computation_{j+1}$ if $j$ is even, and $\computation_j\equiv_{M'}\computation_{j+1}$ if $j$ is odd.
We will generally omit the subscript in $K_{M,M'}$ when $M$ and $M'$ are clear from the context.
For an example, see again Figure~\ref{fig:run_based_preferences_closed}.
Consider the (one-step) computations $\fincomputation_1=(a,a)$, $\fincomputation_2=(a,b)$, $\fincomputation_3=(b,a)$, and
$\fincomputation_4=(b,b)$. Then, $\fincomputation_1\equiv_{M_2,M_3}\fincomputation_3$, because $\fincomputation_1\equiv_{M_3}\fincomputation_2$ and $\fincomputation_2\equiv_{M_2}\fincomputation_3$. On the other hand, some reflection reveals that $\fincomputation_1\not\equiv_{K_{M_1,M_2}}\fincomputation_4$.
It is worth noting that finite congruence of two computations implies statewise bisimilarity of the runs induced, that is, $\fincomputation\equiv_{K_{M,M'}}\fincomputation'$ implies
$\finrun_M(\fincomputation)\statewisebisim\finrun_M(\fincomputation')$.
\begin{lem}%
	\label{lemma:finite_congruence_statewise_bisim}
	Let~$M$ and~$M'$ be two bisimilar concurrent game structures and $\fincomputation=d_0,\dots,d_k$ and $\fincomputation'=d'_0,\dots,d'_k$. Then,
	$\fincomputation\equiv_{K_{M,M'}}\fincomputation'$ implies $\finrun_M(\fincomputation)\statewisebisim\finrun_{M}(\fincomputation')$.
\end{lem}

\begin{proof}
	Assume $\fincomputation\equiv_{K_{M,M'}}\fincomputation'$.
	Then there are $\fincomputation_0,\dots,\fincomputation_m$ such that
	$\fincomputation=\fincomputation_0$,
	$\fincomputation_m=\fincomputation'$, and, for all $0\le \ell<m$,
	$\finrun_M(\fincomputation_\ell)=\finrun_M(\fincomputation_{\ell+1})$ if~$\ell$ is even
	and
	$\finrun_{M'}(\fincomputation_\ell)=\finrun_{M'}(\fincomputation_{\ell+1})$ if~$\ell$ is odd.
	We assume that~$m$ is even; the case where~$m$ is odd follows by the same argument \emph{mutatis mutandis}.
	By virtue of Lemma~\ref{lemma:computations_induce_bisimilar_runs}-\ref{item:computations_induce_i},
we have $\finrun_M(\fincomputation_\ell)\statewisebisim\finrun_{M'}(\fincomputation_\ell)$ for every $0\le \ell<m$. Hence,
	\[
		\finrun_M(\fincomputation_0)
		=
		\finrun_{M}(\fincomputation_1)
		\statewisebisim
		\finrun_{M'}(\fincomputation_1)
		=
		\cdots
		=
		\finrun_{M'}(\fincomputation_{m-1})
		\statewisebisim
		\finrun_{M}(\fincomputation_{m-1})
		=
		\finrun_{M}(\fincomputation_{m})\text.
	\]
As obviously $\finrun_M(\fincomputation_\ell)=\finrun_M(\fincomputation_{\ell+1})$
and $\finrun_{M'}(\fincomputation_\ell)=\finrun_{M'}(\fincomputation_{\ell+1})$
imply, respectively,
$\finrun_M(\fincomputation_\ell)\statewisebisim\finrun_M(\fincomputation_{\ell+1})$
and $\finrun_{M'}(\fincomputation_\ell)\statewisebisim\finrun_{M'}(\fincomputation_{\ell+1})$,
also
	\[
		\finrun_M(\fincomputation_0)
		\statewisebisim
		\finrun_{M}(\fincomputation_1)
		\statewisebisim
		\finrun_{M'}(\fincomputation_1)
		\statewisebisim
		\cdots
		\statewisebisim
		\finrun_{M'}(\fincomputation_{m-1})
		\statewisebisim
		\finrun_{M}(\fincomputation_{m-1})
		\statewisebisim
		\finrun_{M}(\fincomputation_{m})\text.
	\]
By transitivity of~$\statewisebisim$ we may conclude that $\finrun_M(\fincomputation)\statewisebisim\finrun_M(\fincomputation)$.
\end{proof}

For bisimilar concurrent game structures~$M$ and~$M'$, we say that
a computation-based strategy~$f$ is \emph{${K_{M,M'}}$-invariant} if  $\fincomputation\equiv_{K_{M,M'}}\fincomputation'$ implies $f(\fincomputation)=f(\fincomputation')$, for all finite computations $\fincomputation,\fincomputation'\in\fincomputations_M$.
We find that $K_{M,M'}$-invariance exactly captures the concept of a strategy that is run-invariant in two  bisimilar concurrent game structures.

\begin{lem}
	Let~$M$ and~$M'$ be bisimilar concurrent game structures and~$f_i$ a computation-based strategy for player~$i$. Then,~$f_i$ is $K_{M,M'}$-invariant if and only if~$f_i$ is run-invariant in both~$M$ and~$M'$.
\end{lem}

\begin{proof}
	For the ``only if''-direction, assume that $f_i$ is $K_{M,M'}$-invariant and consider arbitrary $\fincomputation,\fincomputation'\in\fincomputations_M$ such that $\finrun_M(\fincomputation)=\finrun_M(\fincomputation')$, that is,
	$\fincomputation\equiv_M\fincomputation'$. By $K_{M,M'}$-invariance of $f_i$ then immediately
	$f_i(\fincomputation)=f_i(\fincomputation')$. Accordingly,~$f_i$ is run-invariant in~$M$. The argument for~$f_i$ being run-invariant in $M'$ is analogous.

	For the ``if''-direction, assume that~$f_i$ is run-invariant in both~$M$ and~$M'$, and consider arbitrary $\fincomputation,\fincomputation'\in\fincomputations$ such that
	$\fincomputation\equiv_K\fincomputation'$. Then, we may assume that there are $\fincomputation_0,\dots,\fincomputation_m$ such that
	\[
		\fincomputation=
		\fincomputation_0 \equiv_M
		\fincomputation_1 \equiv_{M'}
		\fincomputation_2 \equiv_M
		\dots
		\equiv_{M'}
		\fincomputation_m
		=
		\fincomputation'\text.
	\]
Having assumed that~$f_i$ is run-invariant in both~$M$ and~$M'$, then also
	\[
		f_i(\fincomputation)=
		f_i(\fincomputation_0) =
		f_i(\fincomputation_1) =
		f_i(\fincomputation_2) =
		\dots
		=
		f_i(\fincomputation_m)
		=
		f_i(\fincomputation')\text,
	\]
from which follows that~$f_i$ is $K_{M,M'}$-invariant.
\end{proof}

As we have argued in Section~\ref{tab:summary}, for the question whether the Nash equilibria across two bisimilar CGS-games are preserved to make sense, the players' preferences in the two games have to be congruent, that is, they have to be the same and of the same type in both games.
In this section we deal with run-based preferences. We have already seen in Section~\ref{section:congruence}, that identity of a player's computation-based preferences in two CGSs does not guarantee their being congruent as run-based preferences. By imposing an additional closedness restriction, however, we find that a computation-based goal set can be guaranteed to be run-based in two CGS-games based on bisimilar concurrent game structures~$M$ and~$M'$.
Accordingly, call a goal set $\goalset_i$ \emph{$K_{M,M'}$-closed} if for all computations~$\computation=d_0,d_1,d_2,\dots$ and~$\computation'=d'_0,d'_1,d_2',\dots$, we have that $\computation\in\goalset_i$
implies $\computation'\in\goalset_i$ whenever
$d_0,\dots,d_k\equiv_{K_{M,M'}} d_0',\dots,d_k'$ for all $k\ge 0$.

\begin{lem}%
	\label{lemma:K-closed_runbased_M_and_Mprime}
	Let $G=(M,\goalset_1,\dots,\goalset_n)$ and $G'=(M',\goalset_1,\dots,\goalset_n)$ be CGS-games on bisimilar concurrent game structures $M$ and $M'$, and~$\goalset_i$ is $K_{M,M'}$-closed for some player~$i$. Then,~$\goalset_i$ is run-based in both~$M$ and~$M'$.
\end{lem}

\begin{proof}
	Assume that $\goalset_i$ is $K_{M,M'}$-invariant and consider arbitrary infinite computations $\computation=d_0,d_1,d_2,\dots$ and $\computation'=d'_0,d'_1,d'_2,\dots$ such that $\run_M(\computation)=\run_M(\computation')$.
	Also assume that $\computation\in\goalset_i$.
	Let $\run_M(\computation)=s_0,s_1,s_2,\dots$ and $\run_M(\computation')=t_0,t_1,t_2,\dots$.
	Then, for every $k\ge 0$, we also have
	that
	\[
	\finrun_M(d_0,\dots,d_k)=s_0,\dots,s_{k+1}=t_0,\dots,t_{k+1}=\finrun_M(d'_0,\dots,d_k')\text.
	\]
	Having assumed $\goalset_i$ to be $K_{M,M'}$-invariant, it follows that $\computation'\in\goalset_i$, as desired.
\end{proof}


For the remainder, let $G=(M,\Gamma_1,\Gamma_2)$ and $G'=(M',\Gamma_1,\Gamma_2)$ be two two-player CGS-games  based on bisimilar concurrent game structures~$M$ and~$M'$ such that both~$\goalset_1$ and~$\goalset_2$ are $K_{M,M'}$-closed (and thus, in particular, run-based). We prove that if there is a run-invariant equilibrium in~$M$, then there is also a~$K$-invariant profile that is a run-invariant equilibrium in~$M$.
We construct for
a given strategy profile~$f=(f_1,f_2)$ that is run-invariant in $M$ a $K_{M,M'}$-invariant profile $f^K=(f^K_1,f^K_2)$ such that
	\begin{enumerate}[label=$(\roman*)$]
		\item
		$\computation_M(f_1,f_2)=\computation_M(f^K_1,f^K_2)$,
		\item if $f_1$ is a winning strategy against player~$2$, then so is $f^K_1$,
		\item if $f_2$ is a winning strategy against player~$1$, then so is $f^K_2$.
	\end{enumerate}
On basis of Lemma~\ref{lemma:NE_winningstrats} we may then conclude that~$f^K$ corresponds to a run-invariant equilibrium in~$M$. Having defined~$f^K$ formally as a computation-based profile, by Theorem~\ref{theorem:NE_computations} it follows that~$f^K$ is also a computation-based equilibrium in~$G'$. Finally, because~$f^K$ is $K_{M,M'}$-invariant, we know that it furthermore corresponds to a \emph{run-invariant} equilibrium in both~$M$ and~$M'$.

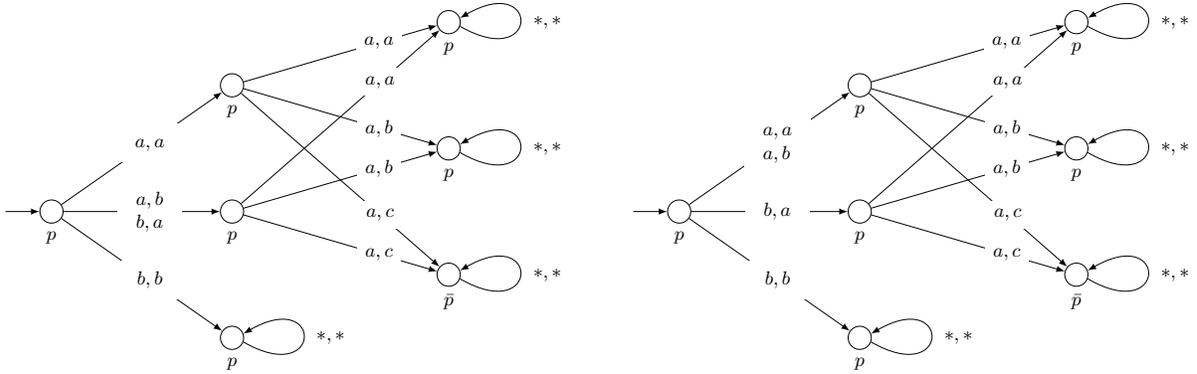
\begin{figure}[t]
\centering
 	\scalebox{.8}{
 	  \begin{tikzpicture}[scale=1]

	  \useasboundingbox[draw=white] (1,-3) rectangle (9.5,4);
 	  \tikzstyle{every ellipse node}=[draw,inner xsep=3.5em,inner ysep=1.2em,fill=black!15!white,draw=black!15!white]
 	  \tikzstyle{every circle node}=[fill=white,draw,minimum size=1em,inner sep=0pt]

	  \draw(0,-3)	node(){};
	  \draw(7.8,3.25)	node(){};

 	  \draw(0,0) node(0){}
 	  			++ 	( 0:.9)	  node[label=-90:{\footnotesize$p$},circle](v0){          $ $}
					+  (.5*6,.7*3)  node[label=-90:{\footnotesize$p$},circle](v1){     $ $}
					+  (.5*6,.7*0) node[label=-90:{\footnotesize$p$},circle](v3){      $ $}
					+  (.5*6,.7*-3) node[label=-90:{\footnotesize$p$},circle](v4){$ $}
					++ (.5*6,.35*3)node(vv0){      $ $}
	 					+  (.6*6,.7*3)  node[label=-90:{\footnotesize$p$},circle](vv1){     $ $}
						+  (.6*6,.7*0) node[label=-90:{\footnotesize$p$},circle](vv3){      $ $}
						+  (.6*6,.7*-3) node[label=-90:{\footnotesize$\bar p$},circle](vv4){$ $}
 				;

 	\draw[-latex] (v0) --node[pos=.55,fill=white](){\footnotesize$\begin{array}{l}a,a\end{array}$} (v1);
 	\draw[-latex] (v0) --node[pos=.55,fill=white](){\footnotesize$\begin{array}{l}a,b\\b,a\end{array}$} (v3);
 	\draw[-latex] (v0) --node[pos=.55,fill=white](){\footnotesize$\begin{array}{l}b,b\end{array}$}  (v4);

 	\draw[-latex] (v1) --node[pos=.7,fill=white](){\footnotesize$a,a$}      (vv1);
 	\draw[-latex] (v1) --node[pos=.7,fill=white](){\footnotesize$a,b$} (vv3);
 	\draw[-latex] (v1) --node[pos=.7,fill=white](){\footnotesize$a,c$}      (vv4);
 	\draw[-latex] (v3) --node[pos=.7,fill=white](){\footnotesize$a,a$}      (vv1);
 	\draw[-latex] (v3) --node[pos=.7,fill=white](){\footnotesize$a,b$} (vv3);
 	\draw[-latex] (v3) --node[pos=.7,fill=white](){\footnotesize$a,c$}      (vv4);

 	\draw[-latex] (0) -- (v0);

 	\draw[-latex] (vv1.70-90) .. controls +(60-90:4em) and +(120-90:4em) .. node[pos=.5,fill=white,right,xshift=.5ex](){\footnotesize$\ast,\ast$} (vv1.110-90);
 	\draw[-latex] (vv3.70-90) .. controls +(60-90:4em) and +(120-90:4em) .. node[pos=.5,fill=white,right,xshift=.5ex](){\footnotesize$\ast,\ast$} (vv3.110-90);
 	\draw[-latex] (vv4.70-90) .. controls +(60-90:4em) and +(120-90:4em) .. node[pos=.5,fill=white,right,xshift=.5ex](){\footnotesize$\ast,\ast$} (vv4.110-90);
 	\draw[-latex] (v4.70-90) .. controls +(60-90:4em) and +(120-90:4em) .. node[pos=.5,fill=white,right,xshift=.5ex](){\footnotesize$\ast,\ast$} (v4.110-90);
 	\end{tikzpicture}
 	}
 	\scalebox{.8}{
 	  \begin{tikzpicture}[scale=1]
	  \useasboundingbox[draw=white] (-.5,-3) rectangle (8,4);

 	  \tikzstyle{every ellipse node}=[draw,inner xsep=3.5em,inner ysep=1.2em,fill=black!15!white,draw=black!15!white]
 	  \tikzstyle{every circle node}=[fill=white,draw,minimum size=1em,inner sep=0pt]

	  \draw(0,-3)	node(){};
	  \draw(7.8,3.25)	node(){};

 	  \draw(0,0) node(0){}
 	  			++ 	( 0:.9)	  node[label=-90:{\footnotesize$p$},circle](v0){          $ $}
					+  (.5*6,.7*3)  node[label=-90:{\footnotesize$p$},circle](v1){     $ $}
					+  (.5*6,.7*0) node[label=-90:{\footnotesize$p$},circle](v3){      $ $}
					+  (.5*6,.7*-3) node[label=-90:{\footnotesize$p$},circle](v4){$ $}
					++ (.5*6,.35*3)node(vv0){      $ $}
	 					+  (.6*6,.7*3)  node[label=-90:{\footnotesize$p$},circle](vv1){     $ $}
						+  (.6*6,.7*0) node[label=-90:{\footnotesize$p$},circle](vv3){      $ $}
						+  (.6*6,.7*-3) node[label=-90:{\footnotesize$\bar p$},circle](vv4){$ $}
 				;

 	\draw[-latex] (v0) --node[pos=.55,fill=white](){\footnotesize$\begin{array}{l}a,a\\a,b\end{array}$} (v1);
 	\draw[-latex] (v0) --node[pos=.55,fill=white](){\footnotesize$\begin{array}{l}b,a\end{array}$} (v3);
 	\draw[-latex] (v0) --node[pos=.55,fill=white](){\footnotesize$\begin{array}{l}b,b\end{array}$}  (v4);

 	\draw[-latex] (v1) --node[pos=.7,fill=white](){\footnotesize$a,a$}      (vv1);
 	\draw[-latex] (v1) --node[pos=.7,fill=white](){\footnotesize$a,b$} (vv3);
 	\draw[-latex] (v1) --node[pos=.7,fill=white](){\footnotesize$a,c$}      (vv4);
 	\draw[-latex] (v3) --node[pos=.7,fill=white](){\footnotesize$a,a$}      (vv1);
 	\draw[-latex] (v3) --node[pos=.7,fill=white](){\footnotesize$a,b$} (vv3);
 	\draw[-latex] (v3) --node[pos=.7,fill=white](){\footnotesize$a,c$}      (vv4);

 	\draw[-latex] (0) -- (v0);

 	\draw[-latex] (vv1.70-90) .. controls +(60-90:4em) and +(120-90:4em) .. node[pos=.5,fill=white,right,xshift=.5ex](){\footnotesize$\ast,\ast$} (vv1.110-90);
 	\draw[-latex] (vv3.70-90) .. controls +(60-90:4em) and +(120-90:4em) .. node[pos=.5,fill=white,right,xshift=.5ex](){\footnotesize$\ast,\ast$} (vv3.110-90);
 	\draw[-latex] (vv4.70-90) .. controls +(60-90:4em) and +(120-90:4em) .. node[pos=.5,fill=white,right,xshift=.5ex](){\footnotesize$\ast,\ast$} (vv4.110-90);
 	\draw[-latex] (v4.70-90) .. controls +(60-90:4em) and +(120-90:4em) .. node[pos=.5,fill=white,right,xshift=.5ex](){\footnotesize$\ast,\ast$} (v4.110-90);
 	\end{tikzpicture}
 	}

 	\caption{\bphtext{Two games $G_6$ (left) and $G_7$ (right) based on~$M_6$ and~$M_7$, respectively.}}%
 	\label{fig:run_invariant_construction}
\end{figure}

For an example consider the games~$G_6$ and~$G_7$ depicted in Figure~\ref{fig:run_invariant_construction}.
The underlying concurrent game structures $M_6$ and $M_7$ only differ with respect to direction~$(a,b)$ at the
initial state and their bisimilarity is easily established.
Assume that the goal of player~$1$ is to see~$p$ false at some point in the future, that is,
\[
	\goalset_1=\set{d_0,d_1,d_2,\dots\in\computations\midd \text{$d_0\in\set{(a,a),(a,b),(b,a)}$ and $d_1=(a,c)$} }\text,
\]
and that player~$2$ tries to prevent this, that is, $\Gamma_2=\computations\setminus\Gamma_1$.
Observe that defined thus, the players' preferences are run-based in both~$G_6$ and~$G_7$.
Concentrating on~$G_6$ first, define the run-invariant strategy profile~$f=(f_1,f_2)$, such that,
\begin{align*}
	f_1(\epsilon)	&	=	b	&
	f_1(a,a)		&	=	a	&
	f_1(a,b)		&	=	a	&
	f_1(b,a) 		&	=	a\\
	f_2(\epsilon)	&	=	b	&
	f_2(a,a)		&	=	a	&
	f_2(a,b)		&	=	b	&
	f_2(b,a) 		&	=	b
\end{align*}
We find that~$f=(f_1,f_2)$ is a run-invariant equilibrium in~$G_6$.
Observe, however, that $f=(f_1,f_2)$ is \emph{not} run-invariant in~$G_7$, as $f_2(a,a)\neq f_2(a,b)$ even though
$\finrun_{M_7}(a,a)=\finrun_{M_7}(a,b)$.
Accordingly, $f=(f_1,f_2)$ fails as a $K_{M_6,M_7}$-invariant equilibrium.
Let $g_2$ be defined such that,
\begin{align*}
	g_2(\epsilon)	&	=	b	&
	g_2(a,a)		&	=	b	&
	g_2(a,b)		&	=	b	&
	g_2(b,a) 		&	=	b
\end{align*}
Then, $(f_1,g_2)$ is readily seen to be a $K_{M_6,M_7}$-invariant equilibrium in $G_6$.
We will find that under the conditions specified, this is no coincidence and that $K$-invariant equilibria can be
constructed from run-invariant equilibria in a systematic fashion.

We first define the strategy profile~$f^K=(f_1^K,f^K_2)$.
The underlying idea is to carefully choose for each finite computation $d_0,\dots,d_k$
computations
$\fincomputation^{d_0,\dots,d_k}_{f_1}$
and
$\fincomputation^{d_0,\dots,d_k}_{f_2}$
from the equivalence class of $d_0,\dots,d_k$  under $\equiv_{K_{M,M'}}$, and set
$f_1^K(d_0,\dots,d_k)=f_1(\fincomputation^{d_0,\dots,d_k}_{f_1})$
and
$f_2^K(d_0,\dots,d_k)=f_2(\fincomputation^{d_0,\dots,d_k}_{f_2})$, respectively.
This guarantees that $f^K=(f^K_1,f^K_2)$ is $K_{M,M'}$-invariant.
Here we give priority to prefixes of $\computation(f_1,f_2)$, that is if
$\computation(f_1,f_2)=d'_0,d'_1,d'_2,\dots$ and $d'_0,\dots,d'_k\equiv_{K_{M,M'}}d_0,\dots,d_k$,
then $\fincomputation_{f_1}^{d_0,\dots,d_k}=\fincomputation_{f_2}^{d_0,\dots,d_k}=d'_0\dots,d_k'$.
In a similar way, we also choose $\fincomputation^{d_0,\dots,d_k}_{f_1}$
and
$\fincomputation^{d_0,\dots,d_k}_{f_2}$
so as to preserve the two players' punishment strategies.
This guarantees that $f^K=(f^K_1,f^K_2)$ is an equilibrium, as it inherits this property from~$f=(f_1,f_2)$.\footnote{The precise definition is rather involved, and the reader may want to skip to page~\pageref{theorem:runbased_runbased_1}, where the main theorem of the section is stated and proven.}
To do so we assume a well-ordering of the action sets~$\actions_1$ and $\actions_2$ of players~$1$ and~$2$, respectively. 		Then, for all finite computations $\fincomputation=d_0,\dots,d_k$ in~$\fincomputations_M$, we define inductively and simultaneously computations $\fincomputation_{f_1}^{d_0,\dots,d_k}$
and
		$\fincomputation_{f_2}^{d_0,\dots,d_k}$ as follows.
For $\fincomputation=\epsilon$ we have
$
		\fincomputation^\epsilon_{f_1}							=	\fincomputation^\epsilon_{f_2}	=\epsilon,
$
		and, for $\fincomputation=d_0,\dots,d_{k+1}$,
		\begin{align*}
		\fincomputation_{f_1}^{d_0,\dots,d_k,d_{k+1}}		& 	=	d'_0,\dots,d'_k,(x_1,x_2)
		\\
		\fincomputation_{f_2}^{d_0,\dots,d_k,d_{k+1}}		&	=	d''_0,\dots,d''_k,(y_1,y_2),
		\end{align*}
		where
		$\fincomputation_{f_1}^{d_0,\dots,d_k}=d'_0,\dots,d'_k$,
		$\fincomputation_{f_2}^{d_0,\dots,d_k}=d''_0,\dots,d''_k$, and
		\begin{enumerate}[label=$(\roman*)$]
		\renewcommand{\itemsep}{2ex}
			\item[$(i.1)$] $x_1=f_1(d'_0,\dots,d'_k)$ and $x_2=f_2(d'_0,\dots,d'_k)$, if
			\[d'_0,\dots,d'_k,(f_1(d'_0,\dots,d'_k),f_2(d'_0,\dots,d'_k))\equiv_{K}d_0,\dots,d_k,d_{k+1}\text,
			\]
			\item[$(i.2)$] $y_1=f_1(d''_0,\dots,d''_k)$ and $y_2=f_2(d''_0,\dots,d''_k)$, if
			\[d''_0,\dots,d''_k,(f_1(d''_0,\dots,d''_k),f_2(d''_0,\dots,d''_k))\equiv_{K}d_0,\dots,d_k,d_{k+1}\text,
			\]
			\item[$(ii.1)$] $x_1=f_1({d'_0,\dots,d'_k})$ and $x_2$ is the least action available to player~$2$ such that
			\[
			d'_0,\dots,d'_k,(x_1,x_2)\equiv_{K}d_0,\dots,d_k,d_{k+1},
			\]
			if such an action~$x_2$ exists and $(i.1)$ does not apply,
			\item[$(ii.2)$] $y_2=f_2({d''_0,\dots,d''_k})$ and $y_1$ is the least action available to player~$1$ such that
			\[
			d''_0,\dots,d''_k,(y_1,y_2)\equiv_{K}d_0,\dots,d_k,d_{k+1},
			\]
			if such an action~$y_1$ exists and $(i.2)$ does not apply,
			\item[$(iii.1)$] $x_1$ and $x_2$ are the least actions available to players~$1$ and~$2$, respectively,  such that $d'_0,\dots,d'_k,(x_1,x_2)\equiv_{K}d_0,\dots,d_k,d_{k+1}$, if neither $(i.1)$ nor $(ii.1)$ apply,
			\item[$(iii.2)$] $y_1$ and $y_2$ are the least actions available to players~$1$ and~$2$, respectively,  such that $d''_0,\dots,d''_k,(x_1,x_2)\equiv_{K}d_0,\dots,d_k,d_{k+1}$, if neither $(i.2)$ nor $(ii.2)$ apply,
		\end{enumerate}
Observe that the actions~$x_1$ and~$x_2$ in the definition above always exist. The reason for this is that, if $d_0,\dots,d_k\equiv_{K_{M,M'}}d'_0,\dots,d'_k$,  by Lemma~\ref{lemma:finite_congruence_statewise_bisim} also $\finrun_M(d_0,\dots,d_k)\statewisebisim\finrun_M(d'_0,\dots,d'_k)$.
Let $\finrun(d_0,\dots,d_k)=s_0,\dots,s_k$,
$\finrun(d'_0,\dots,d'_k)=s'_0,\dots,s'_k$, and
$d'_{k+1}=(a_1,a_2)$.
Then, obviously, $a_1$ and $a_2$ are available to players~$1$ and~$2$ respectively at state~$s'_k$. As~$s_k\bisim s'_k$, that should also be the case at~$s_k$. A similar argument applies to the case where $d_0,\dots,d_k\equiv_{K_{M,M'}}d''_0,\dots,d''_k$.

We now define strategies~$f^K_1$ and~$f^K_2$ such that, for all finite computations~$\fincomputation=d_0,\dots,d_k$ in~$\fincomputations_M$,
	\begin{align*}
		f^K_1(d_0,\dots,d_k)	&= f_1(\fincomputation_{f_1}^{d_0,\dots,d_k})
		&\text{and}&&
		f^K_2(d_0,\dots,d_k)	&=f_2(\fincomputation_{f_2}^{d_0,\dots,d_k})\text.
	\end{align*}
To illustrate how $f^K=(f_1^K,f_2^K)$ is constructed from $f=(f_1,f_2)$, recall the games $M_6$ and $G_7$
in Figure~\ref{fig:run_invariant_construction}, and assume that actions for both players are ordered alphabetically.
For the empty computation~$\epsilon$, we have
\begin{align*}
	f^K_1(\epsilon) &	=	f_1(\fincomputation_{f_1}^\epsilon) =	f_1(\epsilon) = b	& \text{and}&&
	f^K_2(\epsilon) &	=	f_2(\fincomputation_{f_2}^\epsilon) =	f_2(\epsilon) = b
\end{align*}
Now consider the finite computations of length one.
For computation $(b,b)$, we find that $(b,b)\equiv_{K_{M_6,M_7}}(f_1(\fincomputation^\epsilon_{f_1}),f_2(\fincomputation^\epsilon_{f_1}))$.
Now case $(i.1)$ applies and we obtain $\fincomputation_{f_1}^{(b,b)}=(b,b)$.
Similarly, $\fincomputation_{f_2}^{(b,b)}=(b,b)$.
For the other one-step computations~$(x_1,x_2)\in\set{(a,a),(a,b),(b,a)}$, we find that
$(x_1,x_2)\not\equiv_{K_{M_6,M_7}}(f_1(\epsilon),f_2(\epsilon))$.
First consider $(a,a)$, and to determine $\fincomputation^{(a,a)}_{f_1}$, first
observe that
$\fincomputation^{\epsilon}_{f_1}=d'_0,\dots,d'_k=\epsilon$,
$d_0,\dots,d_k=\epsilon$, and $d_{k+1}=(a,a)$.
Now, for $x_1=f_1(\epsilon)=b$ we have
that
\[
	(x_1,a)
	=				(b,a)
	\equiv_{M_6} 	(a,b)
	\equiv_{M_7} 	(a,a)\text.
\]
Accordingly, $(x_1,a)\equiv_{M_6,M_7}(a,a)=d_{k+1}$, and case $(ii.1)$ applies.
With~$a$ being moreover player~$2$'s alphabetically least action, we may therefore conclude that
$
	\fincomputation_{f_1}^{(a,a)}	=
	(b,a)\text.
$
In a similar way we obtain
$
\fincomputation_{f_1}^{(a,b)}	=
	\fincomputation_{f_1}^{(b,a)} 	=
	(b,a)
$
as well as
$
	\fincomputation_{f_2}^{(a,a)}	=
	\fincomputation_{f_2}^{(a,b)}	=
	\fincomputation_{f_2}^{(b,a)} 	=
	(a,b)\text.
$
Hence,
\[
\begin{array}{l@{\;\;=\;\;}l@{\;\;=\;\;}l@{\;\;=\;\;}l@{\quad\quad\quad}l@{\;\;=\;\;}l@{\;\;=\;\;}l@{\;\;=\;\;}l}
	f_1^K(a,a)	   &	 f_1(\fincomputation_{f_1}^{(a,a)})	& f_1(b,a)	&	a	&		f_2^K(a,a)	&		f_2(\fincomputation_{f_2}^{(a,a)})		& f_2(a,b)	&	b	\\[8pt]
	f_1^K(a,b)	   &	 f_1(\fincomputation_{f_1}^{(a,b)})	& f_1(b,a)	&	a	&		f_2^K(a,b)	&		f_2(\fincomputation_{f_2}^{(a,b)})		& f_2(a,b)	&	b	\\[8pt]
	f_1^K(b,a)	   &	 f_1(\fincomputation_{f_1}^{(b,a)})	& f_1(b,a)	&	a	&		f_2^K(b,a)	&		f_2(\fincomputation_{f_2}^{(b,a)})		& f_2(a,b)	&	b	\\[8pt]
	f_1^K(b,b)	   &	 f_1(\fincomputation_{f_1}^{(b,b)})	& f_1(b,b)	&	b	&		f_2^K(b,b)	&		f_2(\fincomputation_{f_2}^{(b,b)})		& f_2(b,b)	&	b
\end{array}
\]
We thus find that $f^K=(f_1^K,f_2^K)$ coincides with the $K_{M_6,M_7}$-invariant equilibrium~$(f_1,g_2)$
that we identified above.

The definition of strategies~$f_1^K$ and $f_2^K$ ensures that $f^K=(f^K_1,f^K_2)$ is ${K}_{M,M'}$-invariant.
\begin{lem}%
	\label{lemma:f^K_Kinvariant}
	Let $\fincomputation=d_0,\dots,d_k$ be s finite computation in $\fincomputations_M$. Then, for $i=1,2$,
	\[
		d_0,\dots,d_k\mathrel{\equiv_{K}} \fincomputation_{f_i}^{d_0,\dots,d_k}\text.
	\]
	Accordingly, $f^K=(f^K_1,f^K_2)$ is ${K}_{M,M'}$-invariant.
\end{lem}

\begin{proof}
	Strategies~$f^K_1$ and~$f^K_2$ have been defined so as to be $K$-invariant. The claim then follows by induction on the length of~$\fincomputation$.
Let~$i=1$; the case for $i=2$ is analogous.
If the length of~$\fincomputation$ is~$0$, we have~$\fincomputation=\fincomputation_{f_1}^\epsilon=\epsilon$, and it immediately follows that $\fincomputation\equiv_K\fincomputation_{f_1}^\epsilon$.
For the induction step, let $\fincomputation=d_0,\dots,d_k,d_{k+1}$ and assume that
$d_0,\dots,d_k\mathrel{\equiv_{K}} \fincomputation_{f_1}^{d_0,\dots,d_k}$.
Let $\fincomputation_{f_1}^{d_0,\dots,d_{k+1}}$ be denoted by $d'_0,\dots,d'_{k},d'_{k+1}$ where $d'_{k+1}=(x_1,x_2)$.
There are three possibilities.

First, assume that $d'_0,\dots,d'_k,(f_1(d'_0,\dots,d'_k),f_2(d'_0,\dots,d'_k))\equiv_{K}d_0,\dots,d_k,d_{k+1}$.
Then, clause~$(i.1)$ applies and we have $x_1=f_1(d'_0,\dots,d'_k)$ and $x_2=f_2(d'_0,\dots,d'_k)$.
It then follows that $d'_0,\dots,d_k',d_{k+1}'\equiv_K d_0,\dots,d_k,d_{k+1}$.
Second, assume that
\[d'_0,\dots,d'_k,(f_1(d'_0,\dots,d'_k),y_2)\equiv_{K}d_0,\dots,d_k,d_{k+1}\] for some action~$y_2\in\actions_2$, but
\[d'_0,\dots,d'_k,(f_1(d'_0,\dots,d'_k),f_2(d'_0,\dots,d'_k))\not\equiv_{K}d_0,\dots,d_k,d_{k+1}.\]
Then clause $(ii.1)$ applies and we have that $x_1=f_1(d'_0,\dots,d'_k)$ and $x_2$ is the least action available to player~$2$ such that $d'_0,\dots,d'_k,(f_1(d'_0,\dots,d'_k),y_2)\equiv_{K}d_0,\dots,d_k,d_{k+1}$.
Again, it immediately follows that $d'_0,\dots,d_k',d_{k+1}'\equiv_K d_0,\dots,d_k,d_{k+1}$.

Finally, assume that neither of the above, then clause~$(iii.1)$ applies and we have that~$x_1$ and~$x_2$ are the least actions available to players~$1$ and~$2$  such that $d'_0,\dots,d'_k,(x_1,x_2)\equiv_{K}d_0,\dots,d_{k+1}$. As such actions exist, we conclude that $d'_0,\dots,d_k',d_{k+1}'\equiv_K d_0,\dots,d_k,d_{k+1}$.
\end{proof}

We are now in a position to prove the following crucial lemmas.\footnote{The proofs of these two lemmas extend over a couple of pages and the reader may skip to page~\pageref{page:cont_1}, where the running text continues.}

\begin{lem}\label{lemma:soundness_fK_1}\label{item:soundness_fK_i}
	Let $f=(f_1,f_2)$ be a run-invariant strategy profile
	for concurrent game structure~$M$ and~$f^K$ defined as above. Then,
	$\fincomputation_M(f_1,f_2)=\fincomputation_M(f_1^K,f_2^K)$.
\end{lem}

\begin{proof}
	Let us use the following notations:
$\fincomputation_M(f_1,f_2)					=	d_0,d_1,\dots$;
$\fincomputation(f^K_1, f^K_2)				=	d^K_0,d^K_1,\dots$;
$\fincomputation_{f_1}^{d_0,\dots,d_k}		=	d_0',d_1',\dots$;
$\fincomputation_{f_2}^{d_0,\dots,d_k}		=	d_0'',d''_1,\dots$;
		where, for every $k\ge 0$,
		\begin{align*}
			d_k	&	= (a_1^k,a_2^k),
			&
			d^K_k	&	= (b_1^k,b_2^k),
			&
			d'_k	&	= (x^k_1,x_2^k),
			&
			d''_k	&	=	(y^k_1,y^k_2)\text.
		\end{align*}
		It then suffices to prove by induction on~$k$ that for every $k\ge0$, we have
		\[
			d_k = d^K_k = d'_k =d''_k\text.
		\]
		For $k=0$, let $d_0=(a_1^0,a_2^0)$ and $d^K_0=(b_1,b_2)$. Then,
		\[
			a_1^0
			=	f_1(\epsilon)
			=	f_1(\fincomputation^\epsilon_{f_1})
			=	f_1^K(\epsilon)
			=   b_1^0
			\text.
		\]
		In a similar way we find that $a_2^0=b_2^0$, and hence $d_0=d_0^K$.
		Because both $f_1(\epsilon)=a_1^0$ and $f_2(\epsilon)=a_2^0$, we also have $d'_0=d_0$, and, \emph{a fortiori}, $d'_0\equiv_K d_0$.
		Hence, clause $(i.1)$ applies and therefore $d'_0=(f_1(\epsilon),f_2(\epsilon))=(a_1^0,a_2^0)=d_0$.
		In a similar way we can establish that $d''_0=d_0$.

		For the induction step, we may assume that
		\[
			d_0,\dots,d_k
			\mathwordbox{=}{===}
			d^K_0,\dots,d^K_k
			\mathwordbox{=}{===}
			d'_0,\dots,d'_k
			\mathwordbox{=}{===}
			d''_0,\dots,d''_k\text.
			\tag{\emph{i.h.}}
		\]
		Then,
		\[
			a_1^{k+1}
			=
			f_1(d_0,\dots,d_k)
			=_{i.h.}
			f_1(d'_0,\dots,d'_k)
			=
			f_1^K(d_0,\dots,d_k)
			=_{i.h.}
			f_1^K(d'_0,\dots,d'_k)
			=
			b_1^{k+1}\text.
		\]
		Observe that the third equality holds because
			$d'_0,\dots,d'_k
			=
			\fincomputation_{f_1}^{d_0,\dots,d_k}
			$.
		For player~$2$, the following equalities hold:
		\[
			a_2^{k+1}
			=
			f_2(d_0,\dots,d_k)
			=_{i.h.}
			f_2(d''_0,\dots,d''_k)
			=
			f_2^K(d_0,\dots,d_k)
			=_{i.h.}
			f_2^K(d'_0,\dots,d'_k)
			=
			b_2^{k+1}\text.
		\]
		Now the third equality holds since $d''_0,\dots,d''_k
			=
			\fincomputation_{f_2}^{d_0,\dots,d_k}
			$.

		Finally, by the induction hypothesis,
			$d'_0,\dots,d'_k=d_0,\dots,d_k$.
		From the above, we already had
		$
			f_1(d'_0,\dots,d'_k)
				=
			a_1^{k+1}
			$ \text{and} $
			f_2(d''_0,\dots,d''_k)
				=
			a_2^{k+1}
			$
		Hence, \[(f_1(d'_0,\dots,d'_k),f_2(d''_0,\dots,d''_k))=d^{k+1}.\]
		It follows that,
		$
			d'_0,\dots,d'_k,(f_1(d'_0,\dots,d'_k),f_2(d''_0,\dots,d''_k))=d_0,\dots,d_k,d_{k+1}.
		$
		In particular,
		\[
			d'_0,\dots,d'_k,(f_1(d'_0,\dots,d'_k),f_2(d''_0,\dots,d''_k))\equiv_K d_0,\dots,d_k,d_{k+1}.
		\]
		Therefore, clause $(i.1)$ is applicable, and both
		\begin{align*}
			x^{k+1}_1
			&	=	f_1(d'_0,\dots,d'_k)
				=	a_1^{k+1}
			&\text{and}&&
			x^{k+1}_2
			&	=	f_2(d''_0,\dots,d''_k)
				=	a_2^{k+1}\text,
		\end{align*}
		signifying that $d_{k+1}=d'_{k+1}$. By an analogous reasoning we show that $d_{k+1}=d''_{k+1}$.
\end{proof}

\noindent The next lemma establishes that~$f^K_1$ and~$f^K_2$ are winning run-invariant strategies against player~$2$ and player~$1$, respectively, if the goal sets~$\goalset_1$ and~$\goalset_2$ are to be run-based. Notice that this result requires~$\goalset_1$ and~$\goalset_2$ to be $K_{M,M'}$-closed.
\begin{lem}\label{lemma:soundness_fK_2}
	Let $f=(f_1,f_2)$ be a run-invariant strategy profile in game $G=(M,\goalset_1,\goalset_2)$ with~$\goalset_1$ and~$\goalset_2$ run-based and $K$-closed, and~$f^K$ defined as above. Then,
		\begin{enumerate}
			\renewcommand{\itemsep}{.75ex}
			\item\label{item:soundness_fK_ii} if $f_1$ is a winning (run-invariant) strategy against player~$2$, then so is $f^K_1$,
			\item\label{item:soundness_fK_iii} if $f_2$ is a winning (run-invariant) strategy against player~$1$, then so is $f^K_2$.
		\end{enumerate}
\end{lem}
\begin{proof}
	For part~\ref{item:soundness_fK_ii}, assume for contraposition that $f^K_1$ is not winning against player~$2$. Then, there is a strategy $g_2$ for player~$2$ such that $\computation(f^K_1,g_2) \in \Gamma_2$.
	We define a $K$-invariant strategy~$g^K_2$ for player~$2$ such that $d_0,\dots,d_k\equiv_K e_0,\dots,e_k$ for every $k\ge 0$, where
	\begin{align*}
	\computation(f^K_1,g_2)
		&	=	d_0,d_1,d_2,\dots
		&\text{and}&&
		d_k
		&	=	(a_1,a_2)\text,
		\\
	\computation(f_1,g^K_2)
		&	=	e_0,e_1,e_2,\dots.
		&\text{and}&&
		e_k
		&	=	(b_1,b_2)\text.
	\end{align*}
	By~$\Gamma_2$ being $K$-closed it then follows that also
	$\computation(f_1,g^K_2)\in\Gamma_2$, which contradicts our initial assumption that~$f_1$ is winning strategy against player~$2$.
		For each $k\ge 0$, let furthermore
		\begin{align*}
		\fincomputation^{d_0,\dots,d_k}
				&	=	d'_0,\dots,d'_k\phantom{,\dots}
				& \text{and}&&
		 d'_k	&	=	(c_1,c_2).
		\end{align*}
	In order to define the strategy~$g^K_2$, we may assume the existence of some $K$-invariant strategy~$h_2$ for player~$2$.
For the empty computation~$\epsilon$ we have $g^K_2(\epsilon)=x_2^0$ where
\begin{enumerate}[label=$(\roman*.3)$]
	\item $x^{0}_2= f_2(\epsilon)$, if
	$
		(f_1(\epsilon),f_2(\epsilon))
		\equiv_K
		d_{0}\text,
	$
	\item $x^{0}_2$ is the least action~$y_2^{0}$ available to player~$2$ such that
	$
		(f_1(\epsilon),y_2^{k+1})
		\equiv_K
		d_0\text,
	$
	if such an action $y^{k+1}_2$ exists and case $(i)$ does not apply,
	\item $x^{k+1}_2=h(\epsilon)$ in all other cases.
\end{enumerate}
 For every finite computation $d_0,\dots,d_k$, we have that $g^K_2(d_0,\dots,d_k)=x^{k+1}_2$, where
\begin{enumerate}[label=$(\roman*.4)$]
	\item $x^{k+1}_2= f_2(d_0,\dots,d_k)$, if
	\[
		d_0,\dots,d_k,(f_1(d_0,\dots,d_k),f_2(d_0,\dots,d_k))
		\equiv_K
		d_0,\dots,d_k,d_{k+1}\text,
	\]
	\item $x^{k+1}_2$ is the least action~$y_2^{k+1}$ available to player~$2$ such that
	\[
		d_0,\dots,d_k,(f_1(d_0,\dots,d_k),y_2^{k+1})
		\equiv_K
		d_0,\dots,d_k,d_{k+1}\text,
	\]
	if such an action $y^{k+1}_2$ exists and case $(i)$ does not apply,
	\item $x^{k+1}_2=h(d_0,\dots,d_k)$ in all other cases.
\end{enumerate}
Observe that~$g^K_2$ is $K$-invariant  by construction.

We now prove by induction on $k$ that $	e_0,\dots,e_k
	=
	d_0',\dots,d_k'
$ for every $k\ge 0$.
%
%
If $k=0$, recall that $e_0=(b_1^0,b_2^0)$ and $d^0_0=(c_1^0,c_2^0)$.
Observe that $f^K_1(\epsilon)=f_1(\fincomputation^\epsilon)= f_1(\epsilon)$.
Thus we have $d_0=(f_1(\epsilon),g_2(\epsilon))=(f_1(\epsilon),g_2(\epsilon))$, and in particular $d_0\equiv_K (f_1(\epsilon),g_2(\epsilon))$.
Hence, there is some $y_2^0$
such that $(f_1(\epsilon),y^0_2)\equiv_K d_0$.

First consider the case where
\text{
		$(f_1(\epsilon),f_2(\epsilon))\equiv_K d_0$.
}
Now clause $(i.1)$ applies and
$c_1^0=f_1(\epsilon)$ and $c_2^0=f_2(\epsilon)$.
Accordingly, clause $(i.3)$ is applicable as well, and we obtain both
\begin{align*}
	b_1^0
	&	=	f_1(\epsilon)	=	c_1^0
	&	\text{and}	&&
	b_2^0
	&	=	g^k_2(\epsilon) = f_2(\epsilon) = c_2^0\text.
\end{align*}

Otherwise,  there is some least $x_2^0$ such that $(f_1(\epsilon),x_2^0)\equiv_K d_0$. Thus, due to clause~$(ii.1)$ we have $c_1^0=f_1(\epsilon)$ and $c_2^0=x_2^0$.
Clause~$(ii.3)$ now also applies and we obtain:
\begin{align*}
	b_1^0
	&	=	f_1(\epsilon)	=	c_1^0
	&	\text{and}	&&
	b_2^0
	&	=	g^K_2(\epsilon) = x_2^0 = c_2^0\text.
\end{align*}

The induction step runs along similar lines.
We may assume that
\[\text{
$	e_0,\dots,e_k
	=
	d_0^0,\dots,d_k$.
}
\tag{$i.h.$}\]
Observe that
\begin{align*}
f^K_1(d_0,\dots,d_k)
&	=	f_1(\fincomputation^{d_0,\dots,d_k})
	=	f_1(d_0',\dots,d_k')\text.
\end{align*}
Thus,
\begin{align*}
d_{k+1}
&	=	(f^K_1(d_0,\dots,d_k),g_2(d_0,\dots,d_k))
	=	(f_1(d'_0,\dots,d'_k),g_2(d_0,\dots,d_k))\text.
\end{align*}
By Lemma~\ref{lemma:f^K_Kinvariant}, moreover, $d_0,\dots,d_k\equiv_K d_0',\dots,d_k'$, and it follows that
\[
	d_0,\dots,d_k,d_{k+1}
	\equiv_K
	d_0',\dots,d'_k,(f_1(d'_0,\dots,d'_{k}),g_2(d_0,\dots,d_k)).
\]
Hence, there is some $y_2^{k+1}$  such that
\[
	d_0,\dots,d_k,d_{k+1}
	\equiv_K
	d_0',\dots,d'_k,(f_1(d'_0,\dots,d'_k),y_2^{k+1})	.
	\tag{$\ast$}
\]

%


First assume that equation~$(\ast)$ holds for $y_2^{k+1}= f_2(d_0',\dots,d_k')$.
Then clause $(i.1)$ applies and for $d'_{k+1}=(c_1^{k+1},c_2^{k+1})$
we have
\begin{align*}
	c_1^{k+1}
	&	=	f_1(d'_0,\dots,d'_k)
	& \text{and}&&
	c_2^{k+1}=f_2(d'_0,\dots,d'_k).
\end{align*}
Recall that $e^{k+1}=(b_1^{k+1},b_2^{k+1})$. Now for player~$1$ we find that
	\[
		b_1^{k+1}
		=
		f_1(e_0,\dots,e_k)
		=_{i.h.}
		f_1(d_0',\dots,d_k')
		=
		c_1^{k+1}\text.
	\]

For player~$2$, observe that, in the case we are considering,
\begin{align*}
	e_0,\dots,e_k,(f_1(d'_0,\dots,d'_k),f_2(d'_0,\dots,d'_k))
	&
		\mathwordbox[l]{\quad=_{i.h.}\quad}{xxxxxx}
	d'_0,\dots,d'_k,(f_1(d'_0,\dots,d'_k),f_2(d'_0,\dots,d'_k))\\
	&
	\mathwordbox[l]{\quad=\quad}{xxxxxx}
	d'_0,\dots,d'_k,d'_{k+1}\text.
\end{align*}
Accordingly, clause $(i.4)$ applies, that is,
$g^K_2(e_0,\dots,e_k)=f_2(d_0',\dots,d_k')$.
Hence,
\[
	b_2^{k+1}
	=
	g^K_2(e_0,\dots,e_k)
	=
	f_2(d_0',\dots,d_k')
	=
	c_2^{k+1}\text,
\]
and we may conclude that
$e_{k+1}
=
(b_1^{k+1},b_2^{k+1})
=
(c_1^{k+1},c_2^{k+1})
=
d'_{k+1}\text.
$

Finally, assume that equation~$(\ast)$ does not hold for $y_2^{k+1}=f_2(d'_0,\dots,d_k')$.
Then, let~$x_2^{k+1}$ be the least action for player~$2$ for which equation~$(\ast)$ does hold with $y^{k+1}=x^{k+1}_2$.
As in this case clause~$(i.2)$ applies and for $d'_{k+1}=(c_1^{k+1},c_2^{k+1})$,
we have,
\begin{align*}
	c_1^{k+1}
	&	=	f_1(d'_0,\dots,d'_k)
	& \text{and}&&
	c_2^{k+1}=x_2^{k+1}.
\end{align*}
Recall that $e_{k+1}=(b_1^{k+1},b_2^{k+1})$.
For player~$1$ we again find that,
	\[
		b_1^{k+1}
		=
		f_1(e_0,\dots,e_k)
		=_{i.h.}
		f_1(d_0',\dots,d_k')
		=
		c_1^{k+1}\text.
	\]
For player~$2$, observe that, in the case we are considering,
\begin{align*}
	e_0,\dots,e_k,(f_1(d'_0,\dots,d_k'),x_2^{k+1})
	&
		\mathwordbox[l]{\quad=_{i.h.}\quad}{xxxxxx}
	d'_0,\dots,d'_k,(f_1(d'_0,\dots,d'_k),x_2^{k+1})\\
	&
	\mathwordbox[l]{\quad=\quad}{xxxxxx}
	d'_0,\dots,d'_k,d'_{k+1}\text.
\end{align*}
Accordingly, clause $(ii.4)$ applies and we have
$g^K_2(e_0,\dots,e_k)=x_2^{k+1}$.
It then follows that,
\[
	b_2^{k+1}
	=
	g^K_2(e_0,\dots,e_k)
	=
	f_2(d_0',\dots,d'_k)
	=
	c_2^{k+1}\text.
\]
Again we may conclude that
$e_{k+1}
=
(b_1^{k+1},b_2^{k+1})
=
(c_1^{k+1},c_2^{k+1})
=
d'_{k+1}\text,
$
as desired.

The proof for part~\ref{item:soundness_fK_iii} is analogous to that of part~\ref{item:soundness_fK_ii}.
\end{proof}

\noindent\label{page:cont_1} The ground has now been cleared for the main result of this paper that the existence of run-invariant equilibria is preserved under bisimulation in two-player games provided that the run-based preferences of the players are $K_{M,M'}$-closed.

\begin{thm}\label{theorem:runbased_runbased_1}
	Let $G=(M,\goalset_1,\goalset_2)$ and $G'=(M',\goalset_1,\goalset_2)$ be two two-player games on bisimilar concurrent game structures such that $\goalset_1$ and~$\goalset_2$ are run-based and $K_{M,M'}$-closed. Then, if $f=(f_1,f_2)$ is a run-invariant strategy profile in~$M$, then $f^{K}=(f^{K}_1,f^{K}_2)$ is a run-invariant equilibrium in~$M'$.
\end{thm}

\begin{proof}
	Assume that  $f=(f_1,f_2)$ is a run-invariant strategy profile in~$M$.
	By Lemma~\ref{lemma:NE_winningstrats}, then both
	\begin{enumerate}
		\item\label{item:runbased_runbased_1_i} $\computation_M(f_1,f_2)\notin\goalset_1$ implies that~$f_2$ is a winning strategy against player~$1$ in~$M$, and
		\item\label{item:runbased_runbased_1_ii}  $\computation_M(f_1,f_2)\notin\goalset_2$ implies that~$f_1$ is a winning strategy against player~$2$ in~$M$.
	\end{enumerate}
On basis of~\ref{item:runbased_runbased_1_i}, Lemma~\ref{item:soundness_fK_i}
yields $\fincomputation_M(f_1,f_2)=\fincomputation_M(f_1^K,f_2^K)$.
Now assume $\computation_M(f^K_1,f^K_2)\notin\goalset_1$. Then, also
$\computation_M(f_1,f_2)\notin\goalset_1$ and we may assume that~$f_2$ is a winning strategy against player~$1$ in~$M$. In virtue of  Lemma~\ref{lemma:soundness_fK_2} we may then conclude that~$f^K_2$ is a winning strategy against player~$1$ in~$M$.
Assuming that $\computation_M(f^K_1,f^K_2)\notin\goalset_2$, we can reason analogously and infer that~$f_1$ is a winning strategy against player~$2$ in~$M$. Hence,
	\begin{enumerate}[label=$(\arabic*')$]
		\renewcommand{\itemsep}{.5ex}
		\item\label{item:runbased_runbased_1_iii} $\computation_M(f^K_1,f^K_2)\notin\goalset_1$ implies that~$f^K_2$ is a winning strategy against player~$1$~in $M$, and
		\item\label{item:runbased_runbased_1_iv} $\computation_M(f^K_1,f^K_2)\notin\goalset_2$ implies that~$f^K_1$ is a winning strategy against player~$2$~in $M$.
	\end{enumerate}
Accordingly,~$f^K$ is a computation-based equilibrium in~$M$.
By Theorem~\ref{theorem:NE_computations} we may infer that~$f^K$ is also a computation-based equilibrium in~$M'$.
Lemma~\ref{lemma:f^K_Kinvariant} then guarantees that~$f^K$ is $K_{M,M'}$-invariant, and it follows that~$f^K$ is run-invariant in~$M'$ as well. By virtue of Theorem~\ref{theorem:comp_run_trace_invariant_eq} we may finally conclude that~$f^K$ is also a run-invariant equilibrium in~$M'$.
\end{proof}
\noindent As an immediate consequence of Theorem~\ref{theorem:runbased_runbased_1}, we have the following result, which is phrased in terms of run-based strategies instead of run-invariant strategies.
\begin{cor}\label{corollary:runbased_runbased_1}
	Let $G=(M,\goalset_1,\goalset_2)$ and $G'=(M',\goalset_1,\goalset_2)$ be two two-player games on bisimilar concurrent game structures~$M$ and~$M'$ such that $\goalset_1$ and~$\goalset_2$ are run-based and $K_{M,M'}$-closed. Let furthermore~$\run\in\runs_{M}$ be a run in~$M$ that is sustained by a run-based equilibrium in~$M$. Then, there is a run~$\run'\in\runs_{M'}$ in~$M'$ that is statewise bisimilar to~$\run$ and that is also sustained by a run-based equilibrium in~$M'$.
\end{cor}

\begin{proof}
	Let run $\run\in\runs_M$ be sustained by a run-based equilibrium~$f=(f_1,f_2)$ in~$G$ and let $\check f=(\check f_1,\check f_2)$ be the run-invariant strategy profile corresponding to~$f$.
	Lemma~\ref{lemma:invariant_eq} guarantees that $\check f=(\check f_1,\check f_2)$ is a run-invariant equilibrium in~$G$.
	Now construct profile $\check f^K=(\check f^K_1,\check f^K_2)$, which by virtue of Theorem~\ref{theorem:runbased_runbased_1}  is then $K_{M,M'}$-invariant and is a run-invariant equilibrium in both~$G$ and~$G'$. By virtue Lemma~\ref{lemma:computation_identity}-\ref{item:computation_identity}, it moreover follows that $\run_M(g^K_1,g^K_2)\pairwisebisim\run_{M'}(g^K_1,g^K_2)$, that is, $\run_M(g^K_1,g^K_2)$ and~$\run_{M'}(g^K_1,g^K_2)$ are statewise bisimilar, which concludes the proof.
\end{proof}

\noindent  A further corollary of Theorem~\ref{theorem:runbased_runbased_1} is that the \emph{existence} of run-based equilibria is preserved in two-player games with run-based and $K_{M,M'}$-closed preferences.

\subsubsection*{Trace-based Preferences}
We find that, with a couple of slight modifications, essentially the same construction as in the previous section can be leveraged to prove that run-based equilibria are also preserved under bisimulation in two-player games with \emph{trace-based} preferences. It be emphasised that here we do not require the preferences to satisfy any other condition than being trace-based.

Let two CGS-games $G=(M,\goalset_1,\goalset_2)$ and $G'=(M',\goalset_1,\goalset_2)$ on
bisimilar concurrent game structures~$M$ and~$M'$ and with~$\goalset_1$ and~$\goalset_2$ trace-based be given.
For a run-invariant equilibrium~$f=(f_1,f_2)$ in game~$G$, we define the $K_{M,M'}$-invariant  strategy-profile $f^K=(f^K_1,f^K_2)$ as in the previous section.
We prove that $f^K=(f^K_1,f^K_2)$ is also a run-invariant equilibrium in~$G'$.
To this end, we adapt Lemma~\ref{lemma:soundness_fK_2} so as to apply to trace-based preferences instead of preferences that are both run-based and $K_{M,M'}$-closed.

\begin{lem}\label{lemma:soundness_fK_tracebased}
	Let $f=(f_1,f_2)$ be a run-invariant strategy profile in game $G=(M,\goalset_1,\goalset_2)$ with~$\goalset_1$ and~$\goalset_2$ trace-based, and~$f^K$ defined as above. Then,
		\begin{enumerate}
			\renewcommand{\itemsep}{.5ex}
			\item\label{item:soundness_fKtb_ii} if $f_1$ is a winning (run-invariant) strategy against player~$2$, then so is $f^K_1$,
			\item\label{item:soundness_fKtb_iii} if $f_2$ is a winning (run-invariant) strategy against player~$1$, then so is $f^K_2$.
		\end{enumerate}
\end{lem}

\begin{proof}
	For part~\ref{item:soundness_fKtb_ii}---part~\ref{item:soundness_fKtb_iii} follows by an analogous argument---assume for contraposition that $f^K_1$ is not a winning strategy against player~$2$. Then, there is a strategy $g_2$ for player~$2$ such that $\computation(f^K_1,g_2) \in \Gamma_2$.
	We define a $K_{M,M'}$-invariant strategy~$g^K_2$ for player~$2$
	exactly as in the proof of Lemma~\ref{lemma:soundness_fK_2}.
	Accordingly, $d_0,\dots,d_k\equiv_{K_{M,M'}} e_0,\dots,e_k$ for every $k\ge 0$, where
	\begin{align*}
	\computation(f^K_1,g_2)
		&	=	d_0,d_1,d_2,\dots
		&\text{and}&&
		d_k
		&	=	(a_1,a_2)\text,
		\\
	\computation(f_1,g^K_2)
		&	=	e_0,e_1,e_2,\dots.
		&\text{and}&&
		e_k
		&	=	(b_1,b_2)\text.
	\end{align*}
	Now consider an arbitrary~$k\ge 0$. Then, by Lemma~\ref{lemma:finite_congruence_statewise_bisim}, also $\finrun_M(d_0,\dots,d_k)\statewisebisim \finrun_M(e_0,\dots,e_k)$.
Letting
$\finrun_M(d_0,\dots,d_k)=s_0,\dots,s_k$
and
$\finrun_M(d'_0,\dots,d'_k)=s'_0,\dots,s'_k$, we then also have~$s_k\bisim s_k'$.
It follows that
$\run_M(f^K_1,g_2)\statewisebisim\run_M(f_1,g^K_2)$ and hence
$\trace_M(f^K_1,g_2)=\trace_M(f_1,g^K_2)$.
As a consequence of $\goalset_2$ being trace-based, we obtain
	$\computation(f_1,g^K_2)\in\Gamma_2$, which contradicts our initial assumption that~$f_1$ is winning strategy against player~$2$.
\end{proof}

\noindent We are now in a position to prove the counterpart of Theorem~\ref{theorem:runbased_runbased_1} for trace-based preferences, showing that run-invariant equilibria are preserved under bisimulation if the players' preferences are trace-based.
\begin{thm}\label{theorem:runbased_tracebased_1}
	Let $G=(M,\goalset_1,\goalset_2)$ and $G'=(M',\goalset_1,\goalset_2)$ be two two-player games on bisimilar concurrent game structures such that $\goalset_1$ and~$\goalset_2$ are trace-based. Then, if $f=(f_1,f_2)$ is  run-invariant in~$M$, then $f^{K}=(f^{K}_1,f^{K}_2)$ is a run-invariant equilibrium in~$M'$.
\end{thm}

\begin{proof}
	The proof is fully analogous to that for Theorem~\ref{theorem:runbased_runbased_1}, invoking Lemma~\ref{lemma:soundness_fK_tracebased} instead of Lemma~\ref{lemma:soundness_fK_2}.
\end{proof}
\noindent As an immediate consequence of Theorem~\ref{theorem:runbased_tracebased_1}, we
find that also the \emph{existence} of run-invariant equilibria is preserved in two-player games with trace-based preferences. Furthermore, also the counterpart of Corollary~\ref{corollary:runbased_runbased_1} for trace-based preferences can easily be demonstrated.

\subsection*{Boolean Game Structures}
We now consider a subclass of concurrent game structures
in which Nash equilibrium is invariant under bisimilarity.
Specifically, we study games played over the class of concurrent game structures induced by iterated Boolean games~\cite{GHW15}, a framework that can be used to reason about Nash equilibria in games and multi-agent systems modelled using the Reactive Modules specification language~\cite{AH99b}.

By a \emph{Boolean game structure} we understand a special type of concurrent game structure $M=(\Ag,\AP,\Ac,\States,s^0_M,\valf,\transf)$ for which there is a partition $\set{\AP_1,\dots,\AP_n}$ of~$\AP$ such that $\Ac_i(s) \subseteq (2^{\AP_i}\setminus\emptyset)$ for all players~$i$ and states~$s$ and for every direction $\direction'=(a_1,\dots,a_n)$ in $2^{\AP_1}\times\cdots\times 2^{\AP_n}$ and every state~$s$, it holds that
\[
	\text{$\transf(s,\direction')=s'$ implies $\valf(s')=a_1\cup\dots\cup a_n$.}
\]

Then, informally, in a Boolean game structure, choice profiles correspond to system states, which is not generally the case in concurrent game structures. In other words, in a Boolean game structure~$M$, if a strategy profile induces a run $s^0_M,s_1,s_2,\ldots$, then we know that it has been induced by the computation~$s_1,s_2,\ldots$. Even more, we also know that the trace induced by such a computation is precisely $s^0_M,s_1,s_2,\ldots$. This very strong correspondence between computations, runs, and traces is key to the proof that in Boolean game structures all strategies for a player are in fact bisimulation-invariant. This result, in turn, can also be used to show that Nash equilibrium is invariant under bisimilarity, regardless of the model of strategies or goals that one chooses. 
To see this, the following preliminary results will be useful.

\begin{lem}\label{lemma:BGS_pairwise_bisim_identity}
	Let $M=(\Ag,\AP,\Ac,\States,s^0,\valf,\transf)$ be a Boolean game structure with partition $\set{\AP_1,\dots,\AP_n}$. Let $\finrun=s_0,\dots,s_k$ and $\finrun'=s'_0,\dots,s_k'$ be statewise bisimilar finite histories, that is, $\finrun\statewisebisim\finrun'$. Then, $\finrun=\finrun'$.
\end{lem}
\begin{proof}
	We may assume that there are finite computations $\computation=d_0,\dots,d_{k-1}$ and $\computation'=d'_0,\dots,d'_{k-1}$ such that  $s_0\transarrow{\direction_0}\cdots\xtransarrow{\direction_{k-1}}s_k$
	and $s'_0\transarrow{\direction'_0}\cdots\xtransarrow{\direction'_{k-1}}s_k'$.
	We show by induction that $s_m=s'_m$ for all~$0\le m\le k$.
	For the basis, we have $s_0=s_M^0=s'_0$. For the induction step we may assume that $s_m=s'_m$.
	Moreover, as $s_{m+1}\bisim s_{m+1}'$, also $\lambda(s_{m+1})=\lambda(s'_{m+1})$. Furthermore, $s_m\xtransarrow{\direction_m}s_{m+1}$ and $s'_m\xtransarrow{\direction'_m}s'_{m+1}$. As~$M$ is a Boolean game structure, it follows that $\direction_m=(\valf(s_{m+1})\cap\AP_1,\dots,\valf(s_{m+1})\cap\AP_n)$ and
	$\direction'=(\valf(s'_{m+1})\cap\AP_1,\dots,\valf(s'_{m+1})\cap\AP_n)$ and, hence, $\direction_{m}=\direction'_{m}$.
	By determinism of~$\transf$, we may conclude that $s_{m+1}=\transf(s_m,\direction_m)=\transf(s_m',\direction'_m)=s'_{m+1}$.
\end{proof}

The above lemma can be used to show that in fact, for Boolean game structures, all models of strategies collapse to the model of bisimulation-invariant strategies.

\begin{lem}\label{lemma:BGS_sw_bisim_strats}
	In Boolean game structures, all strategies for every player are bisimulation-invariant.
\end{lem}
\begin{proof}
	Consider an arbitrary strategy~$f_i$ of a player~$i$ in a Boolean game structure~$M$ along with arbitrary statewise bisimilar histories~$\finrun,\finrun'\in\finruns_M$, that is, $\finrun\statewisebisim\finrun'$. By Lemma~\ref{lemma:BGS_pairwise_bisim_identity}, then $\finrun=\finrun'$. Hence, trivially, $f_i(\finrun)=f_i(\finrun')$.
\end{proof}

We can now present the main result of this section.

\begin{thm}\label{thm:BGS_invariance}
	In Boolean game structures, (the existence of a) Nash equilibrium is invariant under bisimilarity.
\end{thm}
\begin{proof}
	Observe that because of Lemma~\ref{lemma:BGS_sw_bisim_strats}, in Boolean game structures, a strategy profile~$f=(f_1,\dots,f_n)$ is a Nash equilibrium if and only if~$f$ is a Nash equilibrium in bisimulation-invariant strategies. The result then immediately follows from Corollary~\ref{corollary:sustained_invariance_under_bisimulation_run}.
\end{proof}

\section{Nondeterminism}%
\label{secn:nondeterminism}
Our results so far, summarised in Table~\ref{tab:summary},
apply to profiles of deterministic strategies and deterministic systems. In this section, we investigate the case of nondeterministic systems. In this more general setting, most of our notations and definitions remain the same, except for two that are particularly relevant: the notions of \emph{outcome} of a game and Nash {\em equilibrium}.

\begin{table}
\small
\begin{tabular}{llc@{\quad}c@{\quad}c}
\toprule
&&\multicolumn{3}{c}{Preferences}\\
&& computation-based & run-based & trace-based\\\toprule
&computation-based 		& $+$\quad (Th.~\ref{theorem:NE_computations}) & $\phantom{{}^{\ast}}+^{\phantom{\mathrm a}}$\quad (Th.~\ref{theorem:NE_computations})	& $+$\quad(Th.~\ref{theorem:NE_computations})  \\
&run-based (general) 	& $-$\quad(Th.~\ref{thm:variance}) 						& $\phantom{{}^{\ast}}-^{\phantom{\mathrm a}}$\quad (Th.~\ref{thm:variance}) & $-$\quad (Th.~\ref{thm:variance})  \\
Strategies&run-based (two players) 	& $-$\quad (Th.~\ref{thm:variance_two-player}) & $\phantom{{}^{\ast}}+^{\dagger}$\quad (Th.~\ref{theorem:runbased_runbased_1}) 	& $+$\quad \wordbox[l]{(Th.~\ref{theorem:runbased_tracebased_1})}{(Th.~\ref{theorem:NE_computations})}  \\
&trace-based 			& $+$\quad (Th.~\ref{theorem:NE_computations_traces}) & $\phantom{{}^{\ast}}+^{\phantom{\mathrm a}}$\quad (Th.~\ref{theorem:NE_computations_traces})	& $+$\quad (Th.~\ref{theorem:NE_computations_traces})  \\
&bisimulation-invariant 			& $+$\quad(Th.~\ref{theorem:NE_runs_pwbisim}) 	& $\phantom{{}^{\ast}}+^{\phantom{\mathrm a}}$\quad (Th.~\ref{theorem:NE_runs_pwbisim}) & $+$\quad (Th.~\ref{theorem:NE_runs_pwbisim})  \\
\bottomrule
\multicolumn{5}{l}{$\rule{0pt}{3ex}{}^{\text{$\dagger$}}${\footnotesize Assuming the players' (run-based) preferences to be $K_{M,M'}$-closed.}}
\end{tabular}
\caption{Summary of main bisimulation-invariance results for multi-player games in deterministic systems as well as the results in this paper they are based on. In this figure,~$+$ means that Nash equilibria are preserved in (computation/run/trace)-based strategy profiles with preferences given by sets of computations/runs/traces, while~$-$ indicates that they are not for such a pair.}%
\label{tab:summary}
\end{table}

Note that in a deterministic system, a profile of deterministic strategies induces a \emph{unique} system path (and therefore a unique computation, run, and trace). However, if the system is nondeterministic, a profile of deterministic strategies might, instead, determine a \emph{set} of paths of the system: all those complying with the profile of strategies. For instance, in the system in Figure~\ref{fig:nondetex1}, the deterministic strategy profile where every player~$i$ chooses to play $a_i$ at the beginning determines two different runs and traces of the system.


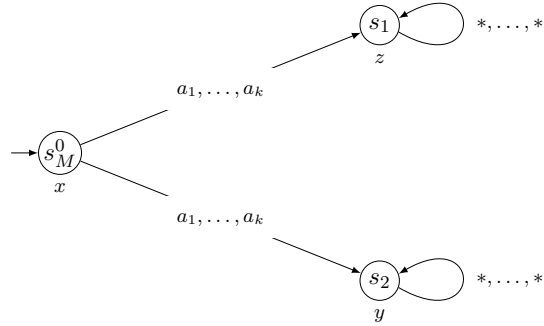
\begin{figure}
\scalebox{.85}{
 	  \begin{tikzpicture}[scale=1]

 	  \tikzstyle{every ellipse node}=[draw,inner xsep=3.5em,inner ysep=1.2em,fill=black!15!white,draw=black!15!white]
 	  \tikzstyle{every circle node}=[fill=white,draw,minimum size=1.6em,inner sep=0pt]


	  \draw(-.9,0)  node(0){};
	  \draw(0,0)	node[label=-90:{\footnotesize$x$},circle](v0){$s_M^0$};
	  \draw(5,2)	node[label=-90:{\footnotesize$z$},circle](v1){$s_1$};
	  \draw(5,-2)	node[label=-90:{\footnotesize$y$},circle](v2){$s_2$};



	\draw[-latex] (0) -- (v0);
 	\draw[-latex] (v0) --node[pos=.5,fill=white](){\footnotesize$\begin{array}{l}a_1,\dots,a_k\end{array}$} (v1);
 	\draw[-latex] (v0) --node[pos=.5,fill=white](){\footnotesize$\begin{array}{l}a_1,\dots,a_k\end{array}$} (v2);

 	%

 	\draw[-latex] (v2.70-90) .. controls +(60-90:4em) and +(120-90:4em) .. node[pos=.5,fill=white,right,xshift=.5ex](){\footnotesize$\ast,\dots,\ast$} (v2.110-90);

 	\draw[-latex] (v1.70-90) .. controls +(60-90:4em) and +(120-90:4em) .. node[pos=.5,fill=white,right,xshift=.5ex](){\footnotesize$\ast,\dots,\ast$} (v1.110-90);


 	\end{tikzpicture}
 	}
\caption{Nondeterministic system with sets of infinite runs given by~$s^0_M.(s_1^\omega+s_2^\omega)$ and infinite traces given by~$x.(z^\omega+y^\omega)$.}%
\label{fig:nondetex1}
\end{figure}

Therefore, formally, a deterministic strategy profile~$f$ on a nondeterministic system~$M$ may determine a set of computations in $\computations_M$. To simplify notations, we will write $\computation_M(f) \subseteq \computations_M$ for such a set, which will correspond to the set of computations that could result in~$M$ when playing strategy profile~$f$. Likewise, we will write~$\run_M(f)\subseteq\runs_M$ and $\trace_M(f)\subseteq\traces_M$, respectively, for the sets of runs and traces determined by~$f$ on~$M$. These three sets of computations, runs, and traces determined by~$f$, namely $\computation_M(f)$, $\run_M(f)$, and $\trace_M(f)$, will correspond to our more general notion of (computation, run, trace) \emph{outcome} of a multi-player game. Clearly, for deterministic systems, these sets of computations, runs, and traces will correspond to the special case where outcomes are singleton sets.%
\footnote{Later on, in this section, we will present some examples of how sets of computations/runs/traces can be induced by (deterministic) computation-based/run-based/trace-based strategies in nondeterministic systems.}


Our more general definition of outcome call for a (just slightly) more general definition of equilibrium. The generalisation is rather simple. With respect to a nondeterministic system~$M$, we will define the preferences~$\Gamma_i$ of a player~$i$ as a set of sets of computations of~$M$, rather than just a set of computations, as in the deterministic case. In other words, while in a deterministic system we have $\goalset_i\subseteq\computations_M$, in a nondeterministic system we have~$\goalset_i\subseteq 2^{\computations_M}$. With this definition in place, we can define a Nash equilibrium in exactly the same way that it is defined for deterministic systems, that is, as a strategy profile~$\strategy=(\strategy_1,\dots,\strategy_n)$ such that for every player~$i$ and every strategy~$g_i$ available to~$i$,
\[
	\text{
	$\computation_M(f_{-i},g_i)\in\goalset_i$
	implies
	$\computation_M(f)\in\goalset_i$.
	}
\]
As for deterministic systems, the concepts of \emph{Nash equilibrium in run-based strategies} and \emph{Nash equilibrium in trace-based strategies} are defined analogously.

We first note that all negative results for deterministic systems immediately carry over to this more general setting as those are simply the case when deterministic strategy profiles induce a unique computation (a singleton set of computations). On the other hand, although positive results for computations and traces also carry over to nondeterministic systems, this is not something that one can immediately conclude. A couple of technical results are needed. In the reminder of this section we will study why positive results for computations and traces do carry over to nondeterministic systems.

The first observation to make is that the set of strategy profiles across bisimilar systems is invariant, that is, that every collection of (computation-based, trace-based) strategies~$f=(f_1,\ldots,f_n)$ is a strategy profile in a system~$M$ if and only if~$f$ is a strategy profile in~$M'$, for every~$M'$ that is bisimilar to~$M$.
\begin{lem}%
\label{lemma:nondet-inv1}
Let $M$ and $M'$ be two bisimilar systems. For all (computation-based, trace-based) strategy profiles~$f$:
\begin{center}
$f$ is a strategy profile in~$M$ if and only if $f$ is a strategy profile in~$M'$.
\end{center}
\end{lem}
\begin{proof}
By induction on the length of computations/traces, and noting that, for every player, the set of actions available to a player in bisimilar states is the same (as otherwise the two states would not be bisimilar).
\end{proof}

The second observation is that, despite nondeterminism, the outcome of games across bisimilar systems is invariant. Formally, we have the following result.

\begin{lem}%
\label{lemma:nondet-inv}
Let $M$ and $M'$ be two bisimilar systems. For all (computation-based, trace-based) strategy profiles~$f$ we have
\begin{align*}
& \computation_M(f) = \computation_{M'}(f) &
&	\text{and}&
& \trace_M(f) = \trace_{M'}(f) \ .
\end{align*}
\end{lem}
\begin{proof}
There are four different cases to consider here: either~$f$ is \emph{computation-based} or~$f$ is {\em trace-based}, and either the outcome of the game is taken to be the set of {\em computations}, or the outcome of the game is taken to be the set of {\em traces}.

By double inclusion, we show the first case: $f$ being \emph{computation-based} and the outcome of the game taken to be the set of {\em computations}. To show that~$\computation_M(f) \subseteq \computation_{M'}(f)$, with~$f$ computation-based, reason by contradiction. Suppose that there is a computation~$\computation^*$ in~$\computation_M(f)$ that is not in~$\computation_{M'}(f)$. Since~$M$ and~$M'$ are bisimilar, $\computation^*$ is also a computation of~$M'$, and due to Lemma~\ref{lemma:nondet-inv1}, for every prefix~$\computation^*_k$ of~$\computation^*$, we know that~$f(\computation^*_k)$ is defined. Since $f$ is functional, $f(\computation^*_k)$ in~$M'$ is the same as $f(\computation^*_k)$ in~$M$, which must be precisely the last direction of~$\computation^*_{k+1}$. By an inductive argument we can conclude that~$\computation^*$ must also be a computation of~$\computation_{M'}(f)$, which is a contradiction to our hypothesis, proving the statement. We can reason in a symmetric way to prove the inclusion in the other direction. Note that for computation-based strategies not only $\computation_M(f) = \computation_{M'}(f)$, but also they are \emph{singleton} sets.

The second case we consider is when~$f$ is \emph{trace-based} and the outcome of the game is taken to be the set of {\em traces}.
To show this case, we can reason similarly, but, unlike for computation-based strategies, the sets $\trace_M(f)$ and $\trace_{M'}(f)$ may not be singleton sets. We, again, show the result by double inclusion, and each direction by contradiction. Thus, first, suppose that there is a trace~$\trace^*$ in~$\trace_M(f)$ that is not in~$\trace_{M'}(f)$. Since~$M$ and~$M'$ are bisimilar, $\trace^*$ is also a trace of~$M'$, and due to Lemma~\ref{lemma:nondet-inv1}, for every prefix~$\trace^*_k$ of~$\trace^*$, we know that~$f(\trace^*_k)$ is defined. Let~$\trace^*_k$ be the smallest prefix of~$\trace^*$ that is not a prefix of any trace in~$\trace_{M'}(f)$, and let~$s$ be any state that can be reached after following the finite trace~$\trace^*_{k-1}$ from~$s^0_M$, the initial state of~$M$. Then, we know that $s \xtransarrow{f(\trace^*_{k-1})} q$, for some~$q$ such that~$\lambda(q)$ is the last element of~$\trace^*_k$.  Necessarily, the prefix~$\trace^*_{k-1}$ is the prefix of some trace in~$\trace_{M'}(f)$ that leads to a state, say~$s'$, that is bisimilar to~$s$. Because $s$ and $s'$ are bisimilar, $s' \transarrow{f(\trace^*_{k-1})} q'$ for some state~$q'$ that is bisimilar to~$q$. Lemma~\ref{lemma:nondet-inv1} ensures that~$f$ is defined at~$\trace^*_{k-1}$ in~$M'$. Since~$q$ and~$q'$ are bisimilar, it also follows that~$\lambda(q)=\lambda(q')$, and therefore that~$\trace^*_{k}$, with~$\lambda(q')$ being the last element of~$\trace^*_k$, is the prefix of some trace in~$\trace_{M'}(f)$, which is a contradiction to our hypothesis. Therefore, via induction on the length of traces, we can conclude, in particular, that $\trace^*\in\trace_{M'}(f)$, and in general that every trace in~$\trace_{M}(f)$ must also be a trace in~$\trace_{M'}(f)$. The inclusion in the other direction is, as before, obtained by symmetric reasoning.

The third case we consider is when~$f$ is \emph{computation-based} and the outcome of the game is taken to be the set of \emph{traces}. This proof is \emph{almost} identical to the previous case.
To show this case we, again, show the result by double inclusion, and each direction by contradiction. First, suppose that there is a trace~$\trace^*$ in~$\trace_M(f)$ that is not in~$\trace_{M'}(f)$. Since~$M$ and~$M'$ are bisimilar, $\trace^*$ is also a trace of~$M'$, and due to Lemma~\ref{lemma:nondet-inv1}, for every prefix~$\trace^*_k$ of~$\trace^*$ and every computation~$\computation_k\in\computation(\trace^*_k)$, we know that~$f(\computation_k)$ is defined. Let~$\trace^*_k$ be the smallest prefix of~$\trace^*$ that is not a prefix of any trace in~$\trace_{M'}(f)$, and let~$s$ be any state that can be reached after following the finite trace~$\trace^*_{k-1}$ from~$s^0_M$, the initial state of~$M$. Then, we know that for some computation~$\computation_{k-1}\in\computation(\trace^*_{k-1})$, we have~$s \xtransarrow{f(\computation_{k-1})} q$, for some~$q$ such that~$\lambda(q)$ is the last element of~$\trace^*_k$. Necessarily, the prefix~$\trace^*_{k-1}$ is the prefix of some trace in~$\trace_{M'}(f)$ that leads to a state, say~$s'$, that is bisimilar to~$s$. Because $s$ and $s'$ are bisimilar, $s' \xtransarrow{f(\computation_{k-1})} q'$ for some state~$q'$ that is bisimilar to~$q$. Lemma~\ref{lemma:nondet-inv1} ensures that~$f$ is defined at~$\computation_{k-1}$ in~$M'$. Since~$q$ and~$q'$ are bisimilar, it also follows that~$\lambda(q)=\lambda(q')$, and therefore that~$\trace^*_{k}$, with~$\lambda(q')$ being the last element of~$\trace^*_k$, is the prefix of some trace in~$\trace_{M'}(f)$, which is a contradiction to our hypothesis. Therefore, via induction on the length of traces, we can conclude, in particular, that $\trace^*\in\trace_{M'}(f)$, and in general that every trace in~$\trace_{M}(f)$ must also be a trace in~$\trace_{M'}(f)$. The inclusion in the other direction is also obtained by symmetric reasoning.

The fourth and final case we consider is when~$f$ is \emph{trace-based} and the outcome of the game is taken to be the set of {\em computations}. This proof is also \emph{almost} identical to the previous two cases.
To show this case we, again, show the result by double inclusion, and each direction by contradiction. First, suppose that there is a computation~$\computation^*$ in~$\computation_M(f)$ that is not in~$\computation_{M'}(f)$. Since~$M$ and~$M'$ are bisimilar, $\computation^*$ is also a computation of~$M'$, and due to Lemma~\ref{lemma:nondet-inv1}, for every prefix~$\computation^*_k$ of~$\computation^*$ and every trace~$\trace_k\in\trace(\computation^*_k)$, we know that~$f(\trace_k)$ is defined. Let~$\computation^*_k$ be the smallest prefix of~$\computation^*$ that is not a prefix of any computation in~$\computation_{M'}(f)$, and let~$s$ be any state that can be reached after following the finite computation~$\computation^*_{k-1}$ from~$s^0_M$, the initial state of~$M$, while complying with $\trace^*_{k-1}$, that is, any state~$s$ such that
\[
\Theta_{k-1} = s^0_M \xtransarrow{f(\lambda(s^0_M))} s_1 \xtransarrow{f(\lambda(s^0_M),\lambda(s_1))} \ldots \xtransarrow{f(\trace_{k-1})} s
\]
with
\[
\trace^*_{k-1} = \lambda(s^0_M),\lambda(s_1),\ldots,\lambda(s)
\]
and
\[
\computation^*_{k-1} = f(\lambda(s^0_M)), f(\lambda(s^0_M),\lambda(s_1)), \ldots, f(\trace_{k-1}) \ .
\]

Then, for trace~$\trace_{k-1}\in \trace(\computation^*_{k-1})$ as above, we have~$s \xtransarrow{f(\trace^*_{k-1})} q$, for some~$q$ such that~$\lambda(q)$ is the last element of~$\trace^*_k$. Necessarily, the prefix~$\computation^*_{k-1}$ is the prefix of some computation in~$\computation_{M'}(f)$ that leads to a state, say~$s'$, that is bisimilar to~$s$, that is, a computation
\[
\computation^*_{k-1} = f(\lambda(s^0_{M'})), f(\lambda(s^0_{M'}),\lambda(s'_1)), \ldots, f(\trace_{k-1})
\]
with
\[
\trace^*_{k-1} = \lambda(s^0_{M'}),\lambda(s'_1),\ldots,\lambda(s')
\]
and
\[
\Theta'_{k-1} = s^0_{M'} \xtransarrow{f(\lambda(s^0_{M'}))} s'_1 \xtransarrow{f(\lambda(s^0_{M'}),\lambda(s'_1))} \ldots \xtransarrow{f(\trace_{k-1})} s'
\]
and~$\Theta_{k-1}(i)$ bisimilar to $\Theta'_{k-1}(i)$, for every~$0\leq i \leq k-1$.

Because~$f$ is functional and
\[
\lambda(s^0_M),\lambda(s_1),\ldots,\lambda(s) = \lambda(s^0_{M'}),\lambda(s'_1),\ldots,\lambda(s')
\]
it follows that $f(\lambda(s^0_M),\lambda(s_1),\ldots,\lambda(s_k)) = f(\lambda(s^0_{M'}),\lambda(s'_1),\ldots,\lambda(s'))$ and that $s' \xtransarrow{f(\trace^*_{k-1})} q'$ for some state~$q'$ bisimilar to~$s'$. Lemma~\ref{lemma:nondet-inv1} ensures that~$f$ is defined at~$\computation_{k-1}$ in~$M'$. Since~$q$ and~$q'$ are bisimilar, it also follows that~$\lambda(q)=\lambda(q')$, and therefore that~$\trace^*_{k}$, with~$\lambda(q')$ being the last element of~$\trace^*_k$, is the prefix of some trace in~$\trace_{M'}(f)$. Since~$\trace^*_k$ is indeed a trace in~$\trace_{M'}(f)$ which can be obtained following some computation~$\computation^*_k$ in~$M'$, then we can conclude that~$\computation^*_k$ is the prefix of some computation~$\computation^*$ in~$\computation_{M'}(f)$, which is a contradiction to our hypothesis. Therefore, via induction on the length of computations, we can infer, in particular, that $\computation^*\in\trace_{M'}(f)$, and in general that every computation in~$\computation_{M}(f)$ must also be a computation in~$\computation_{M'}(f)$. As in all previous cases, because $M$ and $M'$ are bisimilar, the inclusion in the other direction is also obtained by symmetric reasoning.
\end{proof}

Using Lemma~\ref{lemma:nondet-inv1} and Lemma~\ref{lemma:nondet-inv}, we can then show that the set of (computation-based, trace-based) Nash equilibria across bisimilar systems remains invariant too.


\begin{thm}%
	\label{thm:nondet-inv-paul}
	Let $G=(M,\goalset_1,\dots,\goalset_n)$ and $G'=(M',\goalset_1,\dots,\goalset_n)$ be games on bisimilar nondeterministic concurrent game structures~$M$ and~$M'$, respectively.
	Let further~$f^\kappa$ be a computation-based strategy profile and~$f^\tau$ be a trace-based strategy profile.
	Then,
	\begin{enumerate}
		\item
		$f^\kappa$ is a computation-based Nash equilibrium in~$G$ if and only if
				$f^\kappa$ is a computation-based equilibrium in~$G'$, and
		\item
		$f^\tau$ is a trace-based Nash equilibrium in~$G$ if and only if
				$f^\tau$ is a trace-based equilibrium in~$G'$.
	\end{enumerate}
\end{thm}

\begin{proof}
Both proofs are by double implication, where each direction is proved by contradiction. For part~(1), first assume that there is some computation-based equilibrium~$f^\kappa$ in~$G$ that  is not a computation-based equilibrium in~$G'$.
Because of Lemma~\ref{lemma:nondet-inv}, every player~$i$ who gets its goal achieved in~$M$ also gets its goal achieved in~$M'$. Then, they will not deviate in~$M'$. Therefore, there must be a player~$j$ who does not get its goals achieved in~$M$ and has a beneficial deviation~$g_j$ in~$M'$, that is, while~$\computation_M(f^\kappa)\not\in\goalset_j$, we have~$\computation_{M'}(f^\kappa_{-j},g_j)\in\goalset_j$. Lemma~\ref{lemma:nondet-inv1} then ensures that~$g_i$ is also a strategy in~$M$, and Lemma~\ref{lemma:nondet-inv} that~$\computation_{M}(f^\kappa_{-j},g_j)\in\goalset_j$, which is a contradiction with $f^\kappa$ being a computation-based Nash equilibrium in~$M$. We can reason in a symmetric way to show the implication in the opposite direction.
Notice that because of Lemmas~\ref{lemma:nondet-inv1} and~\ref{lemma:nondet-inv}, the result holds for both computation-based preferences and trace-based preferences.
Finally, the proof for trace-based strategies (part~(2)), with either computation-based or trace-based preferences, follows the exact same reasoning.
\end{proof}

We note that the main idea behind the proofs in this section is that if a given strategy profile~$f$, whether computation-based or trace-based, does not determine the same set of computations and traces in bisimilar systems, then that computation or trace could be used to show that the two systems are in fact not bisimilar. This is the main argument behind the four cases in Lemma~\ref{lemma:nondet-inv}, each requiring slightly different proofs that such a witness to the non-bisimilarity of~$M$ and~$M'$ does not exist.
However, it is also important to note that if~$M$ and~$M'$ are not bisimilar, then a given strategy profile~$f$, well defined in both systems, may not determine the same set of outcomes---see, for instance, the example in Figure~\ref{fig:nondetex2}.


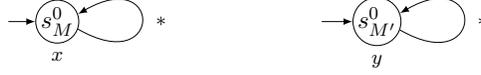
\begin{figure}
\scalebox{.85}{
 	  \begin{tikzpicture}[scale=1]

 	  \tikzstyle{every ellipse node}=[draw,inner xsep=3.5em,inner ysep=1.2em,fill=black!15!white,draw=black!15!white]
 	  \tikzstyle{every circle node}=[fill=white,draw,minimum size=1.6em,inner sep=0pt]


	  \draw(-.9,0)  node(0){};
	  \draw(0,0)	node[label=-90:{\footnotesize$x$},circle](v0){$s_M^0$};

	  \draw(4.0,0)  node(00){};
	  \draw(5,0)	node[label=-90:{\footnotesize$y$},circle](v1){$s_{M'}^0$};
	\draw[-latex] (0) -- (v0);
	\draw[-latex] (00) -- (v1);
 	\draw[-latex] (v0.70-90) .. controls +(60-90:4em) and +(120-90:4em) .. node[pos=.5,fill=white,right,xshift=.5ex](){\footnotesize$\ast$} (v0.110-90);

 	\draw[-latex] (v1.70-90) .. controls +(60-90:4em) and +(120-90:4em) .. node[pos=.5,fill=white,right,xshift=.5ex](){\footnotesize$\ast$} (v1.110-90);

 	\end{tikzpicture}
 	}
\caption{Two non-bisimilar systems where every computation-based strategy profile determines two different (infinite) trace outcomes, namely, $x^\omega$ in the system on the left ($M$) and~$y^\omega$ in the system on the right ($M'$).}%
\label{fig:nondetex2}
\end{figure}

We would also like to note that even though for deterministic systems, computation-based strategies strictly generalise run-based strategies, and run-based strategies strictly generalise trace-base strategies, for nondeterministic systems this is no longer the case. Run-based strategies still generalise trace-based strategies, but not other relation like this holds between any other pair of models of strategies. For instance, as shown in the example in Figure~\ref{fig:nondetex3}, there is a system in which, for instance, a trace-based strategy can be defined (as well as a run-based strategy) while a computation-based strategy cannot.

In case strategies are well defined, as mentioned before, they can induce sets of computations, runs, and traces in a nondeterministic system, except for one case: whenever defined, a computation-based strategy always determines a unique computation of the system, whether deterministic or nondeterministic. Examples of all other cases (8 in total) are easy to build. For instance, in the nondeterministic system in Figure~\ref{fig:nondetex3}, any run-based or trace-based strategy will induce a set of computations containing both $ab^\omega$ and $ac^\omega$. Correspondingly, they will also induce a set of runs and a set of traces, namely, those containing, respectively, $s^0_{M}s_1^\omega$ and $s^0_{M}s_2^\omega$ in case of runs, and $xy^\omega$ and $xz^\omega$ in case of traces. For the two remaining cases, a set of runs and a set of traces induced by a computation-based strategy, consider the nondeterministic system in Figure~\ref{fig:nondetex4}, which is almost the same as the system in Figure~\ref{fig:nondetex3}, save that a computation-based strategy can be defined. In such a system, any well defined computation-based strategy will induce a set of runs and a set of traces containing, respectively, $s^0_{M}s_1^\omega$ and $s^0_{M}s_2^\omega$ in case of runs, and $xy^\omega$ and $xz^\omega$ in case of traces.


\begin{figure}
\scalebox{.85}{
 	  \begin{tikzpicture}[scale=1]

 	  \tikzstyle{every ellipse node}=[draw,inner xsep=3.5em,inner ysep=1.2em,fill=black!15!white,draw=black!15!white]
 	  \tikzstyle{every circle node}=[fill=white,draw,minimum size=1.6em,inner sep=0pt]


	  \draw(-.9,0)  node(0){};
	  \draw(0,0)	node[label=-90:{\footnotesize$x$},circle](v0){$s_M^0$};
	  \draw(4,1.5)	node[label=-90:{\footnotesize$y$},circle](v1){$s_1$};
	  \draw(4,-1.5)	node[label=-90:{\footnotesize$z$},circle](v2){$s_2$};



	\draw[-latex] (0) -- (v0);
 	\draw[-latex] (v0) --node[pos=.5,fill=white](){\footnotesize$\begin{array}{l}a\end{array}$} (v1);
 	\draw[-latex] (v0) --node[pos=.5,fill=white](){\footnotesize$\begin{array}{l}a\end{array}$} (v2);

 	%

 	\draw[-latex] (v2.70-90) .. controls +(60-90:4em) and +(120-90:4em) .. node[pos=.5,fill=white,right,xshift=.5ex](){\footnotesize$c$} (v2.110-90);

 	\draw[-latex] (v1.70-90) .. controls +(60-90:4em) and +(120-90:4em) .. node[pos=.5,fill=white,right,xshift=.5ex](){\footnotesize$b$} (v1.110-90);


 	\end{tikzpicture}
 	}
\caption{A nondeterministic system in which no computation-based strategy can be defined, but where both run-based and trace-based strategies can be defined.}%
\label{fig:nondetex3}
\end{figure}
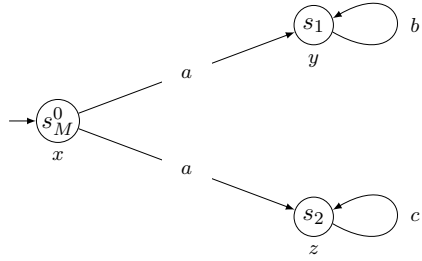

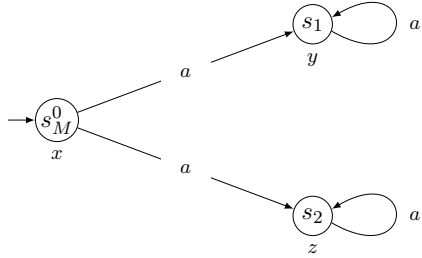
\begin{figure}
\scalebox{.85}{
 	  \begin{tikzpicture}[scale=1]

 	  \tikzstyle{every ellipse node}=[draw,inner xsep=3.5em,inner ysep=1.2em,fill=black!15!white,draw=black!15!white]
 	  \tikzstyle{every circle node}=[fill=white,draw,minimum size=1.6em,inner sep=0pt]


	  \draw(-.9,0)  node(0){};
	  \draw(0,0)	node[label=-90:{\footnotesize$x$},circle](v0){$s_M^0$};
	  \draw(4,1.5)	node[label=-90:{\footnotesize$y$},circle](v1){$s_1$};
	  \draw(4,-1.5)	node[label=-90:{\footnotesize$z$},circle](v2){$s_2$};



	\draw[-latex] (0) -- (v0);
 	\draw[-latex] (v0) --node[pos=.5,fill=white](){\footnotesize$\begin{array}{l}a\end{array}$} (v1);
 	\draw[-latex] (v0) --node[pos=.5,fill=white](){\footnotesize$\begin{array}{l}a\end{array}$} (v2);

 	%

 	\draw[-latex] (v2.70-90) .. controls +(60-90:4em) and +(120-90:4em) .. node[pos=.5,fill=white,right,xshift=.5ex](){\footnotesize$a$} (v2.110-90);

 	\draw[-latex] (v1.70-90) .. controls +(60-90:4em) and +(120-90:4em) .. node[pos=.5,fill=white,right,xshift=.5ex](){\footnotesize$a$} (v1.110-90);


 	\end{tikzpicture}
 	}
\caption{A system in which any computation-based strategy induces a set of runs and a set of traces containing, respectively, runs $s^0_{M}s_1^\omega$ and $s^0_{M}s_2^\omega$ and traces $xy^\omega$ and $xz^\omega$.}%
\label{fig:nondetex4}
\end{figure}

Finally, the reader may have noticed that in this section we did not study the case considering run-based preferences (for run-based strategies we know that the negative results for deterministic systems carry over). The reason is that, as shown for deterministic systems, we can ensure invariance of (the existence of) Nash equilibria with respect to bisimilarity only if the sets of run-based preferences are congruent between bisimilar systems. As this is regarded as a major drawback, even for deterministic systems as illustrated in the example in Figure~\ref{fig:runprefex1}, it is really not an interesting question to be investigated any further.


\begin{figure}[t]
\scalebox{.85}{
 	  \begin{tikzpicture}[scale=1]

 	  \tikzstyle{every ellipse node}=[draw,inner xsep=3.5em,inner ysep=1.2em,fill=black!15!white,draw=black!15!white]
 	  \tikzstyle{every circle node}=[fill=white,draw,minimum size=1.6em,inner sep=0pt]

 	  \draw[use as bounding box,draw opacity=0] (-1,-3) rectangle (7.5,3);

	  \draw(-.9,0)  node(0){};
	  \draw(0,0)	node[label=-90:{\footnotesize$x$},circle](v0){$s_M^0$};
	  \draw(4,1.5)	node[label=-90:{\footnotesize$x$},circle](v1){$s_1$};
	  \draw(4,-1.5)	node[label=-90:{\footnotesize$x$},circle](v2){$s_2$};



	\draw[-latex] (0) -- (v0);
 	\draw[-latex] (v0) --node[pos=.5,fill=white](){\footnotesize$\begin{array}{l}a,a\\b,b\end{array}$} (v1);
 	\draw[-latex] (v0) --node[pos=.5,fill=white](){\footnotesize$\begin{array}{l}a,b\\b,a\end{array}$} (v2);

 	%

 	\draw[-latex] (v2.70-90) .. controls +(60-90:4em) and +(120-90:4em) .. node[pos=.5,fill=white,right,xshift=.5ex](){\footnotesize$\ast$} (v2.110-90);

 	\draw[-latex] (v1.70-90) .. controls +(60-90:4em) and +(120-90:4em) .. node[pos=.5,fill=white,right,xshift=.5ex](){\footnotesize$\ast$} (v1.110-90);


 	\end{tikzpicture}
 	}
	\scalebox{.8}{
	 	  \begin{tikzpicture}[scale=1]

	 	  \tikzstyle{every ellipse node}=[draw,inner xsep=3.5em,inner ysep=1.2em,fill=black!15!white,draw=black!15!white]
	 	  \tikzstyle{every circle node}=[fill=white,draw,minimum size=1.6em,inner sep=0pt]

	 	  \draw[use as bounding box,draw opacity=0] (-2,-3) rectangle (6.5,3);

		  \draw(-.9,0)  node(0){};
		  \draw(0,0)	node[label=-90:{\footnotesize$x$},circle](v0){$s_{M'}^0$};
		  \draw(4,0)	node[label=-90:{\footnotesize$x$},circle](v1){$s_1$};



		\draw[-latex] (0) -- (v0);
	 	\draw[-latex] (v0) --node[pos=.5,fill=white](){\footnotesize$\begin{array}{l}\ast\end{array}$} (v1);

	 	%


	 	\draw[-latex] (v1.70-90) .. controls +(60-90:4em) and +(120-90:4em) .. node[pos=.5,fill=white,right,xshift=.5ex](){\footnotesize$\ast$} (v1.110-90);


	 	\end{tikzpicture}
	 	}

\caption{A pair of bisimilar systems, $M$ and $M'$, where the sets of run-based preferences given by~$\Gamma_1 = \{s^0_M,s_1,s_1,\ldots\}$ for player~1 and $\Gamma_2 = \{s^0_M,s_2,s_2,\ldots\}$ for player~2 in system~$M$, do not have a congruent counterpart in system~$M'$.}%
\label{fig:runprefex1}
\end{figure}
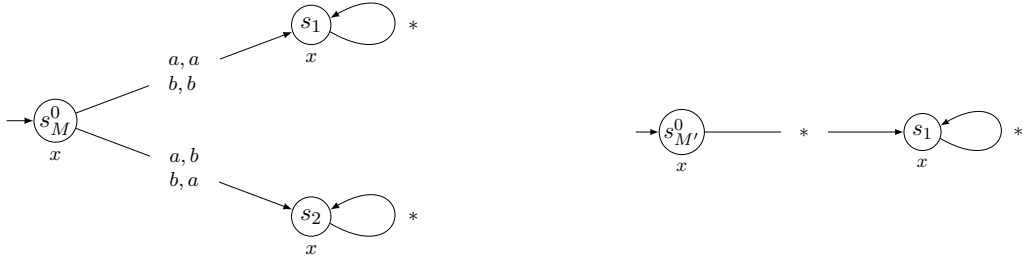

\section{Strategy Logics: New Semantic Foundations}%
\label{secn:logic}

Several logics for strategic reasoning have been proposed in the
literature of computer science and AI, such as
ATL$^*$~\cite{AHK02}, Strategy Logic~\cite{MMPV14,CHP10},
Coalition Logic~\cite{P02}, Coordination
Logic~\cite{FS10}, Game Logic~\cite{PP03a}, and
Equilibrium Logic~\cite{GHW17}. In several cases, the model of
strategies that is used is the one that we refer to as run-based in
this paper, that is, strategies are functions from finite sequences of
states (of some arena) to actions/decisions/choices of players in a
given game. As can be seen from our results so far, of the four
options we have explored, run-based strategies form the least
desirable model of strategies from a semantic point of view since in
such a case Nash equilibrium is not preserved under bisimilarity.

This does not necessarily immediately imply that a particular logic
with a run-based strategy model is not invariant under
bisimilarity. For instance, ATL$^*$ is a bisimulation-invariant logic
and, as shown in~\cite{GHW15-concur} one can reason about Nash
equilibrium using ATL$^*$ only up-to bisimilarity. A question then
remains: whether any of these logics for strategic reasoning becomes
invariant under bisimilarity---as explained before, a desirable
property---if one changes the model of strategies considered there to,
for instance, computation-based or trace-based strategies. We find
that this question has a satisfactory positive answer in some cases.
	In particular, we will consider the above question in the context of Strategy Logic as studied in~\cite{MMPV14}, and in doing so we will provide new semantic foundations for strategy logics.

	Let us start by introducing the syntax and semantics under the
	run-based
	 model of strategies for Strategy Logic (SL~\cite{MMPV14}) as it has been given in~\cite{MMPV16}.
	Syntactically, SL extends LTL with two \emph{strategy quantifiers}, $\EExs{x}$ and $\AAll{x}$, and an \emph{agent binding} operator $(i, x)$, where $i$ is an agent and $x$ is a variable.
	Intuitively, these operators can be understood as \emph{``there exists a strategy $x$''}, \emph{``for all strategies $x$''}, and \emph{``bind agent $i$ to the strategy associated with the variable $x$''}, respectively.
	Formally, SL formulae are inductively built from a set of atomic propositions $\AP$, variables $\VarSet$, and agents $\Ag$, using the following grammar, where $p \in \AP$, $x \in \VarSet$, and $i \in \Ag$:

	\begin{center}
		$\phi ::= p \mid \neg \phi \mid \phi \wedge \phi \mid \X \phi \mid \phi \U \phi \mid \EExs{x} \phi \mid \AAll{x} \phi \mid (i, x) \phi$.
	\end{center}
	We also use the usual abbreviations for LTL formulae, that is, those
  for Boolean and temporal logic formulae.

	We can now present the semantics of SL formulae.
	Given a concurrent game structure $M$, for all SL formulae $\phi$,
  states $s \in \States$ in $M$, and assignments $\asgFun \in
  \AsgSet = (\VarSet \cup \AgSet) \to \StrSet$, mapping variables and
  agents to strategies, the relation $M, \asgFun, s \models \phi$ is
  defined as follows:
	\begin{enumerate}
		\item\label{def:sl(semantics:ap)}
			$M, \asgFun, s \models p$ if $p \in \labFun(s)$,
			with $p \in \APSet$.
		\item\label{def:sl(semantics:bool)}
			For all formulae $\phi$, $\phi_{1}$, and $\phi_{2}$, we have:
			\begin{enumerate}
				\item\label{def:sl(semantics:neg)}
					$M, \asgFun, s \models \neg \phi$ if not $M,
					\asgFun, s \models \phi$;
				\item\label{def:sl(semantics:conj)}
					\scalebox{.92}[1]{$M, \asgFun, s \models \phi_{1} \wedge \phi_{2}$ if
					$M, \asgFun, s \models \phi_{1}$ and $M, \asgFun,
					s \models \phi_{2}$.} 
			\end{enumerate}
		\item\label{def:sl(semantics:qnt)}
			For all formulae~$\varphi$ and variables $x \in \VarSet$ we have:
			\begin{enumerate}
				\item\label{def:sl(semantics:eqnt)}
					$M, \asgFun, s \models \EExs{x} \phi$ if there
					is a strategy $\strFun \in \StrSet$ such that $M,
					\asgFun[x \mapsto \strFun], s \models \phi$;
				\item\label{def:sl(semantics:aqnt)}
					$M, \asgFun, s \models \AAll{x} \varphi$ if for all
					strategies $\strFun \in \StrSet$ we have that $M,
					\asgFun[x \mapsto \strFun], s \models \phi$.
			\end{enumerate}
		\item\label{def:sl(semantics:bnd)}
			For all $i \in \AgSet$ and $x \in \VarSet$, we have
			$M, \asgFun, s \models
			(i, x) \phi$ if $M, \asgFun[i \mapsto
			\asgFun(x)], s \models \phi$.
		\item\label{def:sl(semantics:path)}
			Moreover, for all formulas $\phi$, $\phi_{1}$, and $\phi_{2}$, we
			have:
			\begin{enumerate}
				\item\label{def:sl(semantics:next)}
					$M, \asgFun, s \models \X \phi$ if $M, {(\asgFun, s)}^{1}, \transf(s, \direction) \models \phi$, where $\direction$ is the decision taken from $s$ by following $\asgFun$ and ${(\asgFun, s)}^{1}$ is the update of the assignment function as described in~\cite{MMPV14};
				\item\label{def:sl(semantics:until)}
					$M, \asgFun, s \models \phi_{1} \U \phi_{2}$ if
					there exist $k \in  \SetN$ such that $M,
					{(\asgFun, s)}^{k}, \transf(s, \vec{\direction}) \models \phi_{2}$ and, for all $h \in \SetN$ with $h \leq k$, we have $M, {(\asgFun, s)}^{h}, \transf(s, \vec{\direction}_{\leq h}) \models \phi_{1}$, where $\vec{\direction}$ is the sequence of decisions identified by the assignment function $\asgFun$ starting from $s$, and ${(\asgFun, s)}^{k}$ is the update of the assignment given by the execution of $k$ steps of the strategy profile in $\asgFun$ starting from $s$.
			\end{enumerate}
	\end{enumerate}

    \noindent
	Intuitively, rules~\ref{def:sl(semantics:eqnt)}
	and~\ref{def:sl(semantics:aqnt)}, respectively, are used to interpret the existential
	$\EExs{x}$ and universal $\AAll{x}$ quantifiers over strategies, and
	rule~\ref{def:sl(semantics:bnd)} is used to bind an agent to the strategies associated
	with variable~$x$. All other rules are as in LTL over concurrent game structures.

	As can be seen from its semantics, SL can be interpreted under
  different models of strategies and goals.  As it was originally formulated, SL
  considers run-based strategies and trace-based preferences/goals.
  More specifically, the model of goals is a proper subset of the
  trace-based one, represented by LTL goals over the set $\APSet$ of
  variables.
	In SL, it is possible to represent the existence of a Nash equilibrium in a
  concurrent game structure~\cite{MMPV14}.  This implies, given Theorem~\ref{thm:variance},
  that SL under the standard interpretation is not invariant under
  bisimulation, as the formula expressing the existence of a Nash equilibrium can
  distinguish between two bisimilar models.

  Given the semantics of SL formulae given above,
  we now consider SL under the model of computation-based strategies, and find that
  in such a case SL becomes invariant under bisimilarity.
  %
  Formally, we have the following result.

	\begin{thm}%
		\label{thm:slcompbisin}
		Let $M_{1} = (\Ag, \AP, \Ac, \States_{1}, s_{1}^{0}, \labFun_{1}, \transf_{1})$ and $M_{2} = (\Ag, \AP, \Ac, \States_{2}, s^{0}_{2}, \labFun_{2}, \transf_{2})$ be two bisimilar CGSs.
		Moreover, let $\asgFun$ be an assignment of strategies and $s_{1} \sim s_{2} $ be two bisimilar states.
		Then, for all $\varphi \in SL$, it holds that
		\[
		M_{1}, \asgFun,  s_{1}, \models \varphi \qquad \text{ if and only if } \qquad M_{2}, \asgFun, s_{2} \models \varphi.
		\]
	\end{thm}
	\begin{proof}
		The proof proceeds by induction on the structure of $\varphi$.
		First note that we do not need to prove all the cases, as, for example, we have that $\psi_{1} \vee \psi_{2} = \neg (\neg \psi_{1} \wedge \neg \psi_{2})$ and $\AAll{x} \varphi = \neg \EExs{x} \neg \varphi$.
		Moreover, recall from Lemma~\ref{lemma:strategy_identity} that every computation-based strategy in $M_{1}$ is also a computation-based strategy in $M_{2}$ and vice-versa.
		We have the following.

		\begin{itemize}
			\item $\varphi = p \in \AP$.
			We have that $M_{1}, \asgFun, s_{1} \models p$ if and only if $p \in \labFun_{1}(s_{1}) = \labFun_{2}(s_{2})$ if and only if $M_{2}, \asgFun, s_{2} \models p$;

			\item $\varphi = \neg \psi$.
			We have that $M_{1}, \asgFun, s_{1} \models \neg \psi$ if and only if $M_{1}, \asgFun, s_{1} \not\models \psi$ if and only if, by induction hypothesis, $M_{2}, \asgFun, s_{2} \not\models \psi$ if and only if $M_{2}, \asgFun, s_{2} \models \neg \psi$.

			\item $\varphi = \psi_{1} \wedge \psi_{2}$.
			We have that $M_{1}, \asgFun, s_{1} \models \psi_{1} \wedge \psi_{2}$ if and only if $M_{1}, \asgFun, s_{1} \models \psi_{1}$ and  $M_{1}, \asgFun, s_{1} \models \psi_{2}$, which holds, by induction hypothesis, if and only if $M_{2}, \asgFun, s_{2} \models \psi_{1}$ and  $M_{2}, \asgFun, s_{2} \models \psi_{2}$ if and only if $M_{2}, \asgFun, s_{2} \models \psi_{1} \wedge \psi_{2}$.

			\item $\varphi = \X \psi$.
			We have that $M_{1}, \asgFun, s_{1} \models \X \psi$ if and only if $M_{1}, {(\asgFun)}^{1}, \transf_{1}(s_{1}, \direction) \models \psi$, where $\direction = (\asgFun(1)(\epsilon), \ldots \asgFun(n)(\epsilon))$ is the decision taken by the agents on the first round of the game, according the assignment $\asgFun$. By bisimilarity, we have that $\transf_{1}(s_{1}, \direction) \sim \transf_{2}(s_{2}, \direction)$, and so, by induction hypothesis, that $M_{2}, {(\asgFun)}^{1}, \transf_{2}(s_{2}, \direction) \models \psi$, that holds if and only if $M_{2}, \asgFun, s_{2} \models \X \psi$.

			\item $\varphi = \varphi_{1} \U \varphi_{2}$.
			We have that $M_{1}, \asgFun, s_{1} \models \varphi_{1} \U \varphi_{2}$ if and only if there exists $k \in \SetN$ such that $M_{1}, \asgFun^{k}, \transf_{1}^{*}(s_{1}, \vec{\direction}) \models \varphi_{2}$ and $M_{1}, \asgFun^{h}, \transf_{1}^{*}(s_{1}, \vec{\direction}_{\leq h}) \models \varphi_{1}$ for every $h < k$, where $\vec{\direction}$ is the unique sequence of decisions identified by the $k$-steps application of the transition function that follows $\asgFun$.
			Observe that, for each $h \leq k$, we have that $\transf_{1}(s_{1}, \vec{\direction}_{\leq h}) \sim \transf_{2}(s_{2}, \vec{\direction}_{\leq h})$ and so, by induction hypothesis, we have that $M_{2}, \asgFun^{k}, \transf_{2}^{*}(s_{2}, \vec{\direction}) \models \varphi_{2}$ and $M_{2}, \asgFun^{h}, \transf_{2}^{*}(s_{1}, \vec{\direction}_{\leq h}) \models \varphi_{1}$ for every $h < k$, that is, if and only if $M_{2}, \asgFun, s_{2} \models \varphi_{1} \U \varphi_{2}$.

			\item $\varphi = (i, x) \psi$.
			We have that $M_{1}, \asgFun, s_{1} \models (i, x) \psi$ if and only if $M_{1}, \asgFun[i \mapsto \asgFun(x)], s_{1} \models \psi$ if and only if, by induction hypothesis, $M_{1}, \asgFun[i \mapsto \asgFun(x)], s_{2} \models \psi$ if and only if $M_{2}, \asgFun, s_{2} \models (i, x) \psi$.

			\item $\varphi = \EExs{x} \psi$.
			We have that $M_{1}, \asgFun, s_{1} \models \EExs{x} \psi$ if and only if there exists a strategy $\strFun \in \StrSet$ such that $M_{1}, \asgFun[x \mapsto \strFun], s_{1} \models \psi$ if and only if, by induction hypothesis $M_{2}, \asgFun[x \mapsto \strFun], s_{2} \models \psi$, if and only if $M_{2}, \asgFun, s_{2} \models \EExs{x} \psi$.

		\end{itemize}
		This concludes the proof.
	\end{proof}

	As an immediate corollary, we then obtain the following result about the semantic relationship between the properties that can be expressed in SL and the concept of bisimilarity.

	\begin{cor}%
		\label{cor:slcompbisinvbis}
		SL with the computation-based model of strategies is invariant
    under bisimilarity.
	\end{cor}

Finally, an analogous statement to the above Corollary can also be proved if we consider the model of trace-based strategies,
leading to the next result on the semantics of SL\@.

	\begin{cor}%
		\label{cor:sltrabisinvbis}
		SL with the trace-based model of strategies is invariant
    under bisimilarity.
	\end{cor}

\section{Concluding Remarks and Related Work}\label{secn:conc}
In this paper we showed that with the conventional model of strategies
used in the logic, computer science, and AI literatures, the existence
of Nash equilibria is not necessarily preserved under bisimilarity---in
particular this is the case for multi-player games played over
deterministic concurrent games structures.  By way of some examples, we also
illustrated some of the implications of this result---for example, in
the context of automated formal verification. To resolve this
difficulty, we furthermore investigated alternative models of
strategies which exhibit some desirable properties, in particular,
{\em allowing for a formalisation of Nash equilibrium that is
  invariant under bisimilarity}, even on nondeterministic systems.

We studied applications of these models and found that through their
use, not only Nash equilibria become invariant under bisimilarity, but
also full logics such as Strategy Logic.  This renders it possible to
combine commonly used optimisation techniques for model checking with
decision procedures for the analysis of Nash equilibria, thus
overcoming a critical problem of this kind of logics regarding
practical applications via automated verification.
Some work also in the intersection between bisimulation equivalences,
concurrent game structures, Nash equilibria, and automated formal
verification is summarised next.

\subsection*{Logics for Strategic Reasoning}
There is now a large literature on logics for strategic reasoning.
From this literature, ATL$^*$~\cite{AHK02} and SL~\cite{MMPV14} stand
out, both due to their use within a number of practical tools for
automated verification, and because of their expressive power. On the
one hand, ATL$^*$ is known to be invariant under bisimilarity using
the conventional model of strategies. As such, Nash equilibria can be
expressed within ATL$^*$ only up to
bisimilarity~\cite{GHW15-concur}. On the other hand, SL, which is
strictly more expressive than ATL$^*$, allows for a simple
specification of Nash equilibria, but suffers from not being invariant
under bisimilarity with respect to the conventional model of
strategies. In this paper, we have put forward a number of solutions
to this problem. An additional advantage of replacing the model of
strategies for SL (and therefore for concurrent game structures) is
that other solution concepts in game theory also become invariant
under bisimilarity. For instance, subgame-perfect Nash equilibria and
strong Nash equilibria---which are widely used when considering,
respectively, dynamic behaviour and cooperative behaviour in
multi-agent systems---can also be expressed in SL\@.  Our results
therefore imply that these concepts are also invariant under
bisimilarity, when considering games over concurrent game structures
and goals given by LTL formulae (which correspond to preferences over
traces).

\subsection*{Bisimulation Equivalences for Multi-Agent Systems}
Even though bisimilarity is probably the most widely used behavioural equivalence in concurrency, in the context of multi-agent systems other relations may be preferred, for instance, equivalence relations that take a detailed account of the independent interactions and behaviour of individual components in a multi-agent system.
In such a setting, ``alternating'' relations with natural ATL$^*$ characterisations have been studied~\cite{AHKV98}.
Our results also apply to such alternating equivalence relations.
Alternating bisimulation is very similar to bisimilarity on labelled transition systems~\cite{Milner80,HM85}, only that when defined on concurrent game structures, instead of action profiles taken as possible transitions, one allows individual player's actions, which must be matched in the bisimulation game. Because of this, it immediately follows that any alternating bisimulation as defined in~\cite{AHKV98} is also a bisimilarity as defined here. Despite having a different formal definition, a simple observation can be made: that the counter-example shown in Figures~\ref{fig:cgswithne} and~\ref{fig:cgswithoutne} also apply to such alternating (bisimulation) relations. This immediately implies that Nash equilibria are not preserved by the alternating (bisimulation) equivalence relations in~\cite{AHKV98} either.
Nevertheless, as discussed in~\cite{Benthem02}, the ``right'' notion of equivalence for games and their game theoretic solution concepts is, undoubtedly, an important and interesting topic of debate, which deserves to be investigated further.

\subsection*{Computations vs.\ Traces}
An important remark about the difference between computations and traces is that even though Nash equilibria and their existence are preserved under bisimilarity by three of the four strategy models we have studied, it is not the case that with each strategy model we obtain the same set of Nash equilibria in a given system, or that we can sustain the same set of computations or traces.
For instance, consider again the games in Figures~\ref{fig:cgswithne} and~\ref{fig:cgswithoutne}.
As we discussed above, if we consider the model of computation-based strategies and LTL goals ({\em i.e.}, trace-based goals) as shown in the example, then we obtain two  games, each with an associated non-empty set of Nash
equilibria, which are preserved by bisimilarity.
However, if we consider, instead, the model of trace-based strategies and the same LTL goals, then we obtain two concurrent games both with empty sets of Nash equilibria---thus, in this case, the non-existence of Nash equilibria is preserved by bisimilarity! To observe this, note that whereas in the case of computation-based strategies player~$3$ can implement a uniform ``punishment'' strategy for both player~$1$ and player~$2$, in the case of trace-based strategies player~$3$ cannot do so, even in the game in Figure~\ref{fig:cgswithne}.

\subsection*{Two-Player Games with Trace-Based Goals}
We also showed that if we consider two-player games together with the conventional model of strategies, the problems that arise with respect to the preservation of Nash equilibria disappear.
This is indeed an important finding since most verification games ({\em e.g.}, model and module checking, synthesis, etc.) can be phrased in terms of zero-sum two-player games together with temporal logic specifications ({\em e.g.}, using LTL, CTL, or ATL$^*$).
Our results, then, provide conclusive proof that, if only two-player games and temporal logic goals are needed, then all equilibrium analyses can be carried out using the conventional model of strategies---along with their associated reasoning tools and formal verification techniques.

\subsection*{Nondeterminism}
We extended our main bisimulation-invariant results to nondeterministic systems, making it possible to analyse more complex systems.
This was possible, in turn, because our two main models of strategies, namely computation-based and trace-based, are themselves oblivious to nondeterministic choices.
As a consequence, given a particular strategy (or strategy profile, more generally), the set of outcomes of a multi-player game across bisimilar structures remains the same.
Indeed, the definitions of strategies in the computation-based and trace-based models can be used to show that the set of Nash equilibria in strategy profiles given by these two models is invariant across systems that are equivalent with respect to equivalences for concurrency that are weaker than bisimilarity; for instance, across trace equivalent systems as defined in CSP~\cite{BrookesHR84}.
Thus, with respect to this kind of systems, all our positive results also carry over, even for nondeterministic processes.

\subsection*{Tools for Model Checking and Equilibrium Analysis}
Due to the success of temporal logics and model checking in the verification of concurrent and multi-agent systems, some model checking tools have been extended to cope with the strategic analysis of concurrent systems modelled as multi-player games. For instance, tools such as MCMAS~\cite{CLMM14}, EAGLE~\cite{TGW15}, PRALINE~\cite{Brenguier13}, MOCHA~\cite{AHMQRT98}, and PRISM~\cite{KPW16}, allow for the analysis of \emph{some} strategic properties in a system.
Because all of these tools rely on underlying algorithms for temporal logic model checking, hardly any optimisations are possible when moving to the more complex game-theoretic setting where Nash equilibria needs to be analysed.
In this way, our results find a powerful, and immediate, practical application. Indeed, based on the work presented in this paper, we have developed a new tool for temporal equilibrium analysis~\cite{GNPW18}, which uses the computation-based model of strategies studied here.

As mentioned before, we have developed a new tool for temporal equilibrium analysis, which we call~{EVE}~\cite{GNPW18} (Equilibrium Verification Environment). EVE uses the computation-based model of strategies and trace-based preferences given by LTL formulae. EVE is a formal verification tool for the automated analysis of temporal equilibrium properties of concurrent and multi-agent systems modelled using the {Simple Reactive Module Language} (SRML~\cite{AH99b,vanderhoek:2006c}) as a collection of independent system components (players/agents in a game). In particular, EVE automatically solves three key decision problems in rational synthesis and verification~\cite{GutierrezHW17,WGHMPT16,FismanKL10}: \textsc{Non-Emptiness}, \textsc{E-Nash}, and \textsc{A-Nash}. These problems ask, respectively, whether a multi-player game has at least one Nash equilibrium, whether an LTL formula holds on \emph{some} Nash equilibrium, and whether an LTL formula holds on \emph{all} Nash equilibria. EVE uses a technique based on parity games to check for the existence of Nash equilibria in a concurrent and multi-player game, which crucially relies on the underlying model of strategies being bisimulation invariant.



%
%

\section*{Acknowledgment}
This paper is a revised and extended version of~\cite{GutierrezHPW17}. All authors acknowledge with gratitude the financial support of ERC Advanced Investigator Grant 291528 (``RACE'') at the University of Oxford.
Paul Harrenstein was also supported in part by ERC Starting Grant 639945 (``ACCORD'') also at the University of Oxford. Michael Wooldridge and Paul Harrenstein furthermore acknowledge the financial support of the Alan Turing Institute in London.
Giuseppe Perelli was also supported in part by the ERC Consolidator Grant 772459 (``DSynMA'').
We also thank Johan van Benthem, the reviewers of CONCUR 2017, and the participants of Dagstuhl seminar 17111 (``Game Theory in AI, Logic, and Algorithms'') for their comments and helpful discussions. Finally, we would also like to thank the reviewers of \emph{Logical Methods in Computer Science} for their detailed and thoughtful comments.

\bibliographystyle{alpha}
\bibliography{References}

\end{document}